\documentclass{article}
\usepackage{multicol}
\usepackage{geometry}
\newgeometry{left=1.9cm,right=1.9cm,top=1.9cm,bottom=2.54cm}

\usepackage{amssymb}
\usepackage{amsmath}
\usepackage{amsthm}
\usepackage{stmaryrd}
\usepackage{graphicx}

\usepackage{amssymb, amscd,txfonts,stmaryrd,setspace,mathtools,makecell}
\usepackage{enumerate}
\usepackage[usenames,dvipsnames]{xcolor}
\usepackage[bookmarks=true,colorlinks=true, linkcolor=MidnightBlue, citecolor=cyan]{hyperref}
\usepackage{lmodern}
\usepackage{graphicx,float}

\input xy
\xyoption{all} \xyoption{poly} \xyoption{knot}\xyoption{curve}
\usepackage{xypic,color}
\usepackage{tikz}

\newcommand{\comment}[1]{}

\def\tn{\textnormal}

\def\mc{\mathcal}

\def\NN{{\mathbb N}}

\def\Hom{\tn{Hom}}
\def\Mor{\tn{Mor}}
\def\Fun{\tn{Fun}}
\def\RQry{\tn{RQry}}

\def\Ob{\tn{Ob}}
\def\Arr{\tn{Arr}}

\def\SEL*{\tn{SEL*}}

\def\hsp{\hspace{.3in}}

\def\singleton{{\{*\}}}

\def\queryto{\leadsto}
\newcommand{\church}[1]{\llbracket #1\rrbracket}
\def\Loop{{\mcL oop}}
\def\LoopSchema{{\parbox{.5in}{\fbox{\xymatrix{\LMO{s}\ar@(l,u)[]^f}}}}}

\def\to{\rightarrow}

\def\down{\downarrrow}

\def\cross{\times}
\def\taking{\colon}

\def\too{\longrightarrow}

\def\ss{\subseteq}

\def\iso{\cong}
\def\down{\downarrow}
\def\|{{\;|\;}}
\def\m1{{-1}}

\def\ol{\overline}

\def\vect{\overrightarrow}

\newcommand{\LMO}[1]{\stackrel{#1}{\bullet}}
\newcommand{\LTOO}[1]{\stackrel{\sf{#1}}{\circ}}
\newcommand{\LTO}[1]{\stackrel{\sf{#1}}{\bullet}}

\def\ullimit{\ar@{}[rd]|(.3)*+{\lrcorner}}
\def\urlimit{\ar@{}[ld]|(.3)*+{\llcorner}}
\def\lllimit{\ar@{}[ru]|(.3)*+{\urcorner}}
\def\lrlimit{\ar@{}[lu]|(.3)*+{\ulcorner}}
\def\ulhlimit{\ar@{}[rd]|(.3)*+{\diamond}}
\def\urhlimit{\ar@{}[ld]|(.3)*+{\diamond}}
\def\llhlimit{\ar@{}[ru]|(.3)*+{\diamond}}
\def\lrhlimit{\ar@{}[lu]|(.3)*+{\diamond}}
\newcommand{\clabel}[1]{\ar@{}[rd]|(.5)*+{#1}}
\newcommand{\TriRight}[7]{\xymatrix{#1\ar[dr]_{#2}\ar[rr]^{#3}&&#4\ar[dl]^{#5}\\&#6\ar@{}[u] |{\Longrightarrow}\ar@{}[u]|>>>>{#7}}}
\newcommand{\TriLeft}[7]{\xymatrix{#1\ar[dr]_{#2}\ar[rr]^{#3}&&#4\ar[dl]^{#5}\\&#6\ar@{}[u] |{\Longleftarrow}\ar@{}[u]|>>>>{#7}}}
\newcommand{\TriIso}[7]{\xymatrix{#1\ar[dr]_{#2}\ar[rr]^{#3}&&#4\ar[dl]^{#5}\\&#6\ar@{}[u] |{\Longleftrightarrow}\ar@{}[u]|>>>>{#7}}}

\newcommand{\arr}[1]{\ar@<.5ex>[#1]\ar@<-.5ex>[#1]}
\newcommand{\arrr}[1]{\ar@<.7ex>[#1]\ar@<0ex>[#1]\ar@<-.7ex>[#1]}
\newcommand{\arrrr}[1]{\ar@<.9ex>[#1]\ar@<.3ex>[#1]\ar@<-.3ex>[#1]\ar@<-.9ex>[#1]}
\newcommand{\arrrrr}[1]{\ar@<1ex>[#1]\ar@<.5ex>[#1]\ar[#1]\ar@<-.5ex>[#1]\ar@<-1ex>[#1]}

\newcommand{\To}[1]{\xrightarrow{#1}}
\newcommand{\Too}[1]{\xrightarrow{\ \ #1\ \ }}
\newcommand{\From}[1]{\xleftarrow{#1}}

\newcommand{\Adjoint}[4]{\xymatrix@1{{#2}\ar@<.5ex>[r]^-{#1} &{#3} \ar@<.5ex>[l]^-{#4}}}

\def\id{\tn{id}}

\def\Cat{{\bf Cat}}

\def\Inst{{\bf Inst}}

\def\Set{{\bf Set}}

\def\set{{\text \textendash}{\bf Set}}
\def\inst{{{\text \textendash}\bf \Inst}}

\def\dispInt{\parbox{.1in}{$\int$}}
\def\bhline{\Xhline{2\arrayrulewidth}}
\def\bbhline{\Xhline{2.5\arrayrulewidth}}

\def\mcC{\mc{C}}
\def\mcD{\mc{D}}

\def\mcL{\mc{L}}

\def\mcS{\mc{S}}
\def\mcT{\mc{T}}
\def\mcU{\mc{U}}

\def\TSig{{\bf TSig}}

\newtheorem{thm}{Theorem}
\newtheorem{defn}{Definition}

\newtheorem{theorem}{Theorem}[subsection]
\newtheorem{lemma}[theorem]{Lemma}

\newtheorem{proposition}[theorem]{Proposition}

\newtheorem{corollary}[theorem]{Corollary}

\theoremstyle{remark}
\newtheorem{remark}[theorem]{Remark}
\newtheorem{example}[theorem]{Example}
{
\newtheorem{question}[theorem]{Question}
\newtheorem{guess}[theorem]{Guess}

\theoremstyle{definition}
\newtheorem{definition}[theorem]{Definition}

\newtheorem{construction}[theorem]{Construction}

\setcounter{tocdepth}{2}
\usepackage{diagrams}
\begin{document}

\title{Relational Foundations For Functorial Data Migration
\footnote{Extended version of the DBPL 2015 paper.} }
\author{David I. Spivak, Ryan Wisnesky
\footnote{Work funded by ONR grants N000141010841 and N000141310260.}
\\ Massachusetts Institute of Technology \\ {\sf \{dspivak, wisnesky\}@math.mit.edu} 
}
\date{\today}
%\vspace*{-.1in}
\maketitle
\vspace*{-.4in}
\begin{abstract}
We study the data transformation capabilities associated with schemas that are presented by directed multi-graphs and path equations.  Unlike most approaches which treat graph-based schemas as abbreviations for relational schemas, we treat graph-based schemas as categories.  A schema $S$ is a finitely-presented category, and the collection of all $S$-instances forms a category, $S\inst$.  A functor $F$ between schemas $S$ and $T$, which can be generated from a visual mapping between graphs, induces three adjoint data migration functors, $\Sigma_F\taking S\inst \to T\inst$, $\Pi_F\taking S\inst \to T\inst$, and $\Delta_F\taking T\inst \to S\inst$.  We present an algebraic query language FQL based on these functors, prove that FQL is closed under composition, prove that FQL can be implemented with the select-project-product-union relational algebra (SPCU) extended with a key-generation operation, and prove that SPCU can be implemented with FQL.    
\end{abstract}
%\vspace*{.1in}
\setlength{\columnsep}{.83cm}\begin{multicols}{2}

\section{Introduction}

In this paper we describe how to implement {\it functorial data migration}~\cite{Spivak:2012:FDM:2324905.2325108} using relational algebra, and vice versa.  In the {\it functorial data model}~\cite{Spivak:2012:FDM:2324905.2325108}, a database schema is a finitely-presented category (a directed multigraph with path equations~\cite{BW}), and a database instance on a schema $S$ is a functor from $S$ to the category of sets, $\Set$.  The database instances on a schema $S$ constitute a category, denoted $S\inst$, and a functor $F$ between schemas $S$ and $T$ induces three data migration functors: $\Delta_F\taking T\inst \to S\inst $, defined as $\Delta_F(I) := F \circ I$, and its left and right adjoints $\Sigma_F\taking S\inst \to T\inst$ and $\Pi_F\taking S\inst \to T\inst$, respectively. We define a {\it functorial query language}, FQL, in which every query denotes a data migration functor, prove that FQL is closed under composition, and define a translation from FQL to the select-project-product-union (SPCU) relational algebra extended with an operation for creating globally unique IDs (GUIDs)\footnote{An operation to create IDs is required because functorial data migration can create constants not contained in input instances.}.  In addition, we describe how to implement the SPCU relational algebra with FQL.  

Our primary motivation for this work is to use SQL as a practical deployment platform for the functorial data model.  Using the results of this paper we have implemented a SQL compiler for FQL, and a (partial) translator from SQL to FQL, as part of a visual schema mapping tool available at {\sf categoricaldata.net/fql.html}.  We are interested in the functorial data model because we believe its category-theoretic underpinnings provide a better mathematical foundation for information integration than the relational model.  For example, many practical ``relational'' systems (including SQL-based RDBMSs) are not actually relational, but are based on bags and unique row IDs---like the functorial data model.  Numerous advantages of the functorial data model are described in~\cite{Spivak:2012:FDM:2324905.2325108} and~\cite{Fleming02adatabase}.

\subsection{Related Work}

Although category presentations---directed multi-graphs with path equivalence constraints---are a common notation for schemas~\cite{Buneman92theoreticalaspects}, prior work rarely treats such schemas categorically~\cite{Fleming02adatabase}.  Instead, most prior work treats such schemas as abbreviations for relational schemas.  For example, in Clio~\cite{haas:clio}, users draw lines connecting related elements between two schemas-as-graphs and Clio generates a relational query (semantically similar to our $\Sigma$) that implements the user's intended data transformation.  Behind the scenes, the user's correspondence is translated into a formula in the relational language of second-order tuple generating dependencies, from which a query is generated~\cite{Fagin:2005:CSM:1114244.1114249}.  %As another example, in the Rondo system~\cite{rondo}, users are presented with an ad-hoc collection of operators over schema-as-graphs that they then script together to implement a data transformation.  Although graphs and path-equivalences are used as inputs for both Clio and Rondo, both Clio and Rondo immediately translate from graphs and path-equivalences into a relational schema before proceeding further.  

In many ways, our work is an extension and improvement of Rosebrugh et al's initial work on understanding category presentations as database schemas~\cite{Fleming02adatabase}.  In that work, the authors identify the $\Sigma$ and $\Pi$ data migration functors, but they do not identify $\Delta$ as their adjoint. Moreover, they do not implement $\Sigma$ and $\Pi$ using relational algebra, and they do not formalize a query language or investigate the behavior of $\Sigma$ and $\Pi$ with respect to composition; neither do they consider ``attributes''.  Our mathematical development diverges from their subsequent work on ``sketches''~\cite{Johnson200251}, in which some or all of $\Delta$, $\Sigma$, and  $\Pi$ data migration functors may no longer exist.  

Category-theoretic techniques were instrumental in the development of the nested relational data model~\cite{Wong:1994:QNC:921235}.  However, we do not believe that the functorial data model and the nested relational model are closely connected.  The functorial data model is more closely related to work on categorical logic and type theory, where operations similar to $\Sigma, \Pi$, and $\Delta$ often appear under the slogan that ``quantification is adjoint to substitution''~\cite{opac-b1094856}.

 %With the exception of Alagic and Bernstein's categorical model theory~\cite{DBLP:conf/dbpl/AlagicB01}, little has been done on understanding the data transformation capabilities {\it directly} associated with the graph-based schema model.  
%
\subsection{Contributions and Outline}

We make the following  contributions:

\begin{enumerate}

\item The naive functorial data model~\cite{Fleming02adatabase}, as sketched above, is not quite appropriate for many practical information management tasks.  Intuitively, this is because many categorical constructions are only defined up to isomorphism; at a practical level, this means that every instance must contain only meaningless IDs~\cite{Spivak:2012:FDM:2324905.2325108}.  The idea to extend the functorial data model with ``attributes'' to capture meaningful concrete data such as strings and integers was originally developed in~\cite{Spivak:2012:FDM:2324905.2325108}.  In this paper, we refine that notion of attribute so that our schemas become special kinds of entity-relationship (ER) diagrams (those in ``categorical normal form''), and our instances can be represented as relational tables that conform to such ER diagrams.       (Sections 2 and 3)

\item We define an algebraic query language FQL in which every query denotes a data migration functor in the above extended sense.    We show that FQL queries are closed under composition, and how every query in FQL can be described as a triplet of graph correspondences (similar to schema mappings~\cite{haas:clio}).   Determining whether a triplet of  arbitrary graph correspondences is an FQL query is semi-decidable. (Section 4)

\item We provide a translation of FQL into the SPCU relational algebra of selection, projection, cartesian product, and union extended with a globally unique ID generator that constructs $N+1$-ary tables from $N$-ary tables by adding a globally unique ID to each row.  This allows us to generate SQL programs that implement FQL.  A corollary is that materializing result instances of FQL queries has polynomial time data complexity. (Section 5)
\\ 
\item We show that every SPCU query {\it under bag semantics} is expressible in FQL, and how to extend FQL with an operation allowing it to capture every SPCU query under set semantics.  (Section 6)

\end{enumerate}

For every schema $S$, the category $S\inst$ is a topos~\cite{BW}.  Hence, the functorial data model admits other operations on instances besides the functorial data migration operations that are the focus of this paper.  We conclude with a discussion of which of these operations can be implemented in SPCU+keygen (Section 7).

\section{Categorical Data}\label{sec:FDM}

In this section we define the original signatures and instances of~\cite{Spivak:2012:FDM:2324905.2325108}, as well as ``typed signatures'' and ``typed instances'', which are our extension of~\cite{Spivak:2012:FDM:2324905.2325108} to attributes.  The basic idea is that signatures are stylized ER diagrams that denote categories, and our database instances can be represented as instances of such ER diagrams, and vice versa (up to natural isomorphism).  From this point on, we will distinguish between a {\it signature}, which is a finite presentation of a category, and a {\it schema}, which is the category a signature denotes. 

\subsection{Signatures}
 The functorial data model~\cite{Spivak:2012:FDM:2324905.2325108} uses directed labeled multi-graphs and path equalities for signatures.   A {\it path} $p$ is defined inductively as:
$$
p ::= node \ | \ p.edge 
$$
A {\it signature} $S$ is a finite presentation of a category.  That is, a signature $S$ is a triple $S := (N,E,C)$ where $N$ is a finite set of nodes, $E$ is a finite set of directed edges, and $C$ is a finite set of path equations.  For example:
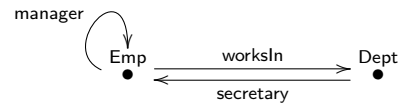
\begin{align*}
\xymatrix@=10pt{&\LTO{Emp}\ar@<.5ex>[rrrrr]^{\sf{worksIn}}\ar@(l,u)[]+<0pt,10pt>^{\sf{manager}}&&&&&\LTO{Dept}\ar@<.5ex>[lllll]^{\sf{secretary}}	}
\end{align*} 
%\mainCatLarge{G :=}
\begin{footnotesize}
\begin{eqnarray*}
&		{\sf Emp}.{\sf manager}.{\sf worksIn} = {\sf Emp}.{\sf worksIn} & \\
&		{\sf Dept}.{\sf secretary}.{\sf worksIn} = {\sf Dept} &
\end{eqnarray*}
\end{footnotesize}
Here we see a signature $S$ with two nodes, three edges, and two path equations.  This information generates a category $\church{S}$: the free category on the graph, modulo the equivalence relation induced by the path equations.  The category $\church{S}$ is the {\it schema} for $S$, and database instances over $\church{S}$ are functors $\church{S}\to\Set$.  Every path $p\taking X \to Y$ in a signature $S$ denotes a morphism $\church{p}\taking \church{X} \to \church{Y}$ in $\church{S}$.  Given two paths $p_1, p_2$ in a signature $S$, we say that they are {\em equivalent}, written $p_1 \cong p_2$, if $\church{p_1}$ and $\church{p_2}$ are the same morphism in $\church{S}$. Two signatures $S$ and $T$ are {\it isomorphic}, written $S \cong T$, if they denote isomorphic schemas, i.e. if the categories they generate are isomorphic.  %We will often write $S$-instance instead of $\llbracket S \rrbracket$-instance.  
%models notation
%We often refer to {\em morphisms in a signature}, by which we mean equivalence classes of paths in the signature. 

\subsection{Cyclic Signatures}

If a signature contains a loop, it may or may not denote a category with infinitely many morphisms.  Hence, some constructions over signatures may not be computable.  Testing if two paths in a signature are equivalent is known as the {\it word problem} for categories. The word problem can be semi-decided using the ``completion without failure'' extension~\cite{Bachmair89completionwithout} of the Knuth-Bendix algorithm.  This algorithm first attempts to construct a strongly normalizing re-write system based on the path equalities; if it succeeds, it yields a linear time decision procedure for the word problem~\cite{xyz}. If a signature denotes a finite category, the Carmody-Walters algorithm~\cite{Carmody1995459} will compute its denotation.  This algorithm computes {\it left Kan extensions} and can be used for many purposes in computational category theory~\cite{Fleming02adatabase}.  %In fact, every $\Sigma$ functor arises as a left Kan extension, and vice versa.

\subsection{Instances}

Let $S$ be a signature.  A $S$-instance is a functor from $\llbracket S \rrbracket$ to the category of sets.  We will represent instances as relational tables using the following binary format:

\begin{itemize}
\item To each node $N$ corresponds an ``identity'' or ``entity'' table named $N$, a reflexive table with tuples of the form $(x,x)$.   We can specify this using first-order logic:
\begin{align}\label{dia:entity FBR}
\forall x y. N(x,y) \Rightarrow x = y.
\end{align}
The entries in these tables are called {\it ID}s or {\it keys}, and for the purposes of this paper we require them to be globally unique.  We call this the {\it globally unique key assumption}.  Note that it is possible to use unary tables instead of binary tables, but we have found the uniform binary format to be simpler when manipulating instances using e.g., SQL.
%It is often convenient to assume that these tuples are globally unique.  This can be specified as an axiom scheme, where for each $N\neq N'$ we have
%\begin{align}\label{dia:guid}
%\forall x. N(x,x) \Rightarrow \neg N^\prime(x,x).
%\end{align}  
\item To each edge $e\taking N_1 \to N_2$ corresponds a ``link'' table $e$ between identity tables $N_1$ and $N_2$. The axioms below say that every edge $e\taking N_1 \to N_2$ designates a total function $N_1 \to N_2$:
\begin{align}\label{dia:link FBR}
\begin{array}{l}
\forall x y. \ e(x,y) \Rightarrow N_1(x,x) \\
\forall x y. \ e(x,y) \Rightarrow N_2(y,y) \\
\forall x y z. \ e(x,y) \wedge e(x,z)  \Rightarrow y = z \\
\forall x. \ N_1(x,x) \Rightarrow \exists  y. e(x,y)
\end{array}
\end{align}
\end{itemize}
An example instance on our employees schema is:
\begin{footnotesize}
$$
\begin{tabular}{| l | l |}\bhline
\multicolumn{2}{| c |}{{\sf Emp}}\\\bbhline 
{\sf  Emp}&{\sf  Emp}\\\hline 
101& 101 \\\hline 
102& 102 \\\hline 
103& 103 \\\bhline
\end{tabular}
\ \ \ \ \ \ 
\begin{tabular}{| l | l |}\bhline
\multicolumn{2}{| c |}{{\sf Dept}}\\\bbhline 
{\sf  Dept}&{\sf  Dept}\\\hline 
q10& q10\\\hline 
x02& x02 \\\bhline 
\end{tabular}
$$
$$
\begin{tabular}{| l | l |}\bhline
\multicolumn{2}{| c |}{{\sf manager}}\\\bbhline 
{\sf  Emp}&{\sf  Emp}\\\hline 
101&103\\\hline 
102&102\\\hline 
103&103 \\\bhline
\end{tabular}
\ \ \ \ \ \ 
\begin{tabular}{| l | l |}\bhline
\multicolumn{2}{| c |}{{\sf worksIn}}\\\bbhline 
{\sf  Emp}&{\sf  Dept}\\\hline 
101&q10\\\hline 
102&q10\\\hline 
103&x02\\\bhline 
\end{tabular}
\ \ \ \ \ \ \
\begin{tabular}{| l | l |}\bhline
\multicolumn{2}{| c |}{{\sf secretary}}\\\bbhline 
{\sf  Dept}&{\sf  Emp}\\\hline 
x02&102\\\hline 
q10&101\\\bhline 
\end{tabular}
$$
\end{footnotesize}

To save space, we will sometimes present instances in a ``joined'' format:

\begin{footnotesize}
$$
\begin{tabular}{| c | c | c |}\bhline
\multicolumn{3}{| c |}{{\sf Emp}}\\\bbhline 
{\sf  Emp}&{\sf  manager }&{\sf  worksIn}\\\hline 
101& 103 & q10 \\\hline 
102& 102 & q10 \\\hline 
103& 103 & x02 \\\bhline
\end{tabular}
\ \ \ \ \ \ 
\begin{tabular}{| l | l |}\bhline
\multicolumn{2}{| c |}{{\sf Dept}}\\\bbhline 
{\sf  Dept}&{\sf  secretary}\\\hline 
q10& 102\\\hline 
x02& 101 \\\bhline 
\end{tabular}
$$
\end{footnotesize}

The natural notion of equivalence of instances is {\it isomorphism}.  In particular, the actual constants in the above tables should be considered meaningless IDs~\cite{Abiteboul:1989:OIQ:67544.66941}.  
\vspace{-.1in}
\subsection{Homomorphisms}
Let $S$ be a signature and $I, J : S \to \Set$ be $S$-instances.  A {\it database homomorphism} $h : I \Rightarrow J$ is a {\it natural transformation} from $I$ to $J$: for every node $N$ in $S$, $h_N$ is a function from the IDs in $I(N)$ to the IDs in $J(N)$ such that for every path $p : X \to Y$ in $S$, we have $h_Y \circ I(p) = J(p) \circ h_X$.  Homomorphisms can be represented as binary tables; for example, the identity homomorphism in our employees example is:

%\vspace{-.2in}
\begin{footnotesize}
%
%\begin{displaymath}
%    \xymatrix{
%        I(X) \ar[r]^{I(p)} \ar[d]_{h_X} & I(Y) \ar[d]^{h_Y} \\
%        J(X) \ar[r]_{J(p)}       & J(Y) }
%        
$$
     \begin{tabular}{| l | l |}\bhline
\multicolumn{2}{| c |}{{\sf Emp}}\\\bbhline 
{\sf  Emp}&{\sf  Emp}\\\hline 
101& 101 \\\hline 
102& 102 \\\hline 
103& 103 \\\bhline
\end{tabular}
\ \ \ \ \ \ 
\begin{tabular}{| l | l |}\bhline
\multicolumn{2}{| c |}{{\sf Dept}}\\\bbhline 
{\sf  Dept}&{\sf  Dept}\\\hline 
q10& q10\\\hline 
x02& x02 \\\bhline 
\end{tabular}
$$
%\end{displaymath}
\end{footnotesize}
\vspace{-.2in}
\subsection{Attributes}
Signatures and instances, as defined above, do not have enough structure to be useful in practice.  At a practical level, we usually need fixed atomic domains like {\sf String} to store actual data.  We cannot simply include a node named, e.g., {\sf String} in a signature because instances must be compared up to isomorphism of IDs but not isomorphism of strings. Hence, in this section we extend the functorial data model by adding {\it attributes}.

Let $S$ be a signature.  A {\it typing} for $S$ is a triple $(A, i, \gamma)$ where $A$ is a discrete category (category containing only objects and identity morphisms) consisting of attribute names,  $i$ is a functor from $A$ to $\llbracket S \rrbracket$ mapping each attribute to its corresponding node, and $\gamma$ is a functor from $A$ to $\Set$, mapping each attribute to its type (e.g., the set of strings, the set of integers).  
$$
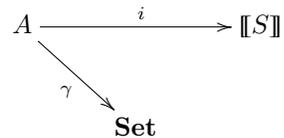
\xymatrix@=25pt{
A\ar[rr]^i\ar[dr]_\gamma&&\llbracket S \rrbracket \\
&\Set
}
$$
Borrowing from ER-diagram notation, we will write attributes as open circles and omit types.  For example, we can enrich our previous signature with a typing (where each attribute has type {\sf String}) as follows:
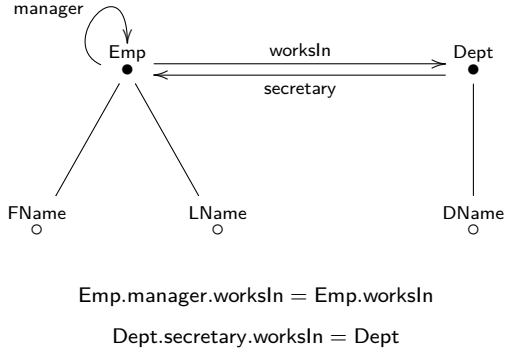
\begin{align*}
\xymatrix@=10pt{&\LTO{Emp}\ar@<.5ex>[rrrrr]^{\sf{worksIn}}\ar@(l,u)[]+<0pt,10pt>^{\sf{manager}}\ar@{-}[dddl]\ar@{-}[dddr]&&&&&\LTO{Dept}\ar@<.5ex>[lllll]^{\sf{secretary}}\ar@{-}[ddd]\\\\\\\LTOO{FName}&&\LTOO{LName}&~&~&~&\LTOO{DName}			}
\end{align*}
%\mainCatLarge{G :=}
\begin{footnotesize}
			$$
		{\sf Emp}.{\sf manager}.{\sf worksIn} = {\sf Emp}.{\sf worksIn}
		$$
		\vspace{-.1in}
		$$
		{\sf Dept}.{\sf secretary}.{\sf worksIn} = {\sf Dept}
		$$
		\end{footnotesize}
		
Importantly, and unlike~\cite{Spivak:2012:FDM:2324905.2325108}, paths may not refer to attributes; they may only refer to nodes and edges.  

\subsection{Typed Instances}
Let $S$ be a signature and $(A,i,\gamma)$ a typing.  A {\it typed instance} $I$ is a pair $(I^\prime,\delta)$ where $I'\taking\mc \llbracket S \rrbracket \to\Set$ is an (untyped) instance together with a natural transformation $\delta\taking I^\prime \circ i\Rightarrow\gamma$.  Intuitively, $\delta$ associates an appropriately typed constant (e.g., a string) to each ID in $I^\prime$: 
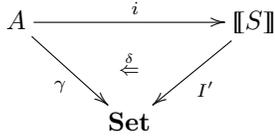
$$\xymatrix@=25pt{
A\ar[rr]^i\ar[dr]_\gamma&\ar@{}[d]|(.4){\stackrel{\delta}{\Leftarrow}}&\llbracket S \rrbracket \ar[dl]^{I^\prime}\\
&\Set
}
$$
  We represent the $\delta$-part of a typed instance $I$ as a set of binary tables as follows:
\begin{itemize}
\item To each node $N$ in $S$ and attribute $a \in i^{-1}(N)$ corresponds a binary attribute table mapping each ID in $I(N)$ to a value in $\gamma(a)$.  
\end{itemize}
In our employees example, we might add the following:

\begin{footnotesize}
$$
\begin{tabular}{| l | l |}\bhline
\multicolumn{2}{| c |}{{\sf FName}}\\\bbhline 
{\sf  Emp}&{\sf  String}\\\hline 
101&Alan\\\hline 
102&Andrey\\\hline 
103&Camille\\\bhline
\end{tabular}
\ \ 
\begin{tabular}{| l | l |}\bhline
\multicolumn{2}{| c |}{{\sf LName}}\\\bbhline 
{\sf  Emp}&{\sf  String}\\\hline 
101&Turing\\\hline 
102&Markov\\\hline 
103&Jordan\\\bhline 
\end{tabular}
\ \ 
\begin{tabular}{| l | l |}\bhline
\multicolumn{2}{| c |}{{\sf DName}}\\\bbhline 
{\sf  Dept}&{\sf  String}\\\hline 
x02&Math\\\hline 
q10&CS\\\bhline 
\end{tabular}
$$
\end{footnotesize}

Typed instances form a category, and two instances are {\it equivalent}, written $\cong$, when they are isomorphic objects in this category.  Isomorphism of typed instances captures our expected notion of equality on typed instances, where the ``structure parts'' are compared for isomorphism and the ``attribute parts'' are compared for equality under such an isomorphism.  This notion of equivalence of instances is the same notion as ``being equivalent up to object identity in semantic data models'' (chapter 11 of~\cite{DBLP:books/aw/AbiteboulHV95}).

\subsection{Typed Homomorphisms}

Database homomorphisms extend to typed instances.  Let $S$ be a typed signature and $I,J$ be $S$-instances.  A {\it database homomorphism} $h \taking I \Rightarrow J$ is still, for each node $N$ in $S$, a function from $I(N)$ to $J(N)$, obeying naturality conditions.  However, we also require that $h$ must preserve attributes, so, for example, if $(i, {\sf Alice}) \in I({\sf Name})$ then $(h(i), {\sf Alice}) \in J({\sf Name})$.  This definition of database homomorphism coincides with the traditional definition of database homomorphism from database theory~\cite{DBLP:books/aw/AbiteboulHV95}, where our IDs play the role of ``labelled nulls''.

\section{Functorial Data Migration}

In this section we define the original signature morphisms and data migration functors of~\cite{Spivak:2012:FDM:2324905.2325108}, as well as ``typed signature morphisms'' and ``typed data migration functors'', which are our extension of~\cite{Spivak:2012:FDM:2324905.2325108} to attributes.  The basic idea is that associated with a morphism between signatures $F \taking S \to T$ is a data transformation $\Delta_F \taking T\inst \to S\inst$ and left and right adjoints $\Sigma_F, \Pi_F : S\inst \to T\inst$.  

\subsection{Signature Morphism}

Let $S$ and $T$ be signatures.  A {\it signature morphism} $F\taking S \to T$ is a mapping that takes nodes in $S$ to nodes in $T$ and edges in $S$ to {\em paths} in $T$; in so doing, it must respect edge sources, edge targets, and path equivalences. In other words, if $p_1 \cong p_2$ is a path equivalence in $S$, then $F(p_1) \cong F(p_2)$ must be a path equivalence in $T$.  Each signature morphism $F\taking S \to T$ determines a unique {\it schema morphism} (indeed, a functor) $\church{F}\taking\church{S}\to\church{T}$ in the obvious way.  Two signature morphisms $F_1\taking S \to T$ and $F_2\taking S \to T$ are {\it equivalent}, written $F_1 \cong F_2$, if they denote naturally isomorphic functors.  Below is an example signature morphism.

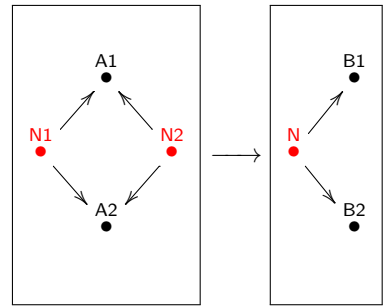
\begin{align*}\parbox{1in}{\fbox{\xymatrix@=10pt{&\\&\LTO{A1}\\\color{red}{\LTO{N1}}\ar[ur]\ar[dr]&&\color{red}{\LTO{N2}}\ar[ul]\ar[dl]\\&\LTO{A2}\\&}}}\Too{}\parbox{.6in}{\fbox{\xymatrix@=10pt{&\\&\LTO{B1}\\\color{red}{\LTO{N}}\ar[ur]\ar[dr]\\&\LTO{B2}\\&}}}\end{align*}

In the above example, the nodes {\sf N1} and {\sf N2} are mapped to {\sf N}, the two morphisms to {\sf A1} are mapped to the morphism to {\sf B1}, and the two morphisms to {\sf A2} are mapped to the morphism to {\sf B2}.

\subsection{Typed Signature Morphisms}
Intuitively, signature morphisms are extended to typed signatures by providing an additional (node and type-respecting) mapping between attributes.  For example, we might have {\sf Name} and {\sf Age} attributes in our typings and map {\sf Name} to {\sf Name} and {\sf Age} to {\sf Age}:
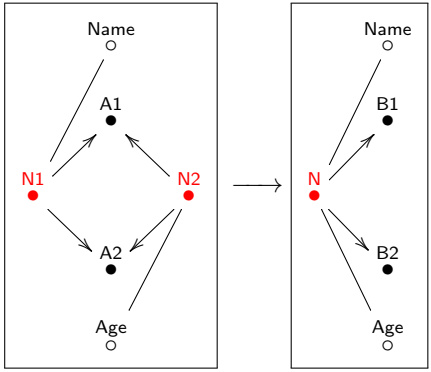
\begin{align*}\parbox{1.15in}{\fbox{\xymatrix@=10pt{&\LTOO{Name}\\&\LTO{A1}\\\color{red}{\LTO{N1}}\ar@{-}[uur]\ar[ur]\ar[dr]&&\color{red}{\LTO{N2}}\ar[ul]\ar[dl]\ar@{-}[ddl]\\&\LTO{A2}\\&\LTOO{Age}}}}\Too{ }\parbox{.8in}{\fbox{\xymatrix@=10pt{&\LTOO{Name}\\&\LTO{B1}\\\color{red}{\LTO{N}}\ar@{-}[uur]\ar[ur]\ar[dr]\ar@{-}[ddr]\\&\LTO{B2}\\&\LTOO{Age}}}}\end{align*}
% .  More complicated mappings of attributes are also possible, as we will see in the next section.
Formally, let $S$ be a signature and $\Gamma := (A, i, \gamma)$ a typing for S.  Let $S'$ be a signature and $\Gamma' := (A', i', \gamma')$ a typing for $S'$.  A typed signature morphism from $(S,\Gamma)$ to $(S',\Gamma')$ consists of a signature morphism $F : S \to S'$ and a functor (i.e. a ``morphism of typings'') $G : A \to A'$ such that the following diagram commutes: 

%Suppose $(\mcC',\mcC_0,'i',\Gamma')$ is another typed signature. A typed signature morphism from $\ol{\mcC}$ to $\ol{\mcC'}$, denoted $(F,F_0)\taking\ol{\mcC}\to\ol{\mcC'}$, consists of a functor $F\taking\mcC\to\mcC'$ and a functor $F_0\taking\mcC_0\to\mcC'_0$ such that the following diagram commutes:
$$
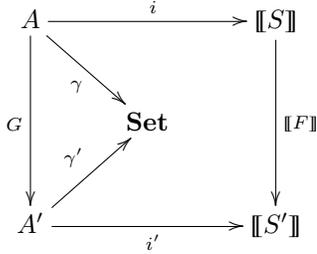
\xymatrix@=25pt{
A\ar[rr]^i\ar[dr]_\gamma\ar[dd]_{G}&&\llbracket S \rrbracket\ar[dd]^{\llbracket F \rrbracket}\\
&\Set\\
A'\ar[ur]^{\gamma'}\ar[rr]_{i'}&&\llbracket S' \rrbracket
}
$$
%\vspace{-.1in}
\subsection{Data Migration Functors}
Each signature morphism $F\taking S \to T$ is associated with three data migration functors, $\Delta_F, \Sigma_F$, and $\Pi_F$, which migrate instance data on $S$ to instance data on $T$ and vice versa.  A summary is in Figure~\ref{summary} on the next page. 
\begin{figure*}
\caption{Data migration functors induced by a signature morphism $F\taking S \to T$}
\label{summary}
\begin{center}\begin{tabular}{| c | c | c | c | c | }\bhline\multicolumn{5}{| c |}{ }\\ \bhline {\bf  Name}&{\bf  Symbol}&{\bf Type}&{\bf Idea of definition}&{\bf Relational partner}\\\bbhline Pullback & $\Delta_F$ & $\Delta_F\taking T\inst\to S\inst$&Composition with $F$&Project\\\hline Right Pushforward&$\Pi_F$&$\Pi_F\taking S \inst\to T \inst$&Right adjoint to $\Delta_F$&Join\\\hline Left Pushforward&$\Sigma_F$&$\Sigma_F\taking S \inst\to T \inst$&Left adjoint to $\Delta_F$&Union\\\bhline\end{tabular}\end{center}
\end{figure*}
%%%%%%%%%%%%%%%%%%%%%%%%%%%%%%
\begin{figure*}
\caption{Data migration example}
\label{ex1}
\vspace{.2in}
\begin{footnotesize}
\begin{minipage}[t]{2.4in}
\begin{align*} 
\parbox{1.0in}{
\fbox{\xymatrix@=8pt{&\LTO{Name}\\&\LTO{Salary}\\\color{red}{\LTO{N1}}\ar[uur]\ar[ur]&&\color{red}{\LTO{N2}}\ar[dl]\\&\LTO{Age}}}}\Too{F}\parbox{.8in}{\fbox{\xymatrix@=8pt{&\LTO{Name}\\&\LTO{Salary}\\\color{red}{\LTO{N}}\ar[uur]\ar[ur]\ar[dr]\\&\LTO{Age}}}
} 
\end{align*}
\end{minipage}
\begin{minipage}[t]{2.3in}
%\begin{align*} 
\parbox{2.3in}{ 
%\hspace{-.25in}
\begin{tabular}{| c || c | c | }\bhline
\multicolumn{3}{| c |}{{\sf N1}}\\\bhline {\sf  ID}&{\sf  Name}&{\sf  Salary}\\\bbhline 
1&Alice&\$100\\\hline 
2&Bob&\$250\\\hline 
3&Sue&\$300\\\bhline
\end{tabular}
%\hspace{.25in}
\begin{tabular}{| c || c | }\bhline
\multicolumn{2}{| c |}{{\sf N2}}\\\bhline {\sf  ID}&{\sf  Age}\\\bbhline 
4&20\\\hline 
5&20\\\hline 
6&30\\\bhline
\end{tabular}
}
$$
\downarrow \Sigma_F   
$$
%$$
 %\ \ \ \ \ \ \ \ \ \ \ \ \ \ \ \ \ \ \ \ \ \ \ \ \ \ \ \ \searrow \Pi_F
%$$
\begin{centering}
\begin{tabular}{| c || c | c | c | }\bhline\multicolumn{4}{| c |}{{\sf N}}\\\bhline {\sf  ID}&{\sf  Name}&{\sf  Salary}&{\sf  Age}\\\bbhline 
a&Alice&\$100&$null_1$\\\hline 
b&Bob&\$250&$null_2$\\\hline 
c&Sue&\$300&$null_3$\\\hline 
d&$null_4$&$null_5$&20\\\hline 
e&$null_6$&$null_7$&20\\\hline 
f&$null_8$&$null_9$&30\\\hline 
\end{tabular}
\end{centering}
\end{minipage}
\begin{minipage}[t]{.01in}
\hspace{-.3in}
$\From{\Delta_F}$ \\ \newline \newline \newline
\hspace*{-.35in}
%\vspace{.25in}
$\searrow \Pi_F$
\end{minipage}
\begin{minipage}[t]{2.5in}
%\hspace{-.21in}
%\parbox{1.5in}{
\begin{tabular}{| c || c | c | c | }\bhline\multicolumn{4}{| c |}{{\sf N}}\\\bhline {\sf  ID}&{\sf  Name}&{\sf  Salary}&{\sf  Age}\\\bbhline 
a&Alice&\$100&20\\\hline 
b&Bob&\$250&20\\\hline 
c&Sue&\$300&30\\\bhline
\end{tabular}
%\newline
%\hspace*{-2in}
$$
\;
$$
%}
%\end{align*} 
%\begin{align*} \parbox{2.5in}{ 
%\hspace{-.25in}
\begin{tabular}{| c || c | c | c | }\bhline\multicolumn{4}{| c |}{{\sf N}}\\\bhline {\sf  ID}&{\sf  Name}&{\sf  Salary}&{\sf  Age}\\\bbhline 
a&Alice&\$100&20\\\hline 
b&Alice&\$100&20\\\hline 
c&Alice&\$100&30\\\hline
d&Bob&\$250&20\\\hline 
e&Bob&\$250&20\\\hline 
f&Bob&\$250&30\\\hline
g&Sue&\$300&20\\\hline 
h&Sue&\$300&20\\\hline 
i&Sue&\$300&30\\\hline
\end{tabular}
%}
%\hspace{-.21in}
%\end{align*} 
\end{minipage}
\end{footnotesize}
\end{figure*}
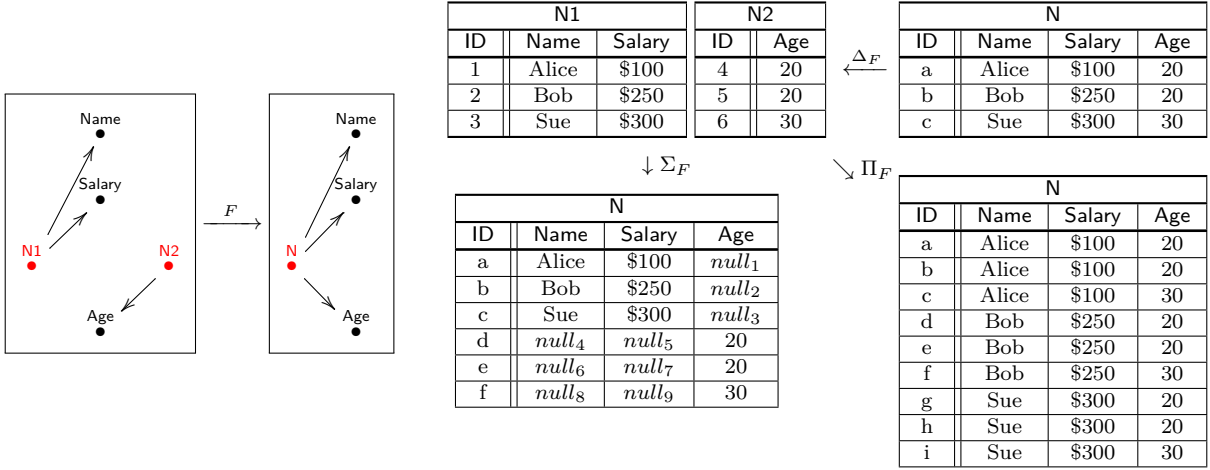
%%%%%%%%%%%%%%%%%%%%%%%%%%%%%%

\begin{figure*}
\caption{Data migration example with foreign keys}
\label{ex2}
\vspace{.2in}
\begin{footnotesize}
\begin{minipage}[t]{2.3in}
\vspace{-.7in}
\begin{align*} \parbox{1.0in}{\fbox{\xymatrix@=7pt{&\LTO{Name}\\&\LTO{Salary}\\\color{red}{\LTO{N1}} \ar[rr]^{\sf f} \ar[uur]\ar[ur]&&\color{red}{\LTO{N2}}\ar[dl]\\&\LTO{Age}}}}\Too{F}\parbox{.8in}{\fbox{\xymatrix@=7pt{&\LTO{Name}\\&\LTO{Salary}\\\color{red}{\LTO{N}}\ar[uur]\ar@{-}[ur]\ar[dr]\\&\LTO{Age}}}} \end{align*}
\end{minipage}
\begin{minipage}[b]{2.5in}
%\begin{align*} 
\parbox{2.7in}{ 
%\hspace{-.25in}
\begin{tabular}{| c || c | c | c | }\bhline
\multicolumn{4}{| c |}{{\sf N1}}\\\bhline {\sf  ID}&{\sf  Name}&{\sf  Salary} & {\sf  f} \\\bbhline 
1&Alice&\$100&4\\\hline 
2&Bob&\$250&5\\\hline 
3&Sue&\$300&6\\\bhline
\end{tabular}
%\hspace{-.2in}
\begin{tabular}{| c || c | }\bhline
\multicolumn{2}{| c |}{{\sf N2}}\\\bhline {\sf  ID}&{\sf  Age}\\\bbhline 
4&20\\\hline 
5&20\\\hline 
6&30\\\bhline
\end{tabular}}
\end{minipage}
\begin{minipage}[c]{.01in}
\vspace*{.2in}
\hspace{-.3in}
$\From{\Delta_F}$ \\ \newline
\hspace*{-.3in}
%\vspace{.25in}
$\To{\Sigma_F}$ \\ \newline 
\hspace*{-.3in}
$\To{\Pi_F}$ \\ \newline 
\end{minipage}
\begin{minipage}[c]{2.4in}
\begin{tabular}{| c || c | c | c | }\bhline\multicolumn{4}{| c |}{{\sf N}}\\\bhline {\sf  ID}&{\sf  Name}&{\sf  Salary}&{\sf  Age}\\\bbhline 
a&Alice&\$100&20\\\hline 
b&Bob&\$250&20\\\hline 
c&Sue&\$300&30\\\hline 
\end{tabular}
\end{minipage}
\end{footnotesize}
\end{figure*}
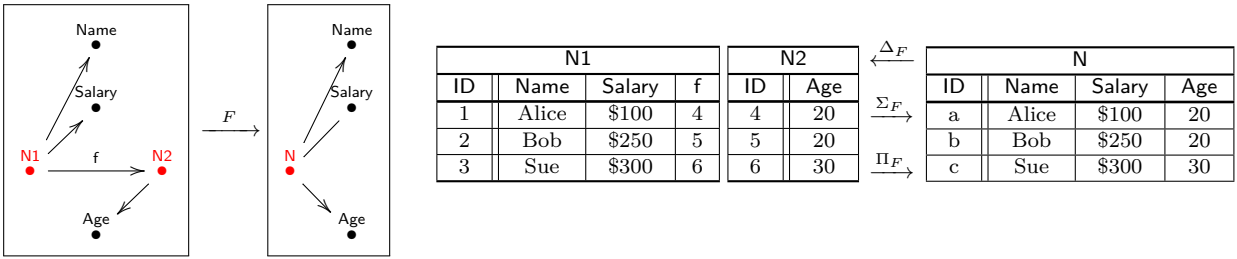
%%%%%%%%%%%%%%%%%%%%%%%%%%%%%%

\begin{defn}[Data Migration Functors]
Let $F\taking S \to T$ be a (untyped) signature morphism.  Define $\Delta_F\taking T\inst\to S\inst$ as: % composition: %; that is, given a $\mcD$-instance $I\taking \mcD\to\Set$ we will construct a $\mcC$-instance $\Delta_F(I)$. This is obtained simply by composing the functors $I$ and $F$ to get 
$$\parbox{1.6in}{$\Delta_F(I):=I\circ F\taking S\to\Set$} \ \ \ 
\parbox{1.5in}{\xymatrix{S\ar[r]^F\ar@/_1pc/[rr]_{\Delta_F(I)}&T\ar[r]^I&\Set}}
$$
and define $\Sigma_F, \Pi_F \taking S\inst\to T\inst$ to be \textnormal{left and right adjoint} to $\Delta_F$, respectively. %and $\Pi_F \taking S\inst\to T\inst $ to be \textnormal{right adjoint} to $\Delta_F$.
\end{defn}
Data migration functors extend to {\it typed data migration functors} over typed instances in a natural way.  An example typed data migration is shown in Figure~\ref{ex1} on the next page.  A similar example, but with a signature that contains an edge, is shown in Figure~\ref{ex2}.   %We will use typed signatures and instances as we examine each data migration functor in turn below.  Consider the following signature morphism $F : C \to D$:

\subsection{Data Migration on Homomorphisms}
\label{dmh}
 So far, we have only shown examples of the ``object'', or ``instance to instance'', part of the $\Delta$, $\Sigma$, and $\Pi$ data migration operations.  Because these data migration operations are functors, they also posses a ``morphism'', or ``database homomorphism to database homomorphism'' part.  Rather than give examples of these operations, we simply summarize that the functorial data model admits the following operations, for any $F \taking C \to D$.
\begin{itemize}
\item If $I,J \in D\inst$ and $h : I \Rightarrow J$ is a database homomorphism, then $\Delta_F(h) : \Delta_F(I) \Rightarrow \Delta_F(J)$. 
\item If $I, J \in C\inst$ and $h : I \Rightarrow J$ is a database homomorphism, then $\Sigma_F(h) : \Sigma_F(I) \Rightarrow \Sigma_F(J)$ and $\Pi_F(h) : \Pi_F(I) \Rightarrow \Pi_F(J)$. 
\end{itemize}

\section{FQL}

The goal of this section is to define and study an algebraic query language where every query denotes a composition of data migration functors.   Our syntax for queries is designed to build-in certain syntactic restrictions discussion in this section, and to provide a convenient normal form.

\begin{defn}[FQL Query]
A (typed) {\em FQL query} $Q$ from $S$ to $T$, denoted $Q\taking S\queryto T$, is a triple of typed signature morphisms $(F,G,H)$:
$$
S\From{F}S'\To{G}S''\To{H}T
$$
%\vspace{-.1in}
such that 
\begin{enumerate}
\item $\llbracket S \rrbracket, \llbracket S' \rrbracket, \llbracket S'' \rrbracket$, and $\llbracket T \rrbracket$ are finite
\item $G$'s attribute mapping is a bijection\footnote{All results in this paper should also apply to surjections, but we have only proved the bijective case.}.
\item $H$ is a {\it discrete op-fibration}~\cite{BW} 
\end{enumerate}

\end{defn}
\noindent
Semantically, the query $Q\taking S \queryto T$ corresponds to a functor $\church{Q}\taking S\inst\to T\inst$ (indeed, a functor taking finite instances to finite instances) given as follows:  
$$\church{Q}:=\church{ \Sigma_{H} \circ \Pi_{G} \circ \Delta_{F}}$$
By choosing two of $F$, $G$, and $H$ to be the identity, we can recover $\Delta$, $\Sigma$, and $\Pi$.  However, grouping $\Delta, \Sigma$, and $\Pi$ together like this formalizes a query as a disjoint union of a join of a projection.  Because $\Delta$ and $\Pi$ commute, we could also formalize a query as a disjoint union of a projection of a join.  

The three conditions above are crucial to ensuring both that FQL can be implemented on SPCU+keygen, and for FQL queries to be closed under composition:

\begin{enumerate}
\item  Condition 1 ensures that path equivalence in each signature is in fact decidable.  Moreover, the fiber product of two (infinite) finitely-presented categories may not be finitely presented, so it may not be possible to compute composition queries when schemas are infinite.  Finiteness of schemas also ensures that $\Pi$ operations always give finite answers.  When the target schema is infinite, $\Pi$ may create uncountably infinite result instances. Consider the unique signature morphism: $$S:=\fbox{$\LMO{s}$}\To{\;\;F\;\;}\LoopSchema=:T$$ Here $\church{T}$ has morphisms $\{f^n\|n\in\NN\}$ so it is infinite. Given the two-element instance $I\taking S\to\Set$ with $I(s)=\{\sf{Alice},\sf{Bob}\}$, the rowset in $\Pi_F(I)(s)$ is the (uncountable) set of infinite streams in $\{\sf{Alice},\sf{Bob}\}$, i.e. $\Pi_F(I)(s) \cong I(s)^\NN.$

\item Condition 2 ensures that the FQL query is ``domain independent''~\cite{DBLP:books/aw/AbiteboulHV95}; i.e., that constants resulting from any $\Pi$ operation are contained in the input instance.  Consider the unique morphism: $$S:=\fbox{$\LMO{s}$}\To{\;\;F\;\;}\fbox{$\LMO{s} - \circ {\sf a \ (Int)}$}=:T$$  Given the one-element instance $I\taking S\to\Set$ with $I(s)=\{\sf{Alice}\}$, the rowset in $\Pi_F(I)({\sf a})$ will contain, for every integer $i$, a row $(\sf{Alice}_i, i)$, i.e., $\Pi_F(I)({\sf a}) \cong \{ (\sf{Alice}_1, 1), (\sf{Alice}_2, 2), \ldots \}$.

\item Condition 3 ensures ``union-compatibility''~\cite{DBLP:books/aw/AbiteboulHV95}, which allows us to implement $\Sigma$ using SPCU's union operation.  A functor $F\taking\mcC\to\mcD$ is called a {\em discrete op-fibration} if, for every object $c\in\mcC$ and every morphism $g\taking d\to d'$ in $\mcD$ with $F(c)=d$, there exists a $c'$ and unique morphism $\bar{g}\taking c\to c'$ in $\mcC$ such that $F(\bar{g})=g$.   For example, the following functor, which maps $a$s to $A$, $b$s to $B$, etc. % $c$s to $C$, $g$s to $G$, and $h$s to $H$, 
is a discrete op-fibration:

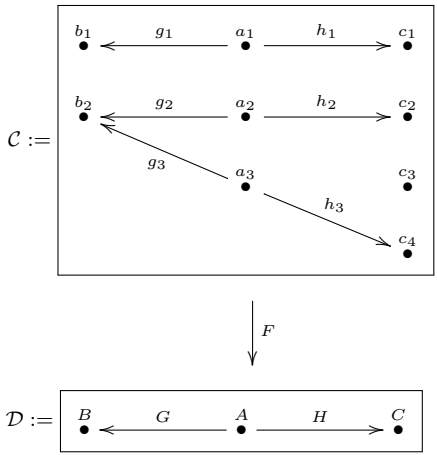
\begin{footnotesize}
$$
\mcC:=\parbox{2in}{\fbox{\xymatrix@=10pt{
\LMO{b_1}&\hspace{.3in}&\LMO{a_1}\ar[rr]^{h_1}\ar[ll]_{g_1}&\hspace{.3in}&\LMO{c_1}\\
\LMO{b_2}&&\LMO{a_2}\ar[ll]_{g_2}\ar[rr]^{h_2}&&\LMO{c_2}\\
&&\LMO{a_3}\ar[llu]^{g_3}\ar[rrd]^{h_3}&&\LMO{c_3}\\
&&&&\LMO{c_4}
}}}
$$
$$
\hspace{.4in}\xymatrix{~\ar[d]^F\\~}
$$
$$
\mcD:=\parbox{2in}{\fbox{\xymatrix@=10pt{
\LMO{B}&\hspace{.3in}&\LMO{A}\ar[rr]^H\ar[ll]_G&\hspace{.3in}&\LMO{C}
}}}
$$
\end{footnotesize}
%\vspace{-.1in}
Intuitively, $F\taking\mcC\to\mcD$ is a discrete op-fibration if ``the columns in each table $T$ of $\mcD$ are exactly matched by the columns in each table in $\mcC$ mapping to $T$.''  Since $a_1,a_2,a_3\mapsto A$, they must have the same column structure and they do: each has two non-ID columns. Similarly each of the $b_i$ and each of the $c_i$ have no non-ID columns, just like their images in $\mcD$. 

{\bf Remark.} $\Sigma_F$ is defined for any $F$, not just for discrete op-fibrations.  However, such unrestricted $\Sigma$s cannot be implemented using SPCU+keygen (even if we allow {\tt OUTER UNION}), and FQL queries that use such $\Sigma$s are not closed under composition.  
\end{enumerate}

\subsection{Composition}

\begin{thm}[Closure under composition]\label{thm:closed under comp}
Given any (typed) FQL queries $f\taking S \queryto X$ and $g\taking X \queryto T$, we can compute an FQL query $g \cdot f\taking S\queryto T$ such that $$\church{g\cdot f} \cong \church{g}\circ\church{f}$$  
\end{thm}
\begin{proof}[Sketch of proof]
Following the proof in \cite[Proposition 1.1.2]{GK}, it suffices to show the Beck-Chevalley conditions hold for $\Delta$ and its adjoints, and that dependent sums distribute over dependent products when the former are indexed by discrete op-fibrations. The proof and algorithm are given in the appendix as Theorem~\ref{xyzabc}.
\end{proof}
\noindent
Detailing the algorithm involves a number of constructions we have not defined, so we only give a brief sketch here.  To compose $\Sigma_t \circ \Pi_f \circ \Delta_s$ with $\Sigma_v \circ \Pi_g \circ \Delta_u$ we must construct the following diagram:

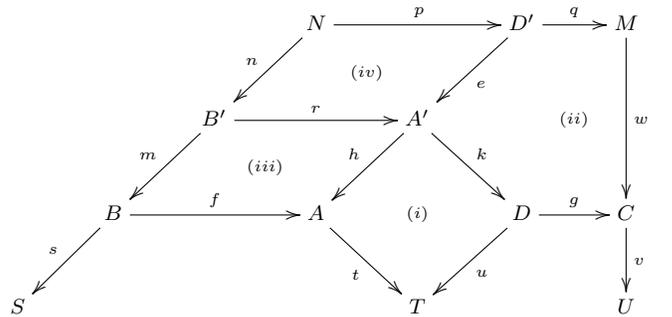
\begin{footnotesize}
\begin{align*}
\xymatrix{
&&&N\ar[rr]^{p}\ar[dl]_n\ar@{}[dr]|{(iv)}&&D'\ar[r]^{q}\ar[dl]^{e}\ar@{}[ddr]|{(ii)}&M\ar[dd]^{w}\\
&&B'\ar[rr]^{r}\ar[dl]_{m}\ar@{}[dr]|{(iii)}&&A'\ar[rd]^{k}\ar[ld]_{h}\ar@{}[dd]|{(i)}&&&\\
&B\ar[dl]_{s}\ar[rr]^{f}&&A\ar[dr]_{t}&&D\ar[dl]^{u}\ar[r]^{g}&C\ar[d]^{v}&\\
S&&&&T&&U}
\end{align*}
\end{footnotesize}
\noindent
First we form the pullback $(i)$. Then we form the typed distributivity diagram $(ii)$ as in Proposition \ref{prop:typed distributivity}. Then we form the typed comma categories $(iii)$ and $(iv)$ as in Proposition \ref{prop:typed comparison Delta Pi}. The result then follows from Proposition \ref{prop:pullback BC for typed Delta Sigma}, \ref{prop:typed distributivity}, and Corollary \ref{cor:comma BC for typed Delta Pi}.  The composed query is $\Sigma_{v \circ w} \circ \Pi_{q \circ p} \circ \Delta_{s \circ m \circ n}$.  We have implemented this composition algorithm in our prototype FQL IDE.

\section{FQL in SPCU+keygen}
In this section we let $F\taking S \to T$ be a typed signature morphism and we define SPCU+keygen translations for $\Delta_F$, $\Sigma_F$, and $\Pi_F$.  By an SPCU+keygen {\it program} we mean a list of SQL expressions, where each SQL expression is one of the following:
\begin{enumerate}
\item A conjuctive query over base tables; i.e., a SQL expression of the form {\tt SELECT DISTINCT - FROM - WHERE -}.
\item A union of SPCU-keygen expressions; i.e., a SQL expression {\tt - UNION - }.
\item A globally-unique ID generation operation (keygen) that appends to a given table a new column with globally unique IDs. To implement keygen in SQL, we first initialize a global variable, using, for example, the following MySQL code: {\tt 
SET @guid := 0}. Then, to implement keygen we increment this global variable.  For example, to add a column ``c'' of new unique IDs to a table ``t'' with one column called ``col'', we would generate:
\hspace*{-.3in}
\begin{verbatim}
SELECT col, @guid:=@guid+1 AS c FROM t
\end{verbatim}
\item A {\tt CREATE TABLE} followed by an {\tt INSERT INTO} of an SPCU-keygen expression.  
\end{enumerate}
A SQL/SPCU+keygen program that implements a functorial data migration $S\inst \to T\inst$ operates over a SQL schema containing one binary table for each node, edge, and attribute in $S$, and stores one binary table for each node, edge, and attribute in $T$.

 \subsection{Translating $\Delta$}
\begin{thm}
We can compute a SPCU+keygen program $[F]_\Delta$ such that for every $T$-instance $I\in T\inst$, we have $\Delta_F(I) \cong [F]_\Delta(I)$.
\end{thm}
\begin{proof}
See Construction \ref{const:Delta}.
\end{proof}
We sketch the algorithm as follows.  We are given a $T$-instance $I$, presented as a set of functional binary tables, and are tasked with creating the $S$-instance $\Delta_F(I)$.  We describe the result of $\Delta_F(I)$ on each table in the result instance by examining the schema $S$:
\begin{itemize}
\item for each node $N$ in $S$, the binary table $\Delta_F(N)$ is the binary table $I(F(N))$.
\item for each attribute $A$ in $S$, the binary table $\Delta_F(A)$ is the binary table $I(F(A))$.
\item Each edge $e : X \to Y$ in $S$ maps to a path $F(e) : F(X) \to F(Y)$ in $T$.  We compose the binary edges tables making up the path $F(e)$, and that becomes the binary table $\Delta_F(e)$.
\end{itemize}
The SQL generation algorithm for $\Delta$ above (and $\Sigma$, below) does not maintain the globally unique ID requirement.  For example, $\Delta$ can copy tables.  Hence we must also generate SQL to restore the unique ID invariant.  

\subsection{Translating $\Sigma$}

\begin{thm}
Suppose $F$ is a discrete op-fibration.  Then we can compute a SPCU+keygen program $[F]_\Sigma$ such that for every $S$-instance $I\in S\inst$, $\Sigma_F(I) \cong [F]_\Sigma(I)$.
\end{thm}
\begin{proof}
See Construction \ref{const:Sigma}.
\end{proof}

We sketch the algorithm as follows.  We are given a $S$-instance $I$, presented as a set of functional binary tables, and are tasked with creating the $T$-instance $\Sigma_F(I)$.  We describe the result of $\Sigma_F(I)$ on each table in the result instance by examining the schema $T$:
\begin{itemize}
\item for each node $N$ in $T$, the binary table $\Sigma_F(N)$ is the union of the binary node tables in $I$ that map to $N$ via $F$.  

\item for each attribute $A$ in $T$, the binary table $\Sigma_F(A)$ is the union of the binary attribute tables in $I$ that map to $A$ via $F$.  

\item  Let $e : X \to Y$ be an edge in $T$.  We know that for each $c \in F^{-1}(X)$ there is at least one path $p_c$ in $S$ such that $F(p_c) \cong e$.  Compose $p_c$ to a single binary table, and define $\Sigma_F(e)$ to be the union over all such $c$.  The choice of $p_c$ will not matter.

\end{itemize}

\subsection{Translating $\Pi$}
\begin{thm}
Suppose $\llbracket S \rrbracket$ and $\llbracket T \rrbracket$ are finite, and $F$ has a bijective attribute mapping.  Then we can compute a SPCU+keygen program $[F]_\Pi$ such that for every $S$-instance $I\in S\inst$, we have $\Pi_F(I) \cong [F]_\Pi(I)$.
\end{thm}
\begin{proof}
See Construction \ref{const:Pi}.
\end{proof}
The algorithm for $\Pi_F$ is more complex than for $\Delta_F$ and $\Sigma_F$.  It makes use of comma categories, which we have not defined, as well as ``limit tables'', which are a sort of ``join all''.  We define these now.

Let $B$ be a typed signature and $H$ a typed $B$-instance.  The limit table ${\sf lim}_B$ is computed as follows.  First, take the cartesian product of every node table in $B$, and naturally join the attribute tables of $B$.  Then, for each edge $e : n_1 \to n_2$ filter the table by $n_1 = n_2$.  This filtered table is the limit table ${\sf lim}_B$.

Let $S : A \to C$ and $T : B \to C$ be functors.  The comma category $(S \downarrow T)$ has for objects triples $(\alpha, \beta, f)$, with $\alpha$ an object in $A$, $\beta$ an object in $B$, and $f : S(\alpha) \to T(\beta)$ a morphism in $C$.  The morphisms from $(\alpha, \beta, f)$ to $(\alpha', \beta', f')$ are all the pairs $(g,h)$ where $g : \alpha \to \alpha'$ and $h : \beta \to \beta'$ are morphisms in $A$ and $B$ respectively such that $T(h) \circ f = f' \circ S(g)$.

The algorithm for $\Pi_F$ proceeds as follows.  First, for every object $d \in T$ we consider the comma category $B_d := (d \downarrow F)$ and its projection functor $q_d : (d \downarrow F) \to C$.  (Here we treat $d$ as a functor from the empty category).  Let $H_d := I \circ q_d : B_d \to \Set$, constructed by generating SQL for $\Delta_{q_d}(I)$.  We say that the limit table for $d$ is ${\sf lim}_{B_d} H_d$, as described above.  Now we can describe the target tables in $T$:

\begin{itemize}
\item for each node $N$ in $T$, generate globally unique IDs for each row in the limit table for $N$.  These IDs are $\Pi_F(N)$.
\item for each attribute $A$ off $N$ in $T$, $\Pi_F(A)$ will be a projection from the limit table for $N$.
\item for each edge $e \taking X \to Y$ in $T$, $\Pi_F(e)$ will be a map from $X$ to $Y$ given by joining the limit tables for $X$ and $Y$ on columns which ``factor through'' $e$.
\end{itemize}

{\bf Remark.} Our SPCU+keygen generation algorithms for $\Delta$ and $\Sigma$ work even when $\llbracket S \rrbracket$ and $\llbracket T \rrbracket$ are infinite, but this is not the case for $\Pi$.   But if we are willing to take on infinite SQL queries, we have $\Pi$ in general; see Construction \ref{const:Pi}.  An interesting direction for future work is to relate these infinite SQL queries to recursive languages such as Datalog.
%However, these infinite SQL queries will not in general be expressible using  to e.g., so-called recursive SQL queries; see condition 1 in Section 4. %Alternatively, it is possible to target a category besides ${\bf Set}$, provided that category is complete and co-complete.  For example, it is possible to store our data not as relations, but as programs  in the Turing-complete language PCF~\cite{Spivak:2012:FDM:2324905.2325108}~\cite{DBLP:journals/tcs/Plotkin77}.

\subsection{Translating Functorial Data Migrations of Homomorphisms}

We can also generate SPCU+keygen to implement data migrations of homomorphisms (c.f. section~\ref{dmh}).

\begin{thm}
For all $T$-instances $I, J\in T\inst$ and homomorphisms $h : I \Rightarrow J$, we can compute a SPCU+keygen program $[F]_\Delta$ such that $\Delta_F(h) \cong [F]_\Delta(h)$, and similarly for $\Sigma_F, \Pi_F$.
\end{thm}
\begin{proof}
See Constructions \ref{const:Delta}, \ref{const:Sigma}, and \ref{const:Pi}.
\end{proof}

%We can also obtain a countable sequence of increasingly good approximation of $\Pi_F$ using finite relational queries.
%
%\begin{thm}
%We can recursively enumerate a sequence $F_0, F_1, \ldots$ of finite relational queries and query morphisms that approximate $[F]_\Pi=\Pi_F$ in the sense that all the triangles below commute:
%$$\xymatrix{[F_0]_\Pi\ar[r]\ar[drrrr]&[F_1]_\Pi\ar[r]\ar[drrr]&[F_2]_\Pi\ar[r]\ar[drr]^\cdots&\cdots\;\;\;\ar[dr]\\&&&&[F]_\Pi}$$
%\end{thm}
%\begin{proof}
%See Proposition \ref{prop:approximation}
%\end{proof}

\vspace*{-.3in}
\section{SPCU in FQL}

To implement SPCU using FQL, we must encode relational schemas as signatures with a single ``active domain''~\cite{DBLP:books/aw/AbiteboulHV95} attribute.   For example, a schema with two relations $R(c_1, \ldots, c_n)$ and $R^\prime(c^\prime_1, \ldots, c^\prime_{n^\prime})$ becomes: 
$$
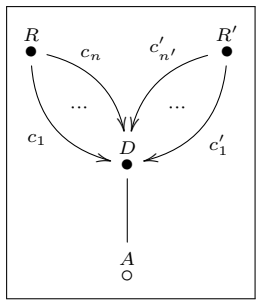
\fbox{\parbox{1in}{\xymatrix{\LMO{R}\ar@/_1pc/[dr]_{c_1}\ar@{}[dr]|{\cdots}\ar@/^1pc/[dr]^(.3){c_n}&&\LMO{R'}\ar@/_1pc/[dl]_(.3){c'_{n'}}\ar@{}[dl]|{\cdots}\ar@/^1pc/[dl]^{c'_1}\\&\LMO{D} \\ & \stackrel{A}{\circ} \ar@{-}[u]}}}
$$

%$$
%\fbox{\parbox{1in}{\xymatrix{ & \LMO{R}  \ar@/_1pc/[d]_{c_1}\ar@{}[d]|\cdots\ar@/^1pc/[d]^{c_n}\\ \stackrel{A}{\circ}   & \LMO{D} \ar@{-}[l] }}}
%$$
We might expect that the $c_1, \ldots, c_n$ would be attributes of node $R$, and hence there would be no node $D$, but that does not work: attributes may not be (directly) joined on.  Instead, we must think of each column of $R$ as a mapping from $R$'s domain to IDs in $D$, and $A$ as a mapping from IDs in $D$ to constants.  Note that $A$ need not be injective; our constructions below will not in general maintain injectiveness of $A$ (although injectiveness can be recovered using the $dedup$ operation defined in the next section). We will write $[R]$ for the encoding of a relational schema $R$ and $[I]$ for the encoding of a relational $R$-instance $I$.
\vspace{-.1in}
%When a relational schema ${\bf R_1}$ is a subset of a relational schema ${\bf R_2}$, we may convert  $[{\bf R_1}]$ instances to $[{\bf R_2}]$ instances using $\Delta$:
%$$
%\Delta_{\subseteq} : [{\bf R_2}]-Inst \to [{\bf R_1}]-Inst  
%$$
\subsection{SPCU in FQL (Bag Semantics)}
\label{conjbag}
FQL can implement SPCU queries {\it under bag semantics} directly using the above encoding.  In what follows we will omit the attribute $A$ from the diagrams.  We may express the (bag) operations $\pi, \sigma, \times, +$ as follows: 
\begin{itemize}
\item Let $R$ be a table. We can express $\pi_{i_1, \ldots, i_k} R$ as $\Delta_F$, where $F$ is any functor sending $\pi R$ to $R$ and $D$ to $D$ in the following diagram:
 $$
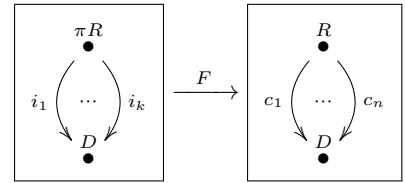
\fbox{\parbox{1in}{\xymatrix{\LMO{\pi R}\ar@/_1pc/[d]_{i_1}\ar@{}[d]|{\cdots}\ar@/^1pc/[d]^{i_k}\\\LMO{D}}}}
\Too{F}
\fbox{\parbox{1in}{\xymatrix{\LMO{R}\ar@/_1pc/[d]_{c_1}\ar@{}[d]|{\cdots}\ar@/^1pc/[d]^{c_n}\\\LMO{D}}}}
 $$
This construction is only appropriate for bag semantics because $\pi R$ will have the same number of rows as $R$.

\item Let $R$ be a table.  We can express $\sigma_{a=b} R$ as $\Delta_F \circ \Pi_F$, where $F$ is:
$$
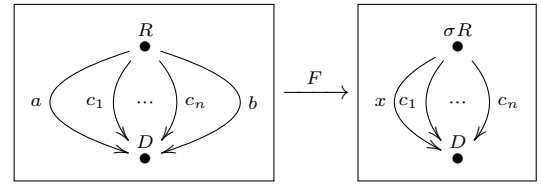
\fbox{\parbox{1in}{\xymatrix{\LMO{R}\ar@/_3pc/[d]_{a}\ar@/_1pc/[d]_{c_1}\ar@{}[d]|{\cdots}\ar@/^1pc/[d]^{c_n}\ar@/^3pc/[d]^{b}\\\LMO{D}}}}
\Too{F}
\fbox{\parbox{1in}{\xymatrix{\LMO{\sigma R}\ar@/_2pc/[d]_{x}\ar@/_1pc/[d]_{c_1}\ar@{}[d]|\cdots\ar@/^1pc/[d]^{c_n}\\\LMO{D}}}}
 $$
Here $F(a)=F(b)=x$ and $F(c_i) = c_i$ for $1 \leq i \leq n$.

{\bf Remark.} In fact, $\Delta_F$ creates the duplicated column found in $\sigma_{a=b} R$. If we wanted the more economical query in which the column is not duplicated, we would use $\Pi_F$ instead of $\Delta_F \circ \Pi_F$.

\item Let $R$ and $R'$ be tables.  We can express $R \times R'$ as $\Pi_F$, where $F$ is:
$$
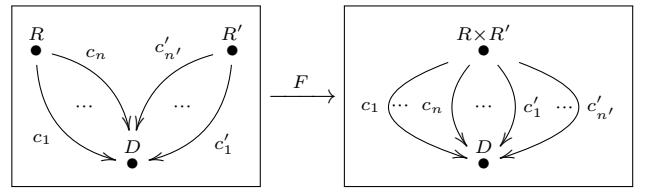
\fbox{\parbox{1in}{\xymatrix{\LMO{R}\ar@/_1pc/[dr]_{c_1}\ar@{}[dr]|{\cdots}\ar@/^1pc/[dr]^(.3){c_n}&&\LMO{R'}\ar@/_1pc/[dl]_(.3){c'_{n'}}\ar@{}[dl]|{\cdots}\ar@/^1pc/[dl]^{c'_1}\\&\LMO{D}}}}
\Too{F}
\fbox{\parbox{1in}{\xymatrix{\LMO{R \times R'}\ar@/_3pc/[d]_{c_1}\ar@/_1pc/[d]_{\cdots\;\; c_n}\ar@{}[d]|{\cdots}\ar@/^1pc/[d]^{c'_1\;\;\cdots}\ar@/^3pc/[d]^{c'_{n'}}\\\LMO{D}}}}
$$

\item Let $R$ and $R'$ be union-compatible tables.  We can express $R + R'$ as $\Sigma_F$, where $F$ is :
$$
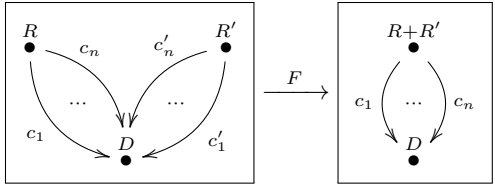
\fbox{\parbox{1in}{\xymatrix{\LMO{R}\ar@/_1pc/[dr]_{c_1}\ar@{}[dr]|{\cdots}\ar@/^1pc/[dr]^(.3){c_n}&&\LMO{R'}\ar@/_1pc/[dl]_(.3){c'_{n}}\ar@{}[dl]|{\cdots}\ar@/^1pc/[dl]^{c'_1}\\&\LMO{D}}}}
\Too{F}
\fbox{\parbox{1in}{\xymatrix{\LMO{R + R'}\ar@/_1pc/[d]_{c_1}\ar@{}[d]|{\cdots}\ar@/^1pc/[d]^{c_n}\\\LMO{D}}}}
$$
\end{itemize}

\begin{thm}[SPCU in FQL (Bag Semantics)]
Let $R$ be a relational schema, and $I$ an $R$-instance. For every SPCU query $q$ under bag semantics we can compute a FQL query $[q]$ such that $[q(I)] \cong [q]([I])$.
\end{thm}
\begin{proof}
See Propositions \ref{prop:conjunctive RA bag} and \ref{ghjkl}. 
\end{proof}

\subsection{SPCU in FQL (Set Semantics)}
The encoding strategy described above fails for SPCU under its typical set-theoretic semantics.  For example, consider a simple two column relational table $R$, its encoded FQL instance $[R]$, and an attempt to project off col1 using $\Delta$:

\begin{footnotesize}
\begin{align*}
\begin{tabular}{| l | l |}\bhline
\multicolumn{2}{| c |}{R}\\\bbhline 
{\sf  col1}&{\sf  col2}\\\hline 
x&y\\\hline 
x&z\\\bhline 
\end{tabular}\hspace{.1in}
\begin{tabular}{| c || c | c |}\bhline
\multicolumn{3}{| c |}{$[R]$}\\\bbhline 
{\sf  ID} & {\sf  col1}&{\sf  col2}\\\hline 
0 & x&y\\\hline 
1 & x&z\\\bhline  
\end{tabular}\hspace{.1in}
\begin{tabular}{| c || c |}\bhline
\multicolumn{2}{| c |}{$\Delta [R]$}\\\bbhline 
{\sf  ID} & {\sf  col1}\\\hline 
0& x \\\hline 
1 & x\\\bhline 
\end{tabular}
\end{align*}
\end{footnotesize}

This answer is incorrect under projection's set-theoretic semantics, because it has the wrong number of rows.  However, it is possible to extend FQL with an idempotent operation, $dedup_T \taking T\inst \to T\inst$, such that FQL+$dedup$ can implement every SPCU query under set-theoretic semantics.  Intuitively, $dedup_T$ converts a $T$-instance to a smaller $T$-instance by equating IDs that cannot be distinguished by any attribute along any path.  In relational languages based on object-IDs, this operation is usually called ``copy elimination''~\cite{Abiteboul:1989:OIQ:67544.66941}. The $dedup$ of the above instance is:

\begin{footnotesize}
\begin{align*}
\begin{tabular}{| c || c |}\bhline  
\multicolumn{2}{| c |}{$dedup(\Delta [R])$}\\\bbhline 
{\sf  ID} & {\sf  col1}\\\hline 
0& x \\\hline 
\end{tabular}
\end{align*}
\end{footnotesize}
\noindent
which is correct for the set-theoretic semantics.  %We have the the following theorem:

\begin{thm}[SPCU in FQL (Set Semantics)]
Let $R$ be a relational schema, and $I$ an $R$-instance. For every SPCU query $q$ under set semantics we can compute a FQL program $[q]$ such that $[q(I)] \cong dedup_T([q]([I]))$.
\end{thm}
\begin{proof}
See Propositions \ref{prop:conjunctive RA bag} and \ref{ghjkl}
\end{proof}
Provided $T$ is obtained from a relational schema (i.e., has the pointed form described in this section), $dedup_T$ can be implemented using SPCU+keygen, and data migrations of the form $dedup \circ Q$, where $Q$ is an FQL query, are closed under composition.  The FQL IDE implements a (partial) translation SQL $\to$ FQL using the above constructions.

\section{Conclusion}\label{sec:future}

The functorial data model admits other operations on instances and homomorphisms besides the $\Delta,\Sigma,\Pi$ data migration functors described in this paper.  In categorical terms, for every untyped schema $S$, the category of $S$ instances is a topos (model of higher-order logic)~\cite{BW}, and for every finite typed schema $S$ with finite attributes (e.g., {\sf String} is not allowed, but {\sf VARCHAR}(256) is), the category of finite $S$ instances is a topos.  The FQL IDE implements $n$-ary products and co-products of instances by translating to SPCU+keygen.  Exponential and sub-object classifier instances cannot be implemented using SQL, but the FQL IDE can execute them directly.

\bibliographystyle{plain}
\begin{footnotesize}

\end{footnotesize}

%\newpage

\end{multicols}

\newgeometry{left=1.6in,right=1.6in,top=1.4in,bottom=1.4in}

\section{Categorical constructions for data migration}

Here we give many standard definitions from category theory and a few less-than-standard (or original) results. The proofs often assume more knowledge than the definitions do. Let us also note our use of the term {\em essential}, which basically means ``up to isomorphism". Thus an object $X$ having a certain property is {\em essentially unique} if every other object having that property is isomorphic to $X$; an object $X$ is in the {\em essential image} of some functor if it is isomorphic to an object in the image of that functor; etc.
\subsection{Basic constructions}

\begin{definition}[Fiber products of sets]

Suppose given the diagram of sets and functions as to the left in (\ref{dia:fp sets}).
\begin{align}\label{dia:fp sets}
\xymatrix{&B\ar[d]^g&&A\cross_CB\ar[r]^{f'}\ar[dr]^h\ar[d]_{g'}&B\ar[d]^g\\
A\ar[r]_f&C&&A\ar[r]_f&C}
\end{align}
Its {\em fiber product} is the commutative diagram to the right in (\ref{dia:fp sets}), where we define 
$$A\cross_CB:=\{(a,c,b)\|f(a)=c=g(b)\}$$ 
and $f', g'$, and $h=g\circ f'=f\circ g'$ are the obvious projections (e.g. $f'(a,c,b)=b$). We sometimes refer to just $(A\cross_CB,f',g', h)$ or even to $A\cross_CB$ as the fiber product. 

\end{definition}

\begin{example}[Chain signatures]\label{ex:chains}

Let $n\in\NN$ be a natural number. The {\em chain signature on $n$ arrows}, denoted $\vect{n}$,  \footnote{The chain signature $\vect{n}$ is often denoted $[n]$ in the literature, but we did not want to further overload the bracket [-] notation} is the graph $$\xymatrix{\LMO{0}\ar[r]&\LMO{1}\ar[r]&\cdots\ar[r]&\LMO{n}}$$ with no path equivalences. One can check that $\vect{n}$ has ${n+2}\choose{2}$-many paths. 

The chain signature $\vect{0}$ is the terminal object in the category of signatures. 

A signature morphism $c\taking \vect{0}\to C$ can be identified with an object in $C$. A signature morphism $f\taking\vect{1}\to C$ can be identified with a morphism in $C$. It is determined up to path equivalence by a pair $(n,p)$ where $n\in\NN$ is a natural number and $p\taking\vect{n}\to C$ is a morphism of graphs.

\end{example}

\begin{definition}[Fiber product of signatures]\label{def:fp sig}

Suppose given the diagram of signatures and signatures mappings as to the left in (\ref{dia:fiber prod of sigs}).
\begin{align}\label{dia:fiber prod of sigs}
\xymatrix{&B\ar[d]^g&&A\cross_CB\ar[r]^{f'}\ar[dr]^h\ar[d]_{g'}&B\ar[d]^g\\
A\ar[r]_f&C&&A\ar[r]_f&C}
\end{align}
On objects and arrows we have the following diagrams of sets: 
\begin{align}\label{dia:fiber prod on objects and arrows}
\parbox{1in}{
\xymatrix{&\Ob(B)\ar[d]^{\Ob(g)}\\
\Ob(A)\ar[r]_{\Ob(f)}&\Ob(C)}
}\hspace{1in}
\parbox{1in}{
\xymatrix{&\Arr(B)\ar[d]^{\Arr(g)}\\
\Arr(A)\ar[r]_{\Arr(f)}&\Mor(C)}
}
\end{align}

We define the {\em fiber product} of the left-hand diagram in (\ref{dia:fiber prod of sigs}) to be the right-hand commutative diagram in (\ref{dia:fiber prod of sigs}), where we the objects and arrows of $A\cross_CB$ are given by the fiber product of the left and right diagram of sets in (\ref{dia:fiber prod on objects and arrows}), respectively. A pair of paths in $A\cross_CB$ are equivalent if they map to equivalent paths in both $A$ and $B$ (and therefore in $C$). Note that we have given $A\cross_CB$ as a signature, i.e. we have provided a finite presentation of it.

\end{definition}

\begin{example}[Pre-image of an object or morphism]

Let $F\taking C\to D$ be a functor. Given an object $d\in \Ob(D)$ (respectively a morphism $d\in\Mor(D)$), we can regard it as a functor $d\taking\vect{0}\to D$ (respectively a functor $d\taking \vect{1}\to D$), as in Example \ref{ex:chains}. The pre-image of $d$ is the subcategory of $C$ that maps entirely into $d$. This is the fiber product of the diagram $C\To{F}D\From{d}\vect{0}$ (respectively the diagram $C\To{F}D\From{d}\vect{1}$).

\end{example}

A signature is assumed to have finitely many arrows, but because it may have loops, $C$ may have infinitely many non-equivalent paths. We make the following definitions to deal with this.

\begin{definition}

A signature $C$ is called {\em finite} if there is a finite set $S$ of paths such that every path in $C$ is equivalent to some $s\in S$. An equivalent condition is that the schema $\church{C}$ has finitely many morphisms. (Note that if $C$ is finite then $C$ has finitely many objects too, since each object has an associated trivial path.)

\end{definition}

Recall (from the top of Section \ref{sec:FDM}) that every signature is assumed to have a finite set of objects and arrows. The above notion is a semantic one: a signature $C$ is called finite if $\church{C}$ is a finite category.

\begin{definition}

A signature $C$ is called {\em acyclic} if it has no non-trivial loops. In other words, for all objects $c\in\Ob(C)$ every path $p\taking c\to c$ is a trivial path.

\end{definition}

\begin{lemma}[Acyclicity implies finiteness]

If $C$ is an acyclic signature then it is finite.

\end{lemma}

\begin{proof}

All of our signatures are assumed to include only finitely many objects and arrows. If $C$ is acyclic then no edge can appear twice in the same path. Suppose that $C$ has $n$ arrows. Then the number of paths in $C$ is bounded by the number of lists in $n$ letters of length at most $n$. This is finite.

\end{proof}

\begin{lemma}[Acyclicity and finiteness are preserved under formation of fiber product]

Given a diagram $A\To{F}C\From{G}B$ of signatures such that each is acyclic (respectively, finite), the fiber product signature $A\cross_CB$ is acyclic (respectively, finite).

\end{lemma}

\begin{proof}

The fiber product of finite sets is clearly finite, and the fiber product of categories is computed by twice taking the fiber products of sets (see Definition \ref{def:fp sig}). We now show that acyclicity is preserved.

Let $\Loop$ denote the signature
\begin{align}\label{dia:loop}
\Loop:=\LoopSchema
\end{align} A category $C$ is acyclic if and only if every functor $L\taking\Loop\to C$ factors through the unique functor $\Loop\to\vect{0}$, i.e. if and only if every loop is a trivial path. The result follows by the universal property for fiber products.

\end{proof}

\begin{definition}[Comma category]\label{def:comma}

To a diagram of categories $A\To{f}C\From{g}B$ we can associate the {\em comma category of $f$ over $g$}, denoted $(f\down g)$, defined as follows. 
$$\Ob(f\down g):=\{(a,b,\gamma)\|a\in\Ob(A),b\in\Ob(B),\gamma\taking f(a)\to g(b)\text{ in }C\}$$
$$\Hom_{(f\down g)}((a,b,\gamma),(a',b',\gamma')):=\{(m,n)\|m\taking a\to a'\text{ in } A,n\taking b\to b'\text{ in } B, \gamma'\circ f(m)=g(n)\circ\gamma\}$$

There are obvious {\em projection} functors $(f\down g)\To{p} A$ and $(f\down g)\To{q} B$ and a natural transformation $\alpha\taking f\circ p\to g\circ q$, and we summarize this in a canonical natural transformation diagram 
\begin{align}\label{dia:nt for comma}
\xymatrix{(f\down g)\ar[r]^q\ar[d]_p\ar@{}[dr]|{\stackrel{\alpha}{\Nearrow}}&B\ar[d]^g\\A\ar[r]_f&C}
\end{align} 

\end{definition}

Below is a heuristic picture that may be helpful. The outside square consists of categories and functors (note that it does not commute!). The inside square represents a morphism in $(f\down g)$ from object $(a,b,\gamma)$ to object $(a',b',\gamma')$; this diagram is required to commute in $C$.

$$\xymatrix@=10pt{
&&&(f\down g)\ar@/_2pc/[lllddd]_p\ar@/^2pc/[rrrddd]^q\\\\
\\\color{red}{A}\ar@/_2pc/[dddrrr]_f&\color{red}{a}\ar@[red][d]_{\color{red}{m}}&\color{blue}{f(a)}\ar@[blue][rr]^{\color{blue}{\gamma}}\ar@[blue][d]_{\color{blue}{f(m)}}&&\color{blue}{g(b)}\ar@[blue][d]^{\color{blue}{g(n)}}&\color{ForestGreen}{b}\ar@[ForestGreen][d]^{\color{ForestGreen}{n}}&\color{ForestGreen}{B}\ar@/^2pc/[dddlll]^g\\
&\color{red}{a'}&\color{blue}{f(a')}\ar@[blue][rr]_{\color{blue}{\gamma'}}&&\color{blue}{g(b')}&\color{ForestGreen}{b'}\\\\
&&&\color{blue}{C}
}
$$

\begin{example}[Slice category and coslice category]\label{ex:slice}

Let $f\taking A\to B$ be a functor, and let $b\in\Ob(B)$ be an object, which we represent as a functor $b\taking\vect{0}\to B$. The {\em slice category of $f$ over $b$} is defined as $(f\down b)$. The {\em coslice category of $b$ over $F$} is defined as $(b\down f)$.

For example let $\singleton$ denote a one element set, and let $\id_\Set$ be the identity functor on $\Set$. Then the coslice category $(\singleton\down\id_\Set)$ is the category of pointed sets, which we denote by $\Set_*$. The slice category $(\id_\Set\down\singleton)$ is isomorphic to $\Set$.

\end{example}

\begin{lemma}[Comma categories as fiber products]\label{lemma:comma as fp}

Given a diagram $A\To{f}C\From{g}B$, the comma category $(f\down g)$ can be obtained as the fiber product in the diagram 
$$
\xymatrix{(f\down g)\ar[r]^{\alpha}\ar[d]_{(p,q)}&C^{\vect{1}}\ar[d]^{(dom,cod)}\\A\cross B\ar[r]_{(f,g)}&C\cross C}
$$
where $C^{\vect{1}}$ is the category of functors $\vect{1}\to C$. 

\end{lemma}

\begin{proof}

Straightforward.

\end{proof}

\begin{lemma}[Acyclicity and finiteness are preserved under formation of comma category]

Let $A, B$, and $C$ be acyclic (resp. finite) categories, and suppose we have functors $A\To{f}C\From{g}B$. Then the comma category $(f\down g)$ is acyclic (resp. finite).

\end{lemma}

\begin{proof}

Since the fiber product of finite sets is finite, Lemma \ref{lemma:comma as fp} implies that the formation of comma categories preserves finiteness. Suppose that $A$ and $B$ are acyclic (in fact this is enough to imply that $(f\down g)$ is acyclic). Given a functor $L\taking\Loop\to(f\down g)$, this implies that composing with the projections to $A$ and $B$ yields trivial loops. In other words $L$ consists of a diagram in $C$ of the form 
$$
\xymatrix{f(a)\ar[r]^\gamma\ar@{=}[d]&g(b)\ar@{=}[d]\\f(a)\ar[r]_{\gamma'}&g(b)}
$$
But this implies that $\gamma=\gamma'$, so $L$ is a trivial loop too, competing the proof.

\end{proof}

\subsection{Data migration functors}

\begin{construction}[$\Delta$]\label{const:Delta}

Let $F\taking C\to D$ be a functor. We will define a functor $\Delta_F\taking D\inst\to C\inst$; that is, given a $D$-instance $I\taking D\to\Set$ we will naturally construct a $C$-instance $\Delta_F(I)$. Mathematically, this is obtained simply by composing the functors to get 
$$\parbox{1.6in}{$\Delta_F(I):=I\circ F\taking C\to\Set$}
\hspace{1in}
\parbox{1.5in}{\xymatrix{C\ar[r]^F\ar@/_1pc/[rr]_{\Delta_FI}&D\ar[r]^I&\Set}}
$$
To understand this on the level of binary tables, we take $F$ and $I$ as given and proceed to define the $C$-instance $J:=\Delta_FI$ as follows. Given an object $c\in\Ob(C)$, it is sent to an object $d:=F(c)\in\Ob(D)$, and $I(d)$ is an entity table. Assign $J(c)=I(d)$. Given an arrow $g\taking c\to c'$ in $C$, it is sent to a path $p=F(g)$ in $D$. We can compose this to a binary table $[p]$. Assign $J(g)=[p]$. It is easy to check that these assignments preserve the axioms.

\end{construction}

\begin{lemma}[Finiteness is preserved under $\Delta$]

Suppose that $C$ and $D$ are finite categories and $F\taking C\to D$ is a functor. If $I\taking D\to\Set$ is a finite $D$-instance then $\Delta_FI$ is a finite $C$-instance.

\end{lemma}

\begin{proof}

This follows by construction: every object and arrow in $C$ is assigned some finite join of the tables associated to objects and arrows in $D$, and finite joins of finite tables are finite.

\end{proof}

Recall the definition of discrete op-fibrations from Definition \ref{def:discrete op-fibration}.

\begin{lemma}[Discrete op-fibrations]\label{lemma:discrete op-fibrations}

Let $F\taking C\to D$ be a signature morphism. The following are conditions on $F$ are equivalent:
\begin{enumerate}
\item The functor $[F]\taking[C]\to[D]$ is a discrete op-fibration.
\item For every choice of object $c_0$ in $C$ and path $q\taking d_0\to d_n$ in $D$ with $F(c_0)=d_0$, there exists a path $p\in C$ such that $F(p)\sim q$, and $p$ is unique up to path equivalence.
\item For every choice of object $c_0$ in $C$ and edge $q\taking d_0\to d_1$ in $D$ with $F(c_0)=d_0$, there exists a path $p\in C$ such that $F(p)\sim q$, and $p$ is unique up to path equivalence.
\end{enumerate}

\end{lemma}

\begin{proof}

It is clear that condition (2) implies condition (3), and the converse is true because concatenations of equivalent paths are equivalent. It suffices to show that (1) and (2) are equivalent.
 
It is shown in \cite{Spivak:1202.2591} that a functor $\pi\taking X\to S$ between categories $X$ and $S$ is a discrete op-fibration if and only if for each solid-arrow square of the form 
$$
\xymatrix@=30pt{\vect{0}\ar[r]^{c_0}\ar[d]_{i_0}&X\ar[d]^\pi\\\vect{1}\ar@{-->}[ur]^{\exists!p}\ar[r]_q&S,}
$$
where $i_0\taking \vect{0}\to\vect{1}$ sends the unique object of $\vect{0}$ to the initial object of $\vect{1}$, there exists a unique functor $\ell$ such that both triangles commute. We now translate this statement to the language of signatures with $\church{C}=X$ and $\church{D}=S$. A functor $\vect{1}\to\church{C}$ corresponds to a path in $C$, and the uniqueness of lift $p$ corresponds to uniqueness of equivalence class. This completes the proof.

\end{proof}

\begin{example}[A discrete op-fibration]\label{ex:discrete opfib}

Let $C, D$ be signatures, and let $F\taking C\to D$ be as suggested by the picture below.
$$
C:=\parbox{2in}{\fbox{\xymatrix@=10pt{
\LMO{b_1}&\hspace{.3in}&\LMO{a_1}\ar[rr]^{h_1}\ar[ll]_{g_1}&\hspace{.3in}&\LMO{c_1}\\
\LMO{b_2}&&\LMO{a_2}\ar[ll]_{g_2}\ar[rr]^{h_2}&&\LMO{c_2}\\
&&\LMO{a_3}\ar[llu]^{g_3}\ar[rrd]^{h_3}&&\LMO{c_3}\\
&&&&\LMO{c_4}
}}}
$$
$$
\hspace{.4in}\xymatrix{~\ar[d]^F\\~}
$$
$$
D:=\parbox{2in}{\fbox{\xymatrix@=10pt{
\LMO{B}&\hspace{.3in}&\LMO{A}\ar[rr]^H\ar[ll]_G&\hspace{.3in}&\LMO{C}
}}}
$$
This is a discrete op-fibration.

\end{example}

\begin{corollary}[Pre-image of an object under a discrete op-fibration is discrete]\label{cor:fibers of discrete op-fibrations}

Let $F\taking C\to D$ be a discrete op-fibration and let $d\in\Ob(D)$ be an object. Then the pre-image $F^\m1(d)$ is a discrete signature (i.e. every path in $F^\m1(d)$ is equivalent to a trivial path).

\end{corollary}

\begin{proof}

For any object $d\in\Ob(D)$ the pre-image $F^\m1(d)$ is either empty or not. If it is empty, then it is discrete. If $F^\m1(d)$ is non-empty, let $e\taking c_0\to c_1$ be an edge in it. Then $F(e)=d=F(c_0)$ so we have an equivalence $e\sim c_0$ by Lemma \ref{lemma:discrete op-fibrations}.

\end{proof}

\begin{example}

Let $F\taking A\to B$ be a functor. For any object $b\in B$, considered as a functor $b\taking\vect{0}\to B$, the induced functor $(b\down F)\to B$ is a discrete op-fibration. Indeed, given an object $b\To{g}F(a)$ in $(b\down F)$ and a morphism $h\taking a\to a'$ in $A$, there is a unique map $b\To{g}F(a)\To{F(h)}F(a')$ over it.

\end{example}

\begin{definition}

A signature $C$ is said to {\em have non-redundant edges} if it satisfies the following condition for every edge $e$ and path $p$ in $C$: If $e\sim p$, then $p$ has length 1 and $e=p$.

\end{definition}

\begin{proposition}

Every acyclic signature is equivalent to a signature with non-redundant edges.

\end{proposition}

\begin{proof}

The important observation is that if $C$ is acyclic then for any edge $e$ and path $p$ in $C$, if $e$ is an edge in $p$ and $e\sim p$ then $e=p$. Thus, we know that if $C$ is acyclic and $e\sim p$ then $e$ is not in $p$. So enumerate the edges of $C$ as $E_C:=\{e_1,\ldots,e_n\}$. If $e_1$ is equivalent to a path in $E_C-\{e_1\}$ then there is an equivalence of signatures $C\To{\iso}C-\{e_1\}=:C_1$; if not, let $C_1:=C$. Proceed by induction to remove each $e_i$ that is equivalent to a path in $C_{i-1}$, and at the end no edge will be redundant.

\end{proof}

\begin{corollary}[Discrete op-fibrations and preservation of path length]

If $C$ and $D$ have non-redundant edges and $F\taking C\to D$ is a discrete op-fibration, then for every choice of object $c_0$ in $C$ and path $q=d_0.e'_1.e'_2.\ldots.e'_n$ of length $n$ in $D$ with $F(c_0)=d_0$, there exists a unique path $p=c_0.e_1.e_2.\ldots.e_n$ of length $n$ in $C$ such that $F(e_i)=e'_i$ for each $1\leq i\leq n$, so in particular $F(p)=q$.

\end{corollary}

\begin{proof}

Let $c_0, d_0$, and $q$ be as in the hypothesis. We proceed by induction on $n$, the length of $q$. In the base case $n=0$ then Corollary \ref{cor:fibers of discrete op-fibrations} implies that every edge in $F^\m1(d_0)$ is equivalent to the trivial path $c_0$, so the result follows by the non-redundancy of edges in $C$. Suppose the result holds for some $n\in\NN$. To prove the result for $n+1$ it suffices to consider the final edge of $q$, i.e. we assume that $q=e'$ is simply an edge. By Lemma \ref{lemma:discrete op-fibrations} there exists a path $p\in C$ such that $F(p)\sim e'$, so by non-redundancy in $D$ we know that $F(p)$ has length 1. This implies that for precisely one edge $e_i$ in $p$ we have $F(e_i)=e'$, and for all other edges $e_j$ in $p$ we have $F(e_j)$ is a trivial path. But by the base case this implies that $p$ has length 1, completing the proof.

\end{proof}

\begin{construction}[$\Sigma$ for discrete opfibrations]\label{const:Sigma}

Suppose that $F\taking C\to D$ is a discrete op-fibration. We can succinctly define $\Sigma_F\taking C\inst\to D\inst$ to be the left adjoint to $\Delta_F$, however the formula has a simple description which we give now. Suppose we are given $F$ and an instance $I\taking C\to\Set$, considered as a collection of binary tables, one for each object and each edge in $C$. We are tasked with finding a $D$-instance, $J:=\Sigma_FI\taking D\to\Set$.

We first define $J$ on an arbitrary object $d\in\Ob(D)$. By Corollary \ref{cor:fibers of discrete op-fibrations}, the pre-image $F^\m1(d)$ is discrete in $C$; that is, it is equivalent to a finite collection of object tables. We define $J(d)$ to be the disjoint union $$J(d):=\coprod_{c\in F^\m1(d)}I(c).$$ %Note that under the global unique id assumption (Axiom \ref{dia:guid}), the disjoint union is given by a simple union.

Similarly, let $e\taking d\to d'$ be an arbitrary arrow in $D$. By Lemma \ref{lemma:discrete op-fibrations} we know that for each $c\in F^\m1(d)$ there is a unique equivalence class of paths $p_c$ in $C$ such that $F(p_c)\sim e$. Choose one, and compose it to a single binary table --- all other choices will result in the same result. Then define $J(e)$ to be the disjoint union $$J(e):=\coprod_{c\in F^\m1(d)}I(p_c).$$

\end{construction}

\begin{lemma}[Finiteness is preserved under $\Sigma$]

Suppose that $C$ and $D$ are finite categories and $F\taking C\to D$ is a functor. If $I\taking C\to\Set$ is a finite $C$-instance then $\Sigma_FI$ is a finite $D$ instance.

\end{lemma}

\begin{proof}

For each object or arrow $d$ in $D$, the table $\Sigma_FI(d)$ is a finite disjoint union of composition-joins of tables in $C$. The finite join of finite tables is finite, and the finite union of finite tables is finite.

\end{proof}

\begin{remark}[$\Sigma$ exists more generally and performs quotients and skolemization]\label{rmk:generalized sigma}

For any functor $F\taking C\to D$ the functor $\Delta_F\taking D\inst\to C\inst$ has a left adjoint, which we can denote by $\Sigma_F\taking C\inst\to D\inst$ because it agrees with the $\Sigma_F$ constructed in \ref{const:Sigma} in the case that $F$ is a discrete op-fibration. Certain queries are possible if we can use $\Sigma_F$ in this more general case---namely, quotienting by equivalence relations and the introduction of labeled nulls (a.k.a. Skolem variables). We do not consider it much in this paper for a few reasons. First, quotients and skolem variables are not part of the relational algebra. Second, the set of queries that include such quotients and skolem variables are not obviously closed under composition (see Section \ref{sec:future}).

\end{remark}

\begin{construction}[Limit as a kind of ``join all"]\label{const:join all}

Let $B$ be a signature and let $H$ be a $B$-instance. The functor $[H]\taking [B]\to\Set$ has a limit $\lim_BH\in\Ob(\Set)$, which can be computed as follows. Forgetting the path equivalence relations, axioms (\ref{dia:entity FBR}) and (\ref{dia:link FBR}) imply that $H$ consists of a set $\{N_1,\ldots,N_m\}$ of node tables, a set $\{e_1,\ldots,e_n\}$ of edge tables, and functions $s,t\taking\{1,\ldots,n\}\to\{1,\ldots,m\}$ such that for each $1\leq i\leq n$ the table $e_i$ constitutes a function $e_i\taking N_{s(i)}\to N_{t(i)}$. For each $i$ define $X_i$ as follows: 
\begin{align*}
X_0&:=\pi_{2,4,\ldots,2m}(N_1\cross\cdots\cross N_m)\\
X_1&:=\pi_{1,2,\ldots m}\sigma_{s(1)=m+1}\sigma_{t(1)=m+2}(X_0\cross e_1)\\
&\;\vdots\\
X_i&:=\pi_{1,2,\ldots m}\sigma_{s(i)=m+1}\sigma_{t(i)=m+2}(X_{i-1}\cross e_i)\\
&\;\vdots\\
X_n&:=\pi_{1,2,\ldots m}\sigma_{s(n)=m+1}\sigma_{t(n)=m+2}(X_{n-1}\cross e_n)
\end{align*}
Then $\ol{X}:=X_n$ is a table with $m$ columns, and its set $|\ol{X}|$ of rows (which can be constructed if one wishes by concatenating the fields in each row) is the limit, $|\ol{X}|\iso\lim_BH$. 

\end{construction}

\begin{construction}[$\Pi$]\label{const:Pi}

Let $F\taking C\to D$ be a signature morphism. We can succinctly define $\Pi_F\taking C\inst\to D\inst$ to be the right adjoint to $\Delta_F$, however the formula has an algorithmic description which we give now. Suppose we are given $F$ and an instance $I\taking C\to\Set$, considered as a collection of binary tables, one for each object and each edge in $C$. We are tasked with finding a $D$-instance, $J:=\Pi_FI\taking D\to\Set$.

We first define $J$ on an arbitrary object $d\taking\vect{0}\to D$ (see Example \ref{ex:chains}). Consider the comma category $B_d:=(d\down F)$ and its projection $q_d\taking (d\down F)\to C$. Note that in case $C$ and $D$ are finite, so is $B_d$. Let $H_d:=I\circ q_d\taking B_d\to\Set$. We define $J(d):=\lim_{B_d}H_d$ as in Construction \ref{const:join all}. 

Now let $e\taking d\to d'$ be an arbitrary arrow in $D$. For ease of notation, rewrite $$B:=B_d,\;\; B':=B_{d'},\;\; q:=q_d,\;\; q':=q_{d'},\;\; H:=H_d, \text{ and } H':=H_{d'}.$$ We have the following diagram of categories
$$
\xymatrix@=10pt{
(d'\down F)\ar[rr]^{(e\down F)}\ar[ddr]_{q'}\ar@/_1pc/[ddddr]_{H'}&&(d\down F)\ar[ddl]^{q}\ar@/^1pc/[ddddl]^{H}\\\\
&C\ar[dd]_-I\\\\
&\Set}
$$
A unique natural map $J(d)=\lim_BH\to\lim_BH'=J(d')$ is determined by the universal property for limits, but we give an idea of its construction. There is a map from the set of nodes in $B'$ to the set of nodes in $B$ and for each node $N$ in $B'$ the corresponding node tables in $H$ and $H'$ agree. Let $X_i$ and $X_i'$ be defined respectively for $(B,H)$ and $(B',H')$ as in Construction \ref{const:join all}. It follows that $X'_0$ is a projection (or column duplication) on $X_0$. The set of edges in $B'$ also map to the set of edges in $B$ and for each edge $e$ in $B'$ the corresponding edge tables in $H$ and $H'$ agree. Thus the select statements done to obtain $J(d')=\ol{X'}$ contains the set of select statements performed to obtain $J(d')=\ol{X}$. Therefore, the map from $J(d)$ to $J(d')$ is given as the inclusion of a subset followed by a projection.

\end{construction}

If $C$ or $D$ is not finite, then the right pushforward $\Pi_F$ of a finite instance $I\in C\set$ may have infinite, even uncountable results.

\begin{example}\label{ex:infinite results}

Consider the unique signature morphism $$\mcC=\fbox{$\LMO{s}$}\To{\;\;F\;\;}\LoopSchema=:\mcD.$$ Here $\church{\mcD}$ has arrows $\{f^n\|n\in\NN\}$ so it is infinite. In this case $(d\down F)$ is the discrete category with a countably infinite set of objects $\Ob(d\down F)\iso\NN$. 

Given the two-element instance $I\taking\mcC\to\Set$ with $I(s)=\{\sf{Alice},\tn{Bob}\}$, the rowset in the right pushforward $\Pi_F(I)$ is the (uncountable) set of infinite streams in $\{\tn{Alice},\tn{Bob}\}$, i.e. $$\Pi_F(I)(s)=I(s)^\NN.$$

\end{example}

\begin{proposition}\label{prop:limits are pushforwards}

Let $H\taking B\to\Set$ be a functor, and let $q\taking B\to\vect{0}$ be the terminal functor. Noting that there is an isomorphism of categories $\vect{0}\set\iso\Set$, we have a bijection $\lim_BH\iso\Pi_qH.$

\end{proposition}

\begin{proof}

Obvious by construction of $\Pi$.

\end{proof}

\begin{proposition}[Behavior of $\Delta,\Sigma,\Pi$ under natural transformations]\label{prop:dmf under nt}

Let $C$ and $D$ be categories, let $F,G\taking C\to D$ be functors, and let $\alpha\taking F\to G$ be a natural transformation as depicted in the following diagram:
$$\xymatrix{C\ar@/^1pc/[rr]^F\ar@/_1pc/[rr]_G\ar@{}[rr]|{\alpha\Downarrow}&&D}$$
Then $\alpha$ induces natural transformations $$\Delta_\alpha\taking\Delta_F\to\Delta_G,\hsp\Sigma_\alpha\taking\Sigma_G\to\Sigma_F,\hsp\text{and}\hsp\Pi_\alpha\taking\Pi_G\to\Pi_F.$$

\end{proposition}

\begin{proof}

For any instance $J\taking D\to\Set$ and any object $c\in\Ob(C)$ we have $\alpha_c\taking J\circ F(c)\to J\circ G(c)$, and the naturality of $\alpha$ implies that we can gather these into a natural transformation $\Delta_\alpha(J)\taking\Delta_F(J)\to\Delta_G(J)$. One checks easily that this assignment is natural in $J$, so we have $\Delta_\alpha\taking\Delta_F\to\Delta_G$ as desired.

Now suppose that $I\taking C\to\Set$ is an instance on $C$. Then for any $J\in D\set$ we have natural maps
$$\Hom(J,\Pi_GI)\iso\Hom(\Delta_GJ,I)\To{\Delta_\alpha}\Hom(\Delta_FJ,I)\iso\Hom(J,\Pi_FI)$$
so by the Yoneda lemma, we have a natural map $\Pi_G\to\Pi_F$, as desired. We also have natural maps
$$\Hom(\Sigma_FI,J)\iso\Hom(I,\Delta_FJ)\To{\Delta_\alpha}\Hom(I,\Delta_GJ)\iso\Hom(\Sigma_GI,J)$$
so by the Yoneda lemma, we have a natural map $\Sigma_G\to\Sigma_F$, as desired.

\end{proof}

\subsection{Relation to SQL queries}

Suppose we are working in a domain $DOM\in\Ob(\Set)$.

\begin{definition}\label{def:SQL bag set}
A {\em set-theoretic SQL query} $q$ is an expression of the form
\begin{tabbing}
\hsp\=SELECT DISTINCT\;\;\=$P$\\
\>FROM \>$(c_{1,1},\ldots,c_{1,k_1}),\ldots,(c_{n,1},\ldots,c_{n,k_n})$\\
\>WHERE\>$W$
\end{tabbing}
where $n,k_1,\ldots,k_n\in\NN$ are natural numbers, $C$ is a set of the form $C=\coprod_i\{c_{i,1},\ldots,c_{i,k_i}\}$, $W$ is a set of pairs $W\ss C\times C$, and $P\taking P_0\to C$ is a function, for some set $P_0$. 

A {\em bag-theoretic SQL query} $q'$ is an expression of the form
\begin{tabbing}
\hsp\=SELECT \;\;\=$P$\\
\>FROM \>$(c_{1,1},\ldots,c_{1,k_1}),\ldots,(c_{n,1},\ldots,c_{n,k_n})$\\
\>WHERE\>$W$
\end{tabbing}
where $n,k_1,\ldots,k_n, C, W, P, P_0$ are as above. We call $(n,k_1,\ldots,k_n, C, W, P, P_0)$ an {\em agnostic SQL query}.

Suppose that for each $1\leq i\leq n$ we are given a relation $R_i\ss \prod_{1\leq j\leq k_i}DOM$. 

The {\em evaluation of $q$} (respectively, the {\em evaluation of $q'$}) on $R_1,\ldots,R_n$ is a relation $Q(R_1,\ldots,R_n)\ss \prod_{i\in P_0}R_i$ (respectively, a function $Q'(R_1,\ldots,R_n)\taking B'\to \prod_{i\in P_0}R_i$ for some set $B'$), defined as follows. Let $$B'=\{r\in R_1\times\cdots\times R_n\|r.c_1=r.c_2\tn{ for all }(c_1,c_2)\in W\}.$$ We compose an inclusion with a projection to get a function 
$$B'\ss R_1\times\cdots\times R_n\To{P} \prod_{i\in P_0}R_i.$$ 
and define $Q'(R_1,\ldots,R_n)$ to be this function. Its image is the desired relation $Q(R_1,\ldots,R_n)\ss\prod_{i\in P_0}R_i$.

\end{definition}

\begin{definition}

Let $q=(n,k_1,\ldots,k_n,C,W,P,P_0)$ be an agnostic SQL query as in Definition \ref{def:SQL bag set}. The set $W$ generates an equivalence relation $\sim$ on $C$; let $C':=C/\sim$ be the quotient and $\ell\taking C\to C'$ the induced function. Suppose that $C'$ has $r=|C'|$ many elements and that $P_0$ has $s=|P_0|$ many elements.

We define the {\em categorical setup for $q$}, denoted $CS(q)=\mcS\To{F}\mcT\From{G}\mcU$ as follows.
$$
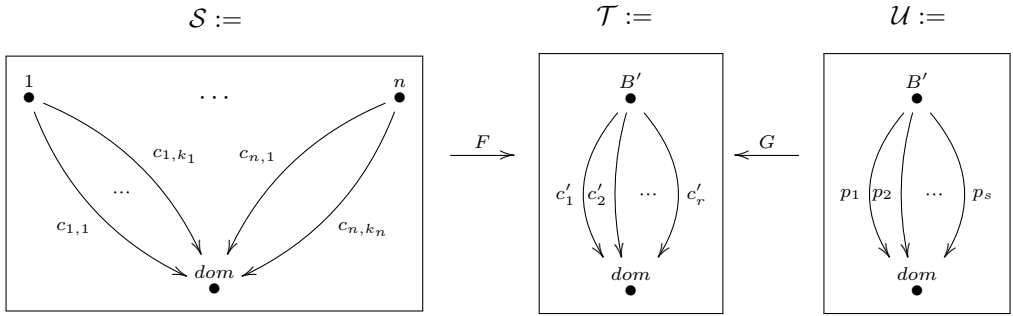
\parbox{2.2in}{\begin{center}$\mcS:=$\end{center}\fbox{\xymatrix{
\LMO{1}\ar@/_1pc/[ddrr]_{c_{1,1}}\ar@/^1pc/[ddrr]^{c_{1,k_1}}\ar@{}[ddrr]|\cdots&&\cdots&&\LMO{n}\ar@/_1pc/[ddll]_{c_{n,1}}\ar@/^1pc/[dlld]^{c_{n,k_n}}\\\\
&&\LMO{dom}}
}}
\parbox{.45in}{\xymatrix{~\ar[r]^F&~}}
\parbox{.9in}{\begin{center}$\mcT:=$\end{center}\fbox{\xymatrix{
\LMO{B'}\ar@/_1.5pc/[dd]_{c'_1}\ar@/_.5pc/[dd]_{c'_2}\ar@{}[dd]^\cdots\ar@/^1.5pc/[dd]^{c'_r}\\\\
\LMO{dom}
}}}
\parbox{.45in}{\xymatrix{~&~\ar[l]_G}}
\parbox{1in}{\begin{center}$\mcU:=$\end{center}\fbox{\xymatrix{
\LMO{B'}\ar@/_1.5pc/[dd]_{p_1}\ar@/_.5pc/[dd]_{p_2}\ar@{}[dd]^\cdots\ar@/^1.5pc/[dd]^{p_s}\\\\
\LMO{dom}
}}}
$$
Here, the functor $F$ sends $dom$ to $dom$, sends the other objects to $B$, and acts on morphisms by $\ell\taking C\to C'$; the functor $G$ sends $dom$ to $dom$, $B'$ to $B'$, and $P_0\to C'$ by the composite $P_0\To{P}C\To{\ell}C'$. 

We define the {\em categorical analogue of $q$}, denoted $CA(q)$, to be the FQL query $\Delta_G\Pi_F$.\footnote{The functor $\Delta_G\Pi_F$ is technically not in the correct $\Sigma\Pi\Delta$ order, but it is equivalent to a query in that order by Theorem \ref{thm:query comp}.}

\end{definition}

\begin{proposition}\label{prop:conjunctive RA bag}

Let $q'=(n,k_1,\ldots,k_n,C,W,P,P_0)$ be an agnostic SQL query, let $Q(R_1,\ldots,R_n)$ be its set-theoretic evaluation, and let $Q'(R_1,\ldots,R_n)$ be its bag-theoretic evaluation as in Definition \ref{def:SQL bag set}. Let $\mcS\To{F}\mcT\From{G}\mcU$ be the categorical setup for $q$. Suppose that for each $1\leq i\leq n$ we are given a relation $R_i\ss \prod_{1\leq j\leq k_i}DOM$. Let $I\taking\mcS\to\Set$ be the corresponding functor (i.e., for each $1\leq i\leq n$ let $I(i)=R_i$, let $I(dom)=DOM$, and let $c_{i,j}\taking R_i\to DOM$ be the projection). 

Then the categorical analogue $CA(q)(I)=\Delta_G\Pi_F(I)$ is equal to the bag-theoretic evaluation $Q'(R_1,\ldots,R_n)$, and the relationalization $REL(CA(q)(I))$ (see Definiton \ref{def:dedup}) is equal to the set-theoretic evaluation $Q(R_1,\ldots,R_n)$.

\end{proposition}

\begin{proof}

The first claim follows from Construction \ref{const:Pi}. The second claim follows directly, since both sides are simply images.

\end{proof}
\begin{proposition}\label{ghjkl}
Let $R$ and $R'$ be union-compatible tables.  We can express the disjoint union $R + R'$ as $\Sigma_F$, where $F$ sends $R$ and $R'$ to $S$, $c_i$ and $c_i'$ to $d$, and $D$ to $D$ in the following diagram :
$$
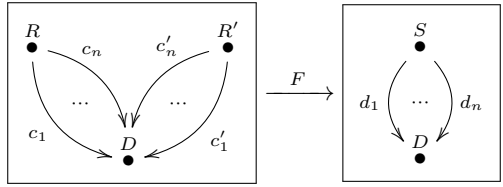
\fbox{\parbox{1in}{\xymatrix{\LMO{R}\ar@/_1pc/[dr]_{c_1}\ar@{}[dr]|{\cdots}\ar@/^1pc/[dr]^(.3){c_n}&&\LMO{R'}\ar@/_1pc/[dl]_(.3){c'_{n}}\ar@{}[dl]|{\cdots}\ar@/^1pc/[dl]^{c'_1}\\&\LMO{D}}}}
\Too{F}
\fbox{\parbox{1in}{\xymatrix{\LMO{S}\ar@/_1pc/[d]_{d_1}\ar@{}[d]|{\cdots}\ar@/^1pc/[d]^{d_n}\\\LMO{D}}}}
$$
\end{proposition}
\begin{proof}
This can be computed directly from Construction~\ref{const:Sigma}, noting that the inverse image of $S$ (resp. $d_i$, and $D$) is $\{R,R'\}$ (resp. $\{c_i,c_i'\}$, and $\{D\}$).
\end{proof}

\subsection{Query composition}

In this section and those that follow we were greatly inspired by, and follow closely, the work of \cite{GK}. While their work does not apply directly, the adaptation to our context is fairly straightforward.

\begin{lemma}[Unit and counit for $\Sigma$ are Cartesian]\label{lemma:cartesian}

Let $F\taking C\to D$ be a discrete op-fibration, and let 
$$\eta\taking\id_{C\inst}\to\Delta_F\Sigma_F\hsp\text{and}\hsp \epsilon\taking\Sigma_F\Delta_F\to\id_{D\inst}$$
be (respectively) the unit and counit of the $(\Sigma_F,\Delta_F)$ adjunction. Each is a Cartesian natural transformation. In other words, for any morphism $a\taking I\to I'$ of $C$-instances and for any morphism $b\taking J\to J'$ of $D$-instances, each of the induced naturality squares (left below for the unit and right below for the counit) 
$$
\parbox{1in}{
\xymatrix@=30pt{I\ar[r]^{a}\ar[d]_{\eta_I}\ullimit&I'\ar[d]^{\eta_{I'}}\\
\Delta_F\Sigma_F(I)\ar[r]_{\Delta_F\Sigma_F(a)}&\Delta_F\Sigma_F(I')}
}
\hspace{1in}
\parbox{1in}{
\xymatrix@=30pt{
\Sigma_F\Delta_F(J)\ar[r]^{\Sigma_F\Delta_F(b)}\ar[d]_{\epsilon_J}\ullimit&\Sigma_F\Delta_F(J')\ar[d]^{\epsilon_{J'}}\\
J\ar[r]_b&J'}
}
$$
is a pullback in $C\inst$ and $D\inst$ respectively.

\end{lemma}

\begin{proof}

It suffices to check that for an arbitrary object $c\in\Ob(C)$ and $d\in\Ob(D)$ respectively, the induced commutative diagram in $\Set$ 
$$
\parbox{1in}{
\xymatrix@=35pt{I(c)\ar[r]^{a_c}\ar[d]_{(\eta_I)_c}&I'(c)\ar[d]^{(\eta_{I'})_c}\\
\Delta_F\Sigma_F(I)(c)\ar[r]_{\Delta_F\Sigma_F(a)_c}&\Delta_F\Sigma_F(I')(c)}
}
\hspace{.8in}
\parbox{1in}{
\xymatrix@=35pt{
\Sigma_F\Delta_F(J)(d)\ar[r]^{\Sigma_F\Delta_F(b)_d}\ar[d]_{(\epsilon_J)_d}\ullimit&\Sigma_F\Delta_F(J')(d)\ar[d]^{(\epsilon_{J'})_d}\\
J(d)\ar[r]_{b_d}&J'(d)}
}
$$
is a pullback. Unpacking definitions, these are:
$$
\parbox{2.7in}{
\xymatrix@=30pt{I(c)\ar[r]^{a_c}\ar[d]_{c'=c\;\;}&I'(c)\ar[d]^{\;\;c'=c}\\
\parbox{.7in}{\vspace{-.1in}$$\coprod_{\{c'\|F(c')=F(c)\}}$$\vspace{-.2in}}I(c')\ar[r]_{\amalg a_{c'}}&\parbox{.7in}{\vspace{-.1in}$$\coprod_{\{c'\|F(c')=F(c)\}}$$\vspace{-.2in}}I'(c')}
}
\hspace{.8in}
\parbox{2.2in}{
\xymatrix@=30pt{
\parbox{.6in}{$$\coprod_{\{c\|F(c)=d\}}$$\vspace{-.1in}}J(d)\ar[r]^{\amalg b_d}\ar[d]&\parbox{.6in}{$$\coprod_{\{c\|F(c)=d\}}$$\vspace{-.1in}}J'(d)\ar[d]\\
J(d)\ar[r]_{b_d}&J'(d)}~\\\vspace{.15in}
}\vspace{-.2in}
$$
Roughly, the fact that these are pullbacks squares follows from the way coproducts work (e.g. the disjointness property) in $\Set$, or more precisely it follows from the fact that $\Set$ is a positive coherent category \cite[p. 34]{Johnstone:MR1953060}.

\end{proof}

\begin{definition}[Grothendieck construction]

Let $C$ be a signature and $I\taking C\to\Set$ an instance. The {\em Grothendieck category of elements for $I$ over $C$} consists of a pair $(\int_CI,\pi_I)$, where $\int_CI$ is a signature and $\pi_I\taking\int_CI\to C$ is a signature morphism, constructed as follows. The set of nodes in $\int_CI$ is $\{(c,x)\|c\in\Ob(C), x\in I(c)\}$. The set of edges in $\int_CI$ from node $(c,x)$ to node $(c',x')$ is $\{e\taking c \to c'\|I(e)(x)=x'\}$. The signature morphism $\pi_I\taking\int_CI\to C$ is obvious: send $(c,x)$ to $c$ and send $e$ to $e$. Two paths are equivalent in $\int_CI$ if and only if their images under $\pi_I$ are equivalent. We sometimes denote $\int_C$ simply as $\int$.

\end{definition}

\begin{lemma}\label{lemma:grothendieck is dop}

Let $I\taking C\to\Set$ be an instance. Then $\pi_I\taking\int I\to C$ is a discrete op-fibration. 

\end{lemma}

\begin{proof}

This follows by Lemma \ref{lemma:discrete op-fibrations}.

\end{proof}

\begin{lemma}[Acyclicity and finiteness are preserved under the Grothendieck construction]

Suppose that $C$ is an acyclic (respectively, a finite) signature and that $I\taking C\to\Set$ is a finite instance. Then $\int I$ is an acyclic (respectively, a finite) signature.

\end{lemma}

\begin{proof}

Both are obvious by construction. 

\end{proof}

\begin{definition}[DeGrothendieckification]

Let $\pi\taking X\to C$ be a discrete op-fibration. Let $\singleton^X$ denote a terminal object in $X\inst$, i.e. any instance in which every node table and edge table consists of precisely one row. Define {\em the deGrothendieckification of $\pi$}, denoted $\partial\pi\taking C\to\Set$ to be $\partial\pi:=\Sigma_\pi\left(\singleton^X\right)\in C\inst$. 

\end{definition}

One checks that for a discrete op-fibration $\pi\taking X\to C$ and object $c\in\Ob(C)$ we have the formula $$\partial\pi(c)=\pi^\m1(c),$$ so we can say that deGrothendieckification is given by pre-image.

\begin{lemma}[Finiteness is preserved under DeGrothendiekification]

If $X$ is a finite signature and $\pi\taking X\to C$ is any discrete op-fibration, then $\partial\pi$ is a finite $C$-instance.

\end{lemma}

\begin{proof}

Obvious.

\end{proof}

\begin{proposition}\label{prop:on dopf}

Given a signature $C$, let $Dopf_C\ss\Cat_{/C}$ denote the full category spanned by the discrete op-fibrations over $C$. Then $\int\taking C\inst\to Dopf_C$ and $\partial\taking Dopf_C\to C\inst$ are functorial, $\partial$ is left adjoint to $\int$, and they are mutually inverse equivalences of categories. 

\end{proposition}

\begin{proof}

See \cite[Lemma 2.3.4, Proposition 3.2.5]{Spivak:1202.2591}.

\end{proof}

\begin{corollary}

Suppose that $C$ and $D$ are categories, that $F,G\taking C\to D$ are discrete op-fibrations, and $\alpha\taking F\to G$ is a natural transformation. Then there exists a functor $p\taking C\to C$ such that $F\circ p=G$, and $\Sigma_\alpha=\partial(p)$, where $\Sigma_\alpha\taking\Sigma_G\to\Sigma_F$ is the natural transformation given in Proposition \ref{prop:dmf under nt}.

\end{corollary}

\begin{proof}

By Proposition \ref{prop:dmf under nt} the natural transformation $\alpha\taking F\to G$ induces a natural transformation $\Sigma_G\to\Sigma_F$, and applying deGrothendieckification supplies a map $\partial G\to\partial F$. By Proposition \ref{prop:on dopf} this induces a map $p\taking C\to C$ over $D$ with the above properties.

\end{proof}

\begin{proposition}\label{prop:functor between dopfs is dopf}

Given a commutative diagram
$$
\xymatrix@=10pt{A\ar[rr]^f\ar[rdd]_h&&B\ar[ldd]^g\\\\&C}
$$
in which $g$ and $h$ are discrete op-fibrations, it follows that $f$ is also a discrete op-fibration.

\end{proposition}

\begin{proof}

Consider the diagram to the left below
$$
\xymatrix{\vect{0}\ar[r]^a\ar[d]_i&A\ar[d]^f\\\vect{1}\ar[r]_q\ar@{=}[d]&B\ar[d]^g\\\vect{1}\ar[r]_{gq}&C}
\hspace{.8in}
\xymatrix{\vect{0}\ar[r]^a\ar[d]_i&A\ar[d]^h\\\vect{1}\ar[ur]^{\ell}\ar[r]_{gq}&C}
\hspace{.8in}
\xymatrix@=30pt{\vect{0}\ar[r]^{fa}\ar[d]_i&B\ar[d]^g\\\vect{1}\ar@<.5ex>[ur]^q\ar@<-.5ex>[ur]_{f\ell}\ar[r]_{gq}&C}
$$
Then since $h$ is a discrete op-fibration there exists a unique $\ell\taking \vect{1}\to A$ making the middle diagram commute. But now we have two lifts, $q$ and $f\ell$, and since $g$ is a discrete op-fibration, it follows that $q=f\ell$, completing the proof.

\end{proof}

\begin{proposition}[Discrete op-fibrations are stable under pullback]\label{prop:dop and pullback}

Let $F\taking C'\to C$ be a functor and $\pi\taking X\to C$ be a discrete op-fibration. Then given the pullback square
$$
\xymatrix{C'\cross_CX\ar[r]\ar[d]_{\pi'}\ullimit&X\ar[d]^\pi\\C'\ar[r]_F&C}
$$
the map $\pi'$ is a discrete op-fibration, and there is a natural isomorphism $$\pi'\iso\int_{C'}\Delta_F(\partial\pi).$$

\end{proposition}

\begin{proof}

A functor $p\taking Y\to D$ is a discrete op-fibration if and only there exists a functor $d\taking D\to\Set$ such that the diagram 
$$
\xymatrix{Y\ar[r]\ar[d]_{p}\ullimit&\Set_*\ar[d]\\D\ar[r]_d&\Set}
$$
is a pullback square, where $\Set_*$ is the category of pointed sets as in Example \ref{ex:slice}. The result follows from the pasting lemma for fiber products, and Lemma \ref{lemma:grothendieck is dop}.

\end{proof}

\begin{lemma}[Comparison morphisms for squares]\label{lemma:comparison morphism}

Suppose given the following diagram of categories:
$$\xymatrix{
R\ar[r]^f\ar[d]_e\ar@{}[dr]|{\stackrel{\alpha}{\Nearrow}}&S\ar[d]^h\\T\ar[r]_g&U
}
$$
Then there are natural transformations of functors $S\inst\to T\inst$:
$$\Sigma_f\Delta_e\too\Delta_h\Sigma_g\hsp\tn{and}\hsp\Delta_g\Pi_h\too\Pi_e\Delta_f.$$
If $hf=ge$ and $\alpha=\id$ (i.e. if the diagram commutes) then, by symmetry, there are natural transformations of functors $T\inst\to S\inst$:
$$\Sigma_e\Delta_f\too\Delta_g\Sigma_h\hsp\tn{and}\hsp\Delta_h\Pi_g\too\Pi_f\Delta_e.$$

\end{lemma}

\begin{proof}

These arise from units and counits, together with Proposition \ref{prop:dmf under nt}.
\begin{align*}
%\Sigma_f\Delta_e\To{\eta_h}\Delta_h\Sigma_h\Sigma_f\Delta_e\To{\Sigma_\alpha}\Delta_h\Sigma_g\Sigma_e\Delta_e\To{\epsilon_e}\Delta_h\Sigma_g\\
\Sigma_f\Delta_e\To{\eta_g}\Sigma_f\Delta_e\Delta_g\Sigma_g\To{\Delta_\alpha}\Sigma_f\Delta_f\Delta_h\Sigma_g\To{\epsilon_f}\Delta_h\Sigma_g\\
%\Delta_g\Pi_h\To{\eta_f}\Delta_g\Pi_h\Pi_f\Delta_f\To{\Pi_\alpha}\Delta_g\Pi_g\Pi_e\Delta_f\To{\epsilon_g}\Pi_e\Delta_f.\\
\Delta_g\Pi_h\To{\eta_e}\Pi_e\Delta_e\Delta_g\Pi_h\To{\Delta_\alpha}\Pi_e\Delta_f\Delta_h\Pi_h\To{\epsilon_h}\Pi_e\Delta_f.
\end{align*}
The second claim is symmetric to the first if $\alpha=\id$. 
\end{proof}

\begin{definition}

Suppose given a diagram of the form 
$$\xymatrix{
R\ar[r]^f\ar[d]_e\ar@{}[dr]|{\stackrel{\alpha}{\Nearrow}}&S\ar[d]^h\\T\ar[r]_g&U
}
$$
We say it is {\em exact} if the comparison morphism $\Sigma_f\Delta_e\to\Delta_h\Sigma_g$ is an isomorphism. Note that this is the case if and only if the comparison morphism $\Delta_g\Pi_h\to\Pi_e\Delta_f$ is an isomorphism.

\end{definition}

%\begin{proposition}\label{prop:complexity of dopf}
%
%The complexity of checking that a functor is a discrete op-fibration is a combination of the complexity of checking equivalence of categories and the complexity of the word problem for categories (which is r.e.).
%
%\end{proposition}
%
%\begin{proof}
%
%Any equivalence of categories is a discrete op-fibration, so one containment is trivial. Suppose that $F\taking C\to D$ is a signature morphism. Let $C'$ be the signature with the same objects and arrows as $C$ and with all the path equivalences in $C$, plus a path equivalence $p\sim p'$ whenever $F(p)\sim F(p')$ in $D$. Then $F$ factors as a composition $$\xymatrix{C\ar[r]^{i}\ar@/_1pc/[rr]_F&C'\ar[r]^{F'}&D}$$ and $F$ is a discrete op-fibration if and only if $i\taking C\to C'$ is an equivalence of categories. %Thus the complexity of checking that $F$ is a discrete op-fibration is the sum of the complexity of checking equivalence of categories and the complexity of the word problem, which is r.e.
%
%\end{proof}

\begin{proposition}[Comma Beck-Chevalley for $\Sigma,\Delta$]\label{prop:BC Sigma}

Let $F\taking C\to D$ and $G\taking E\to D$ be functors, and consider the canonical natural transformation diagram (see Definition \ref{def:comma})
\begin{align}\label{dia:comma sigma delta exact}
\xymatrix{(F\down G)\ar[r]^q\ar[d]_p\ar@{}[dr]|{\stackrel{\alpha}{\Nearrow}}&E\ar[d]^G\\C\ar[r]_F&D}
\end{align}
Let $\Sigma_F$ be the generalized left push-forward as defined in Remark \ref{rmk:generalized sigma}. Then the comparison morphism (Lemma \ref{lemma:comparison morphism}) is an isomorphism
$$\Sigma_q\Delta_p\To{\iso}\Delta_G\Sigma_F$$
 of functors $C\inst\to E\inst$. In other words, (\ref{dia:comma sigma delta exact}) is exact.

\end{proposition}

\begin{proof}

The map is given by the composition
$$
\xymatrix{\Sigma_q\Delta_p\ar[r]^-{\eta_F}&\Sigma_q\Delta_p\Delta_F\Sigma_F\ar[r]^{\alpha}&\Sigma_q\Delta_q\Delta_G\Sigma_F\ar[r]^-{\epsilon_q}&\Delta_G\Sigma_F}
$$
To prove it is an isomorphism, one checks it on each object $e\in\Ob(E)$ by proving that the diagonal map $(F\down Ge)\to(q\down e)$ in the diagram 
$$
\xymatrix{
(F\down Ge)\ar@/^1pc/[rrd]\ar[rd]\ar@/_1pc/[rdd]\\
&(q\down e)\ar[r]\ar[d]\ar@{}[rd]|{\Nearrow}&[0]\ar[d]^e\\
&(F\down G)\ar@{}[dr]|{\Nearrow}\ar[d]_p\ar[r]^q&E\ar[d]^G\\
&C\ar[r]_F&D}
$$
is a final functor, and this is easily checked.

\end{proof}

\begin{corollary}[Comma Beck-Chevalley for $\Delta,\Pi$]\label{cor:BC Pi}

Let $F\taking C\to D$ and $G\taking E\to D$ be functors, and consider the canonical natural transformation diagram (\ref{dia:nt for comma})
$$
\xymatrix{(F\down G)\ar[r]^q\ar[d]_p\ar@{}[dr]|{\stackrel{\alpha}{\Nearrow}}&E\ar[d]^G\\C\ar[r]_F&D}
$$
Then the comparison morphism (Lemma \ref{lemma:comparison morphism}) is an isomorphism $$\Delta_F\Pi_G\To{\iso}\Pi_p\Delta_q$$ of functors $E\inst\to C\inst$.

\end{corollary}

\begin{proof}

This follows from Proposition \ref{prop:BC Sigma} by adjointness.

\end{proof}

Sometimes when we push forward an instance $I$ along a functor $B\To{G}C$, the resulting instance $\Sigma_GI$ agrees with $I$, at least on some subcategory $A\ss B$. This is very useful to know, because it gives an easy way to calculate row sets and name rows in the pushforward. The following lemma roughly says that this occurs if for all $b\in\Ob(B)$ and $a\in\Ob(A)$, the sense to which $b$ is to the left of $a$ in $C$, is the same as the sense to which $b$ is to the left of $a$ in $B$.

\begin{lemma}

Suppose that we are given a diagram of the form 
$$
\xymatrix{A\ar[r]^q\ar@{=}[d]&B\ar[d]^G\\A\ar[r]_F&C}
$$
Then the map $\beta\taking\Delta_q\to\Delta_F\Sigma_G$ is an isomorphism if, for all objects $a\in\Ob(A)$ the functor 
$$\ol{G}\taking(\id_B\down qa)\to (G\down Fa)$$ 
is final.

\end{lemma}

\begin{proof}

For typographical convenience, let $[0]$ denote the terminal category, usually denoted $\vect{0}$. To check that $\beta$ is an isomorphism, we may choose an arbitrary object $a\taking[0]\to A$ and check that $\Delta_a\beta$ is an isomorphism. Consider the diagram
$$
\xymatrix{
(G\down Fa)\ar@/^1pc/[rrrd]^v\ar@/_.5pc/[rdd]_u\\
&(\id_A\down a)\ar[r]^-p\ar[d]_t\ar[ul]_{G'}\ar@{}[dr]|{\Swarrow}&A\ar@{=}[d]\ar[r]^(.4)q&B\ar[d]^G\\
&[0]\ar[r]_a&A\ar[r]_F&C}
$$
The map $\beta$ is the composite 
$$\Delta_a\Delta_q\iso\Sigma_t\Delta_p\Delta_q=\Sigma_u\Sigma_{G'}\Delta_{G'}\Delta_v\To{\;\beta'\;}\Sigma_u\Delta_v\iso\Delta_a\Delta_F\Sigma_G,$$
and this is equivalent to showing that $\beta'$ is an isomorphism. It suffices to show that $\Sigma_t\to\Sigma_u\Delta_v$ is an isomorphism, but this is equivalent to the assertion that $G'$ is a final functor. One can factor $G'$ as 
$$
(\id_A\down a)\To{\ol{q}}(\id_B\down qa)\To{\ol{G}}(G\down Fa)
$$
and the first of these is easily checked to be final. Thus the composite is final if and only if the latter map is final, as desired.

\end{proof}

\begin{proposition}[Pullback Beck-Chevalley for $\Sigma,\Delta$ when $\Sigma$ is along a discrete op-fibration]\label{prop:pb beck sigma}

Suppose that $F\taking D\to C$ is a functor and $p\taking X\to C$ is a discrete op-fibration, and form the pullback square
\begin{align}\label{dia:pullback sigma delta exact}
\xymatrix{Y\ullimit\ar[r]^G\ar[d]_{q}&X\ar[d]^p\\D\ar[r]_F&C}
\end{align}
The comparison morphism (Lemma \ref{lemma:comparison morphism}) is a natural isomorphism $$\Sigma_q\Delta_G\To{\iso}\Delta_F\Sigma_p$$ of functors $X\inst\to D\inst$. In other words, (\ref{dia:pullback sigma delta exact}) is exact.

\end{proposition}

\begin{proof}

Consider the morphism given in Lemma \ref{lemma:comparison morphism}.
Since $\eta_p$ and $\epsilon_q$ are Cartesian (by Lemma \ref{lemma:cartesian}), it suffices to check the claim on the terminal object of $X\inst$, where it follows by Proposition \ref{prop:dop and pullback}.

\end{proof}

\begin{corollary}[Pullback Beck-Chevalley for $\Delta,\Pi$ when $\Delta$ is along a discrete op-fibration]\label{cor:pb bc DP dop}

Suppose that $F\taking D\to C$ is a functor and $p\taking X\to C$ is a discrete op-fibration, and form the pullback square
\begin{align}\label{dia:pullback delta pi exact}
\xymatrix{Y\ullimit\ar[r]^G\ar[d]_{q}&X\ar[d]^p\\D\ar[r]_F&C}
\end{align}
The comparison morphism (Lemma \ref{lemma:comparison morphism}) is an isomorphism $$\Delta_p\Pi_F\To{\iso}\Pi_G\Delta_q$$ of functors $D\inst\to X\inst$. In other words, (\ref{dia:pullback delta pi exact}) is exact.

\end{corollary}

\begin{proof}

This follows from Proposition \ref{prop:pb beck sigma} by adjointness. 

\end{proof}

\begin{proposition}[Distributive law]\label{prop:distributive}

Let $u\taking C\to B$ be a discrete op-fibration and $f\taking B\to A$ any functor. Construct the {\em distributivity diagram}
\begin{align}\label{dia:dd}
\xymatrix@=15pt{&N\ullimit\ar[ld]_e\ar[rr]^g\ar[dd]^w&&M\ar[dd]^v\\
C\ar[dr]_u&&\\
&B\ar[rr]_f&&A}
\end{align}
as follows. Form $v\taking M\to A$ by $v:=\int_A\Pi_f(\partial u)$, form $w\taking N\to B$ by $w:=\int_B\Delta_f\Pi_f\partial u$. Finally let $e\taking N\to C$ be given by $\int_B\epsilon_f(\partial u)$, where $\epsilon_f\taking\Delta_f\Pi_f\to\id_B$ is the counit. Note that $N\iso B\cross_AM$ by Proposition \ref{prop:dop and pullback} and that one has $w=u\circ e$ by construction. 

Then there is a natural isomorphism 
\begin{align}\label{dia:Pi past Sigma}
\Sigma_v\Pi_g\Delta_e\To{\iso}\Pi_f\Sigma_u
\end{align}
of functors $C\inst\to A\inst$.

\end{proposition}

\begin{proof}

We have that $u,v,$ and $w$ are discrete op-fibrations. By Corollary \ref{cor:pb bc DP dop} we have the following chain of natural transformations 
\begin{align*}
\Sigma_v\Pi_g\Delta_e\To{\eta_u}\Sigma_v\Pi_g\Delta_e\Delta_u\Sigma_u=\Sigma_v\Pi_g\Delta_w\Sigma_u\To{\iso}\Sigma_v\Delta_v\Pi_f\Sigma_u\To{\epsilon_v}\Pi_f\Sigma_u.
\end{align*}
Every natural transformation in the chain is Cartesian, so it suffices to check that the composite is an isomorphism when applied to the terminal object $\singleton^C$ in $C\inst$. But there the composition is simply the identity transformation on $\Pi_f(\partial u)$, proving the result.

\end{proof}

\begin{remark}\label{rem:interpretation of distributive law}

The distributive law above takes on a much simpler form when we realize that any discrete opfibration over $B$ is the category of elements of some functor $I\taking B\to\Set$. We have an equivalence of categories $\int(I)\set\iso B\set_{/I}$. Thinking about the distributive law in these terms is quite helpful. 

Let $f\taking B\to A$ be a functor and consider the functor
$$B\set_{/I}\Too{``\Pi_f"}A\set_{/\Pi_fI}\;\;
\footnote{The quotes around $``\Pi_f"$ indicate that this is not actually a right Kan extension, but that $``\Pi_f"$ is a suggestive name for the functor we indicate.}$$
given by $(J\to I)\mapsto(\Pi_fJ\to\Pi_fI)$. Let $u\taking\int I\to B$ and $v\taking\int(\Pi_fI)\to A$ be the canonical projections. Form the distributivity diagram
$$\xymatrix{
\int\Delta_f\Pi_fI\ullimit\ar[d]_e\ar[r]^g\ar[d]_e&\int\Pi_fI\ar[dd]^v\\
\int I\ar[d]_u&\\
B\ar[r]_f&A
}
$$
Then the following diagram of categories commutes:
$$\xymatrix{
B\set_{/I}\ar[r]^{``\Pi_f"}\ar[d]_{\iso}&A\set_{/\Pi_fI}\ar[d]^\iso\\
(\int I)\set\ar[r]_{\Pi_g\Delta_e}\ar[d]_{\Sigma_u}&(\int\Pi_fI)\set\ar[d]^{\Sigma_v}\\
B\set\ar[r]_{\Pi_f}&A\set
}
$$
where the vertical composites are the ``forgetful" functors. In this diagram, the big rectangle clearly commutes. The distributive law precisely says that the bottom square commutes. Deducing that the top square commutes, we realize that the rather opaque looking $\Pi_g\Delta_e$ is quite a simple functor.

\end{remark}

\begin{lemma}\label{lemma:lots of pullbacks}

Suppose given a pullback square. 
$$\xymatrix{
R\ullimit\ar[r]^f\ar[d]_e&S\ar[d]^h\\T\ar[r]_g&U
}
$$
For any functor $\Gamma\taking T\to\Set$, every square in the following diagram is a pullback:
$$
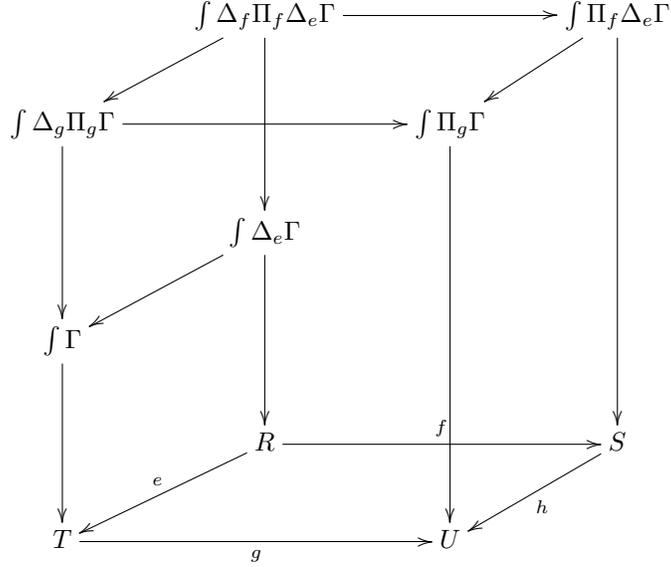
\xymatrix{
&\int\Delta_f\Pi_f\Delta_e\Gamma\ar[dd]\ar[rr]\ar[dl]&&\int\Pi_f\Delta_e\Gamma\ar[dddd]\ar[dl]\\
\int\Delta_g\Pi_g\Gamma\ar[dd]\ar[rr]&&\int\Pi_g\Gamma\ar[dddd]\\
&\int\Delta_e\Gamma\ar[dd]\ar[dl]\\
\int\Gamma\ar[dd]\\
&R\ar[rr]^f\ar[ld]_e&&S\ar[ld]^h\\
T\ar[rr]_g&&U
}
$$

\end{lemma}

\begin{proof}

While the format of this diagram contains a couple copies of the distributivity diagram, this result has nothing to do with that one. Indeed it follows from Proposition \ref{prop:dop and pullback} and Corollary \ref{cor:pb bc DP dop}. The bottom square is a pullback by hypothesis. The lower left-hand square, the front square, and the back square are each pullbacks by the proposition. It now follows that the top square is a pullback. The right-hand square is a pullback by a combination of the proposition and the corollary. It now follows that the upper left-hand square is a pullback.

\end{proof}

\begin{example}[Why $\Sigma$-restrictedness appears to be necessary for composability of queries]\label{ex:why dist needs dopf}

Without the condition that $u$ be a discrete op-fibration, Proposition \ref{prop:distributive} does not hold. Indeed, let $C=\Loop$ as in (\ref{dia:loop}) and let $B=A=\vect{0}$ be the terminal category. One finds that $N=M=\vect{0}$ too in (\ref{dia:dd}), and $e$ is the unique functor. If one considers $C$-sets as special kinds of graphs (discrete dynamical systems) then we can say that $\Pi_f\Sigma_u$ will extract the set of connected components for a $C$-set, whereas $\Sigma_v\Pi_g\Delta_e$ will extract its set of vertices.

We may still wish to ask ``whether a $\Pi$ can be pushed past a $\Sigma$", i.e. whether for any $C\To{u}B\To{f}A$, some appropriate $v,g,e$ exist such that an isomorphism as in (\ref{dia:Pi past Sigma}) holds. The following example may help give intuition. Let $C\To{u}B\To{f}A$ be given as follows
$$
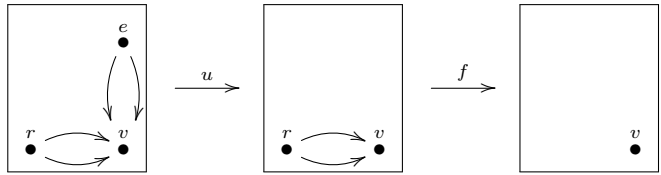
\parbox{.75in}{
\fbox{\xymatrix{&\LMO{e}\ar@/_.5pc/[d]\ar@/^.5pc/[d]\\\LMO{r}\ar@/^.5pc/[r]\ar@/_.5pc/[r]&\LMO{v}}}}
\xymatrix{~\ar[r]^u&~}
\parbox{.75in}{
\fbox{\xymatrix{&\color{white}{\LMO{e}}\\\LMO{r}\ar@/^.5pc/[r]\ar@/_.5pc/[r]&\LMO{v}}}}
\xymatrix{~\ar[r]^f&~}
\parbox{.75in}{
\fbox{\xymatrix{&\color{white}{\LMO{e}}\\\color{white}{\LMO{r}}&\LMO{v}}}}
$$
where $u(e)=v$ and $f(r)=v$. The goal is to find some $C\From{e}N\To{g}M\To{v}A$ such that isomorphism (\ref{dia:Pi past Sigma}) holds. This does not appear possible if $M$ and $N$ are assumed finitely presentable.

\end{example}

\begin{theorem}[Query composition]\label{thm:query comp}

Suppose that one has $\Sigma$-restricted data migration queries $Q\taking S\queryto T$ and $Q'\taking T\queryto U$ as follows:
\begin{align}\label{dia:composable queries}
\xymatrix{&B\ar[r]^f\ar[dl]_s&A\ar[dr]^t&&&D\ar[r]^g\ar[dl]_u&C\ar[dr]^v\\
S\ar@{}[rrr]|Q&&&T&T\ar@{}[rrr]|{Q'}&&&U}
\end{align}
Then there exists a $\Sigma$-restricted data migration query $Q''\taking S\queryto U$ such that $\church{Q''}\iso\church{Q'}\circ\church{Q}$.

\end{theorem}
 
\begin{proof}

We follow the proof in \cite{GK}, as we have been throughout this section. We will form $Q''$ by constructing the following diagram, which we will explain step by step:
\begin{align}\label{dia:query comp}
\xymatrix{
&&&N\ar[rr]^p\ar[dl]_n\ar@{}[dr]|{(iv)}&&D'\ar[r]^q\ar[dl]^e\ar@{}[ddr]|{(ii)}&M\ar[dd]^w\\
&&B'\ar[rr]^r\ar[dl]_m\ar@{}[dr]|{(iii)}&&A'\ar[rd]^k\ar[ld]_h\ar@{}[dd]|{(i)}&&&\\
&B\ar[dl]_s\ar[rr]^f&&A\ar[dr]_t&&D\ar[dl]^u\ar[r]^g&C\ar[dr]^v&\\
S&&&&T&&&U}
\end{align}
First, form (i) by taking the pullback; note that $k$ is a discrete op-fibration by Proposition \ref{prop:dop and pullback}. Second, form (ii) as a distributivity diagram (see Proposition \ref{prop:distributive}) and note that $w$ is a discrete op-fibration. Third, form (iii) and (iv) with $B'=(f\down h)$ and $N=(r\down e)$. 

By Proposition \ref{prop:pb beck sigma} we have $\Delta_u\Sigma_t\iso\Sigma_k\Delta_h$. By Proposition \ref{prop:distributive} we have $\Pi_g\Sigma_k\iso\Sigma_w\Pi_q\Delta_e$. By Corollary \ref{cor:BC Pi} we have both $\Delta_h\Pi_f\iso\Pi_r\Delta_m$ and $\Delta_e\Pi_r\iso\Pi_p\Delta_n$. Pulling this all together, we have an isomorphism
\begin{align*}
\Sigma_v\Pi_g\Delta_u\Sigma_t\Pi_f\Delta_s
\iso\Sigma_v\Pi_g\Sigma_k\Delta_h\Pi_f\Delta_s
&\iso\Sigma_v\Sigma_w\Pi_q\Delta_e\Delta_h\Pi_f\Delta_s\\
&\iso\Sigma_v\Sigma_w\Pi_q\Delta_e\Pi_r\Delta_m\Delta_s\\
&\iso\Sigma_v\Sigma_w\Pi_q\Pi_p\Delta_n\Delta_m\Delta_s,
\end{align*}
which proves that $\church{Q''}\iso\church{Q'}\circ\church{Q}$, where $Q''$ is the data migration functor given by the triple of morphisms $Q':=(s\circ m\circ n, q\circ p, v\circ w)$, i.e. $\church{Q}=\Sigma_{vw}\Pi_{qp}\Delta_{smn}.$ This completes the proof.

\end{proof}

\begin{corollary}\label{cor:query comp for Delta-restricted}

Suppose that $Q\taking S\queryto T$ and $Q'\taking T\queryto U$ are queries as in (\ref{dia:composable queries}) and that both are $(\Delta,\Sigma)$-restricted. Then there exists a $(\Delta,\Sigma)$-restricted data migration query $Q''\taking S\queryto U$ such that $\church{Q''}\iso\church{Q'}\circ\church{Q}$.

\end{corollary}

\begin{proof}

We form a diagram similar to (\ref{dia:query comp}), as follows. First, form (i) by taking the pullback; note that $k$ is a discrete op-fibration by Proposition \ref{prop:dop and pullback}. Second, form (ii) as a distributivity diagram (see Proposition \ref{prop:distributive}) and note that $w$ is a discrete op-fibration. Third, form (iii) and (iv) as pullbacks. Note that $h$, $m$, and $n$ will be discrete opfibrations. Following the proof of Theorem \ref{thm:query comp}, the result follows by Corollary \ref{cor:pb bc DP dop} in place of Corollary \ref{cor:BC Pi}. 

\end{proof}

\subsection{Syntactic characterization of query morphism and query equivalence}

We continue to closely follow \cite{GK}. 

\begin{definition}[Morphism of data migration queries]

A {\em morphism of $(\Delta,\Sigma)$-restricted data migration queries from $P=(u,g,v)$ to $Q=(s,f,t)$}, denoted $(c,a,\beta)\taking P\to Q$, consists of two functors $c,a$ and a natural transformation $\beta$, fitting into a diagram of the form
\begin{align}\label{dia:morphism of polys}
\xymatrix@=35pt{
P:&S\ar@{=}[dd]&D\ar@[blue][l]_u\ar[r]^g&C\ar@{=}[d]\ar@[blue][r]^v&T\ar@{=}[dd]\\
&&B\cross_AC\ar@{}[ur]|{\beta\Uparrow}\ar[u]^a\ar[r]^-{f'}\ar[d]_b\ullimit&C\ar@[blue][d]^c&\\
Q:&S&B\ar@[blue][l]^s\ar[r]_f&A\ar@[blue][r]_t&T
}
\end{align}
where blue arrows $(u,v,s,t,c)$ are required to be discrete op-fibrations and every square commutes except for the top middle square in which $\beta\taking f'\to ga$ is a natural transformation. 

\end{definition}

Note that for any diagram of the form (\ref{dia:morphism of polys}), the maps $a$ and $b$ will automatically be discrete op-fibrations too.

\begin{remark}

Note that a morphism of queries $P\to Q$ involves a ``backward" morphisms of schemas, $a\taking B\cross_AC\to D$ as in (\ref{dia:morphism of polys}). In fact, this reversal happens in the conjunctive fragment only (see Lemma \ref{lemma:crucial}). This may not be surprising because $a$ is a morphism of indexing categories, or ``bound variables", and thus the directionality is in line with classical database theory. 

\end{remark}

\begin{definition}[Composition of data migration queries]

Suppose given data migration queries $P,Q,R\taking S\queryto T$ and morphisms $P\to Q$ and $Q\to R$ as in Definition \ref{dia:morphism of polys}. Then they can be composed to give a query morphism $P\to R$, as follows.

We begin the following diagram 
$$
\xymatrix{
S\ar@{=}[dd]&F\ar[l]_w\ar[r]^h\ar@{}[dr]|{\Uparrow}&E\ar[r]^x\ar@{=}[d]&T\ar@{=}[dd]\\
&D\cross_CE\ar[u]\ar[r]\ar[d]\ullimit&E\ar@[blue][d]^e\\
S\ar@{=}[dd]&D\ar@{}[dr]|{\Uparrow}\ar[r]^g\ar[l]_u&C\ar@{=}[d]\ar[r]^v&T\ar@{=}[dd]\\
&B\cross_AC\ar[r]^-{f'}\ar[d]_b\ar[u]^a\ullimit&C\ar@[blue][d]^c\\
S&B\ar[r]_f\ar[l]^s&A\ar[r]_t&T
}
$$
By Proposition \ref{prop:dmf under nt} there is a canonical map $f'\cross_ce\to (ga)\cross_Ce$, and this allows us to construct a composite 
$$
\xymatrix{
S\ar@{=}[dd]&F\ar[l]_w\ar[r]^h\ar@{}[dr]|{\Uparrow}&E\ar[r]^x\ar@{=}[d]&T\ar@{=}[dd]\\
&B\cross_AE\ar[r]\ar[d]\ar[u]\ullimit&E\ar@[blue][d]^{ce}\\
S&B\ar[r]_f\ar[l]^s&A\ar[r]_t&T
}
$$

\end{definition}

\begin{remark}

Note that this composition is associative {\em up to isomorphism}, but not on the nose. Thus we cannot speak of the category of data migration queries $P\queryto Q$, but only the $(\infty,1)$-category of such.

\end{remark}

The following definition is slightly abbreviated, i.e. not fully spelled out, but hopefully it is clear to anyone who is following so far.

\begin{definition}[Category of data migration queries]\label{def:cat of dmq}

We define the {\em category of $(\Delta,\Sigma)$-restricted data migration queries from $S$ to $T$}, denoted $\RQry(S,T)$, to be the category whose objects are $(\Delta,\Sigma)$-restricted data migration queries $S\queryto T$ and whose morphisms are equivalence classes of diagrams as in (\ref{dia:morphism of polys}), where two such diagrams are equivalent if there are equivalences of categories between corresponding objects in their middle rows and the functors commute appropriately.

\end{definition}

\begin{lemma}

Given a morphism of $(\Delta,\Sigma)$-restricted data migration queries $P\to Q$, there is an induced natural transformation $\church{P}\to\church{Q}$.

\end{lemma}

\begin{proof}

By Proposition \ref{prop:dmf under nt} and Corollary \ref{cor:pb bc DP dop} we have the natural natural transformations below:

\begin{align}\label{dia:poly morphism to nt}
\Sigma_v\Pi_g\Delta_u=\Sigma_t\Sigma_c\Pi_g\Delta_u\To{\eta_a}\;\;&\Sigma_t\Sigma_c\Pi_g\Pi_a\Delta_a\Delta_u\\
\nonumber\To{\;\beta\;}\;\;&\Sigma_t\Sigma_c\Pi_{f'}\Delta_a\Delta_u\\
\nonumber=\;\;\;&\Sigma_t\Sigma_c\Pi_{f'}\Delta_b\Delta_s\\
\nonumber\iso\;\;\;&\Sigma_t\Sigma_c\Delta_c\Pi_f\Delta_s\To{e_c}\Sigma_t\Pi_f\Delta_s
\end{align}

\end{proof}

\begin{lemma}\label{lemma:cartesian over poly is poly}

Let $Q\taking S\queryto T$ be a $(\Delta,\Sigma)$-restricted data migration query, and let $P\taking S\inst\to T\inst$ be any functor. If there exists a Cartesian natural transformation $\phi\taking P\to\church{Q}$, then there is an essentially unique $(\Delta,\Sigma)$-restricted data migration query isomorphic to $P$.

\end{lemma}

\begin{proof}

Suppose that $Q$ is represented by the bottom row in the diagram below
\begin{align}\label{dia:cartesian map of polys}
\xymatrix{Q':&S\ar@{=}[d]&B\cross_AC\ar[l]_{s'}\ar[r]^{f'}\ar[d]_b\ullimit&C\ar[r]^{t'}\ar[d]_c&T\ar@{=}[d]\\
Q:&S&B\ar[l]^s\ar[r]_f&A\ar[r]_t&T}
\end{align}
We form the rest of the diagram as follows. Let $\ast$ be the terminal object in $S\inst$, so in particular $\int_TQ(\ast)\iso t$. Let $t'=\int_TP(\ast)$ and define $c\taking t'\to t$ over $T$ to be $\int_T\phi(\ast)$. Let $b$ be the pullback of $c$ and $s'=s\circ b$. The Diagram (\ref{dia:cartesian map of polys}) now formed, let $Q'\taking S\queryto T$ be the top row. Note that $c$, and therefore $b$ and $s$, are discrete op-fibrations, so in particular $Q'$ is a $(\Delta,\Sigma)$-restricted data migration query. Note also that there was essentially no choice in the definition of $Q'$.

The map $\church{Q'}\to\church{Q}$ is given by 
$$
\Sigma_{t'}\Pi_{f'}\Delta_{s'}=\Sigma_{t}\Sigma_{c}\Pi_{f'}\Delta_b\Delta_s\To{\iso}\Sigma_t\Sigma_c\Delta_c\Pi_f\Delta_s\To{\eta_c}\Sigma_t\Pi_f\Delta_s,
$$
so it is Cartesian by Lemma \ref{lemma:cartesian}. But if both $P$ and $Q'$ are Cartesian over $Q$ and if they agree on the terminal object (as they do by construction), then they are isomorphic. Thus $P\iso Q'$ is a $(\Delta,\Sigma)$-restricted data migration query. 

\end{proof}

\begin{lemma}[Yoneda Functorialization]\label{lemma:yoneda}

Let $C$ be a category and let $b\taking C\to\Set$ be a functor with $s\taking B\to C$ its Grothendieck category of elements. Then there is a natural isomorphism of functors 
$$\Hom_{C\set}(b,-)\iso\Pi_t\Pi_s\Delta_s$$ 
where $t\taking I\to\vect{0}$ is the unique functor. Moreover, for any $b'\taking C\to\Set$ with $s'\taking B'\to C$ its category of elements, there is a bijection 
$$\Hom_{\Cat/C}(B,B')\iso\Hom_{C\set}(b,b')\iso\Hom(\Pi_t\Pi_{s'}\Delta_{s'},\Pi_t\Pi_s\Delta_s).$$

\end{lemma}

\begin{proof}

The second claim follows from the first by Proposition \ref{prop:on dopf} and the usual Yoneda-style argument. For the first, we have \begin{align*}
\Pi_u\Pi_s\Delta_s(-)\iso\Hom_\Set(\ast,\Pi_u\Pi_s\Delta_s(-))\iso\Hom_{B\set}(\ast^B,\Delta_s(-))&\iso\Hom_{C\set}(\Sigma_s(\ast^B),-)\\&\iso\Hom_{C\set}(b,-).
\end{align*}

\end{proof}

\begin{lemma}\label{lemma:crucial}

Suppose given the diagram to the left, where $u$ and $u'$ are discrete op-fibrations:
$$
\xymatrix@=30pt{
S\ar@{=}[d]&D\ar[l]_u\ar[r]^g&C\ar@{=}[d]\\
S&X\ar[l]^{u'}\ar[r]_{g'}&C
}
\hspace{1in}
\xymatrix@=30pt{
S\ar@{=}[d]&D\ar[l]_u\ar[r]^g\ar@{}[dr]|{\beta\Uparrow}&C\ar@{=}[d]\\
S&X\ar[u]^a\ar[l]^{u'}\ar[r]_{g'}&C
}
$$
and a natural transformation $j\taking\Pi_g\Delta_u\to\Pi_{g'}\Delta_{u'}$. Then there exists an essentially unique diagram as to the right, such that $j$ is the composition (see Proposition \ref{prop:dmf under nt}), 
\begin{align}\label{dia:composite for j}
\Pi_g\Delta_u\To{\eta_a}\Pi_g\Pi_a\Delta_a\Delta_u\To{\beta}\Pi_{g'}\Delta_a\Delta_u\iso\Pi_{g'}\Delta_{u'}.
\end{align}

\end{lemma}

\begin{proof}

Choose an object $c\taking\vect{0}\to C$ in $C$. Form the diagram 
$$
\xymatrix{
&D\ar[ld]_u\ar[dr]^g&&(c\down g)\ar[ll]_p\ar[rd]^q\\
S&&C&&\vect{0}\ar[ll]_c\\
&X\ar[ul]^{u'}\ar[ur]_{g'}&&(c\down g')\ar[ll]^{p'}\ar[ur]_{q'}}
$$
Let $s=u\circ p$ and $s'=u'\circ p'$, and let $t\taking S\to\vect{0}$ denote the unique functor. It is easy to show that $p$ and $p'$, and hence $s$ and $s'$ are discrete op-fibrations. We have a natural transformation 
\begin{align*}
\Pi_t\Pi_s\Delta_s=\Pi_q\Delta_s=\Pi_q\Delta_p\Delta_u\;\iso\;&\Delta_c\Pi_g\Delta_u\\
\To{j}&\Delta_c\Pi_{g'}\Delta_{u'}\iso\Pi_{q'}\Delta_{p'}\Delta_{u'}=\Pi_{q'}\Delta_{s'}=\Pi_t\Pi_{s'}\Delta_{s'}.
\end{align*}
By Lemma \ref{lemma:yoneda} this is equivalent to giving a map $j_c\taking(c\down g')\to (c\down g)$ over $S$. Moreover, given a map $f\taking c_0\to c_1$ in $C$, the natural transformation $\Delta_f\taking\Delta_{c_0}\to\Delta_{c_1}$ (see Proposition \ref{prop:dmf under nt}) induces a commutative diagram of functors $S\to\Set$, as to the left
$$
\xymatrix{\Delta_{c_1}\Pi_{g'}\Delta_{u'}&\Delta_{c_1}\Pi_g\Delta_u\ar[l]\\\Delta_{c_0}\Pi_{g'}\Delta_{u'}\ar[u]&\Delta_{c_0}\Pi_g\Delta_u\ar[u]\ar[l]}
\hspace{1in}
\xymatrix{(c_1\down g')\ar[r]^{j_{c_1}}\ar[d]_f&(c_1\down g)\ar[d]^f\\(c_0\down g')\ar[r]_{j_{c_0}}&(c_0\down g)}
$$
which induces a commutative diagram of categories over $S$ as to the right. 

We are now in a position to construct a unique functor $a\taking X\to D$ with $u'=u\circ a$ and natural transformation $\beta\taking g'\to g\circ a$ agreeing with each $j_c$. Given an object $x\in\Ob(X)$, we apply the map $j_{g'x}\taking (g'x\down g')\to(g'x\down g)$ to $(x,\id_{g'x})$ to get some $(d,g'x\To{b} gd)$. We assign $a(x):=d$ and $\beta_x:=b$, and note that $u'(x)=u(d)$ as above. A similar argument applies to morphisms, so $a$ and $\beta$ are constructed. To see that $j$ is the composite given in (\ref{dia:composite for j}), one checks it on objects and arrows in $C$ using the formula from Construction \ref{const:Pi} in conjunction with Proposition \ref{prop:limits are pushforwards}.

\end{proof}

\begin{theorem}\label{thm:syntactic query morphism}

Suppose that $P,Q\taking S\queryto T$ are $(\Delta,\Sigma)$-restricted data migration queries, and $X\taking\church{P}\to\church{Q}$ is a natural transformation of functors. Then there exists an essentially unique morphism of $(\Delta,\Sigma)$-restricted data migration queries $x\taking P\to Q$ (i.e. a diagram of the form (\ref{dia:morphism of polys})) such that $\church{x}\iso X$. In other words, the functor $\RQry(S,T)\to\Fun(S,T)$ is fully faithful.

\end{theorem}

\begin{proof}

Let $*$ denote the terminal object in $S\inst$. We define a functor $\ol{P}\taking S\inst\to T\inst$ as follows. For $I\taking S\to\Set$, define $\ol{P}(I)$ as the fiber product
$$
\xymatrix{\ol{P}(I)\ar[r]\ar[d]\ullimit&\church{Q}(I)\ar[d]\\\church{P}(*)\ar[r]_{X(\ast)}&\church{Q}(*).}
$$
There is a induced natural transformation $i\taking\church{P}\to\ol{P}$ and an induced Cartesian natural transformation $X'\taking \ol{P}\to\church{Q}$, such that $X=X'\circ i$. 

Let $v\taking C\to T$ and $t\taking A\to T$ denote the category of elements $v:=\int_T\church{P}(*)=\int_T\ol{P}(*)$ and $t:=\int_T\church{Q}(*)$ respectively, and let $c:=\int_TX(*)\taking v\to t$ over $T$. If $P$ and $Q$ are the top and bottom rows in the diagram below, then Lemma \ref{lemma:cartesian over poly is poly} and in particular Diagram (\ref{dia:cartesian map of polys}) implies that there is an essentially unique way to fill out the middle row, and its maps down to the bottom row, as follows:
$$\xymatrix{
P:&S\ar@{=}[d]&D\ar[l]_u\ar[r]^g&C\ar@{=}[d]\ar[r]^v&T\ar@{=}[d]\\
P':&S\ar@{=}[d]&B\cross_AC\ar[l]_{s'}\ar[r]^{f'}\ar[d]_b\ullimit&C\ar[d]^c\ar[r]^{v}&T\ar@{=}[d]\\
Q:&S&B\ar[l]^s\ar[r]_f&A\ar[r]_t&T
}
$$
with $s'=s\circ b$ and $\church{P'}=\ol{P}$. 

Lemma \ref{lemma:cartesian} implies that $i\taking\Sigma_v\Pi_g\Delta_u\to\Sigma_v\Pi_{f'}\Delta_{s'}$ induces a natural transformation $j\taking\Pi_g\Delta_u\to\Pi_{f'}\Delta_{s'}$ with $\Sigma_vj=i$. The result follows by Lemma \ref{lemma:crucial}.

\end{proof}

\section{Typed signatures}

\subsection{Definitions and data migration}

\begin{definition}

A {\em typed signature} $\ol{\mcC}$ is a sequence $\ol{\mcC}:=(\mcC,\mcC_0,i,\Gamma)$ where $\mcC$ is a signature, $\mcC_0$ is a discrete category,
\footnote{In fact, one does not need to assume that $\mcC_0$ is a discrete category for the following results to hold. Still, it is conceptually simpler. To get a feeling for what could be done if $\mcC_0$ were not discrete, consider the possibility of a table having two data columns, an attribute $A$ of type string and an attribute $B$ of type integer, where the integer in $B$ was the {\em length} of the string in $A$.}
$i\taking\mcC_0\to\mcC$ is a functor, and $\Gamma\taking\mcC_0\to\Set$ is a functor,
$$\xymatrix@=25pt{
\mcC_0\ar[rr]^i\ar[dr]_\Gamma&&\mcC\\
&\Set
}
$$
The signature $\mcC$ is called the {\em structure part} of $\ol{\mcC}$, and the rest is called the {\em typing setup}.

Suppose $(\mcC',\mcC_0,'i',\Gamma')$ is another typed signature. A {\em typed signature morphism from $\ol{\mcC}$ to $\ol{\mcC'}$}, denoted $\ol{F}=(F,F_0)\taking\ol{\mcC}\to\ol{\mcC'}$, consists of a functor $F\taking\mcC\to\mcC'$ and a functor $F_0\taking\mcC_0\to\mcC'_0$ such that the following diagram commutes:
$$\xymatrix@=25pt{
\mcC_0\ar[rr]^i\ar[dr]_\Gamma\ar[dd]_{F_0}&&\mcC\ar[dd]^F\\
&\Set\\
\mcC_0'\ar[ur]^{\Gamma'}\ar[rr]_{i'}&&\mcC'
}
$$
The category of typed signatures is denoted $\TSig$.

A {\em $\ol{\mcC}$-instance} $\ol{I}$ is a pair $\ol{I}:=(I,\delta)$ where $I\taking\mcC\to\Set$ is a functor, called the {\em structure part of $\ol{I}$} together with a natural transformation $\delta\taking I\circ i\to\Gamma$, called the {\em data part of $\ol{I}$}.
$$\xymatrix@=25pt{
\mcC_0\ar[rr]^i\ar[dr]_\Gamma&\ar@{}[d]|(.4){\stackrel{\delta}{\Leftarrow}}&\mcC\ar[dl]^I\\
&\Set
}
$$
Suppose $\ol{I'}:=(I',\delta')$ is another $\ol{\mcC}$-instance. An {\em instance morphism from $\ol{I}$ to $\ol{I'}$}, denoted $\alpha\taking\ol{I}\to\ol{I'}$, is a natural transformation $\alpha\taking I\to I'$ such that $\delta'\circ\alpha=\delta$. 
$$\xymatrix@=25pt{
\mcC_0\ar[rr]^i\ar[dr]_\Gamma&\ar@{}[d]|(.4){\stackrel{\delta'}{\Leftarrow}}&\mcC\ar@/^.5pc/[dl]^{I}\ar@{}[dl]|{\stackrel{\alpha}{\Leftarrow}}\ar@/_.5pc/[dl]_{I'}\\
&\Set
}
$$
The category of $\ol{\mcC}$-instances is denoted $\ol{\mcC}\inst$.

\end{definition}

\begin{proposition}

Let $\ol{\mcC}$ be a typed signature. Then the category $\ol{\mcC}\inst$ is a topos.

\end{proposition}

\begin{proof}

$\ol{\mcC}\inst$ is equivalent to the slice topos $\mcC\set/\Pi_i\Gamma$.

\end{proof}

\begin{construction}

Let $\ol{\mcC}:=(\mcC,\mcC_0,i,\Gamma)$ be a typed signature. We set up the tables for it as follows. For every arrow in $\mcC$ we make a binary table; we call these {\em arrow tables}. For every object $c\in\Ob(\mcC)$, let $N=i^\m1(c)\ss\mcC_0$. We make an $1+|N|$ column table, where $|N|$ is the cardinality of $N$; we call these {\em node tables}. For each node table, one column is the primary key column and the other $N$ columns are called {\em data columns}. The {\em data type} for each $n\in N$ is the set $\Gamma(n)$.

Let $(I,\delta)$ be an instance of $\ol{\mcC}$. For each object $c\in\Ob(\mcC)$ we fill in the primary key column of the node table with the set $I(c)\in\Ob(\Set)$. For each data column $n\in N$ we have a function $\delta_n\taking I(c)\to\Gamma(n)$, which we use to fill in the data in that column. For every arrow $f\taking c\to c'$ in $\mcC$ we fill in the primary key column of the arrow table with $I(c)$, and we fill in the other column with the function $I(f)\taking I(c)\to I(c')$.

\end{construction}

\begin{example}

If $\mcC_0=\emptyset$ then for any category presentation $\mcC$ there is a unique functor $i\taking\mcC_0\to\mcC$ and a unique functor $\Gamma\taking\mcC_0\to\Set$, so there is a unique schema $\ol{\mcC}=(\mcC,\emptyset,i,\Gamma)$ with structure $\mcC$ and empty domain setup. 

Given a functor $I\taking\mcC\to\Set$, there is a unique instance $\ol{I}:=(I,!)\in\ol{\mcC}\inst$ with that structure part. The data part of $\ol{I}$ is empty.

\end{example}

\begin{example}

Let $\mcC=\fbox{$\bullet^X$}$ be a terminal category, and let $\mcC_0=\{\tn{First},\tn{Last}\}$. There is a unique $i\taking\mcC_0\to\mcC$. Let 
$$\Gamma(\tn{First})=\Gamma(\tn{Last})=Strings.$$
Thus we have our schema $\ol{\mcC}$.

An instance on $\ol{\mcC}$ consists of a functor $I\taking\mcC\to\Set$ and a natural transformation $\delta\taking I\circ i\to\Gamma$. Let $I(X)=\{1,2\}$, let $\delta_{\tn{First}}(1)=\tt{David}$, let $\delta_{\tn{First}}(2)=\tt{Ryan}$, let $\delta_{\tn{Last}}(1)=\tt{Spivak}$, let $\delta_{\tn{Last}}(2)=\tt{Wisnesky}.$ We display this as 
$$\begin{tabular}{| l || l | l |}
\bhline
\multicolumn{3}{|c|}{X}\\\bhline
{\bf ID}&{\bf First}&{\bf Last}\\\bbhline
1&\tt{David}&\tt{Spivak}\\\hline
2&\tt{Ryan}&\tt{Wisnesky}\\
\bhline
\end{tabular}
$$

\end{example}

\begin{definition}\label{def:dedup}

Suppose given a typed signature $\ol{\mcC}$ and a typed instance $\ol{I}$,  
$$\xymatrix@=25pt{
\mcC_0\ar[rr]^i\ar[dr]_\Gamma&\ar@{}[d]|(.4){\stackrel{\delta}{\Leftarrow}}&\mcC\ar[dl]^I\\
&\Set
}
$$
We have a morphism $\delta\taking I\to\Pi_i\Gamma$. We say our instance $\ol{I}$ is {\em relational} if $\delta$ is a monomorphism. We define the {\em relationalization of $\ol{I}$} denoted $REL(\ol{I})$ to be the image $\tn{im}(\delta)\ss\Pi_i\Gamma$.

\end{definition}

\begin{proposition}\label{prop:typed Delta}

Let  $\ol{\mcC}=(\mcC,\mcC_0,i,\Gamma)$ and $\ol{\mcC'}=(\mcC',\mcC_0',i',\Gamma')$ be typed signatures, and let $\ol{F}=(F,F_0)\taking\ol{\mcC}\to\ol{\mcC'}$ be a typed signature morphism. Then the pullback functor $\Delta_F\taking\mcC'\inst\to\mcC\inst$ of untyped instances extends to a functor of typed instances 
$$\Delta_{\ol{F}}\taking\ol{\mcC'}\inst\to\ol{\mcC}\inst.$$

\end{proposition}

\begin{proof}

An object $(I',\delta')\in\Ob(\ol{\mcC'}\inst)$ is drawn to the left; simply compose with $F$ to get the diagram on the right
$$
\xymatrix@=25pt{
\mcC_0\ar[rr]^i\ar[dr]_\Gamma\ar[dd]_{F_0}&&\mcC\ar[dd]^F\\
&\Set\\
\mcC_0'\ar[ur]^{\Gamma'}\ar[rr]_{i'}&\ar@{}[u]|{\stackrel{\delta'}{\Leftarrow}}&\mcC'\ar[ul]_{I'}
}
\hsp
\xymatrix@=25pt{
\mcC_0\ar[rr]^i\ar[dr]_\Gamma\ar[dd]_{F_0}&&\mcC\ar[dd]^F\ar[dl]^{I'\circ F}\\
&\Set\\
\mcC_0'\ar[ur]^{\Gamma'}\ar[rr]_{i'}&\ar@{}[u]|(.4){\stackrel{\delta'}{\Leftarrow}}&\mcC'\ar[ul]_{I'}
}
$$
More formally, by Lemma \ref{lemma:comparison morphism} we have a natural transformation $\Delta_F\Pi_{i'}\to\Pi_i\Delta_{F_0}$, so we set $I=\Delta_FI'$ and we set $\delta'$ to be the composite
$$
\Delta_iI=\Delta_i\Delta_FI'=\Delta_{F_0}\Delta_{i'}I'\To{\Delta_{F_0}\delta'}\Delta_{F_0}\Gamma'=\Gamma
$$

\end{proof}

\begin{remark}\label{rem:interpretation of typed delta}

Suppose given the following diagram: 
\begin{align}\label{dia:typed delta helper}
\xymatrix@=25pt{
\mcC_0\ar[rr]^i\ar[dr]_\Gamma\ar[dd]_{F_0}&&\mcC\ar[dd]^F\\
&\Set\\
\mcC_0'\ar[ur]^{\Gamma'}\ar[rr]_{i'}&&\mcC'
}
\end{align}

A $\ol{\mcC'}$-instance is a functor $\int\Pi_{i'}\Gamma'\to\Set$, whereas a $\ol{\mcC}$ instance is a functor $\int\Pi_{i}\Gamma\to\Set$. We can reinterpret the typed-$\Delta$ functor using the following commutative diagram 
$$\xymatrix{
\int\Pi_i\Gamma\ar[dr]&\int\Delta_F\Pi_{i'}\Gamma'\ar[l]_q\ar[d]\ar[r]^r\ullimit&\int\Pi_{i'}\Gamma'\ar[d]\\
&\mcC\ar[r]_F&\mcC'
}
$$
The functor $\Delta_{\ol{F}}$ is given by the composition $\Sigma_q\Delta_r$.

Note that if (the exterior square of) Diagram \ref{dia:typed delta helper} is a pullback square, then $q$ is an isomorphism.

\end{remark}

\begin{definition}

Let  $\ol{\mcC}=(\mcC,\mcC_0,i,\Gamma)$ and $\ol{\mcC'}=(\mcC',\mcC_0,i',\Gamma)$ be typed signatures, and let $\ol{F}=(F,F_0)\taking\ol{\mcC}\to\ol{\mcC'}$ be a typed signature morphism. We say that $\ol{F}$ is {\em $\Pi$-ready} if $F_0=\id_{\mcC_0}$.

\end{definition}

\begin{proposition}\label{prop:typed pi}

Let  $\ol{\mcC}=(\mcC,\mcC_0,i,\Gamma)$ and $\ol{\mcC'}=(\mcC',\mcC_0,i',\Gamma)$ be typed signatures, and let $\ol{F}=(F,F_0)\taking\ol{\mcC}\to\ol{\mcC'}$ be a typed signature morphism that is $\Pi$-ready. Then the right pushforward functor $\Pi_F\taking\mcC\inst\to\mcC'\inst$ of untyped instances extends to a functor of typed instances 
$$\Pi_{\ol{F}}\taking\ol{\mcC}\inst\to\ol{\mcC'}\inst.$$

\end{proposition}

\begin{proof}

Given the diagram to the left, we form $\Pi_F(I)\taking\mcC'\to\Set$ to get the diagram on the right:
$$
\xymatrix@=25pt{
\mcC_0\ar[rr]^i\ar[dr]_\Gamma\ar@{=}[dd]&\ar@{}[d]|(.4){\stackrel{\delta}{\Leftarrow}}&\mcC\ar[dd]^F\ar[dl]^I\\
&\Set\\
\mcC_0\ar[ur]^{\Gamma}\ar[rr]_{i'}&&\mcC'
}
\hsp
\xymatrix@=25pt{
\mcC_0\ar[rr]^i\ar[dr]_\Gamma\ar@{=}[dd]&\ar@{}[d]|(.4){\stackrel{\delta}{\Leftarrow}}&\mcC\ar[dd]^F\ar[dl]^I\\
&\Set\ar@{}[r]|(.6){\epsilon\Uparrow}&\\
\mcC_0\ar[ur]^{\Gamma}\ar[rr]_{i'}&&\mcC'\ar[ul]_{\Pi_FI}
}
$$
The morphism $\epsilon\taking\Delta_F\Pi_FI\to I$ is the counit map. We set $\delta'\taking\Delta_{i'}\Pi_FI\to \Gamma$ to be the composite
$$\Delta_{i'}\Pi_FI=\Delta_i\Delta_F\Pi_FI\To{\epsilon}\Delta_iI\To{\delta}\Gamma.$$

\end{proof}

\begin{remark}\label{rem:interpretation of typed pi}

A $\ol{\mcC}$-instance is a functor $\int\Pi_i\Gamma\to\Set$, whereas a $\ol{\mcC'}$ instance is a functor $\int\Pi_{i'}\Gamma=\int\Pi_F\Pi_i\Gamma\to\Set$. We can reinterpret the typed-$\Pi$ functor using the following diagram (see Remark \ref{rem:interpretation of distributive law}). 
$$\xymatrix{
\int\Delta_F\Pi_F\Pi_i\Gamma\ullimit\ar[d]_e\ar[r]^g\ar[d]_e&\int\Pi_F\Pi_i\Gamma\ar[dd]^v\\
\int\Pi_i\Gamma\ar[d]_u&\\
\mcC\ar[r]_F&\mcC'
}
$$
The functor $\Pi_{\ol{F}}$ is given by 
$$\Pi_g\Delta_e\taking(\dispInt\Pi_i\Gamma)\set\to(\dispInt\Pi_F\Pi_i\Gamma)\set.$$

\end{remark}

\begin{definition}

Let  $\ol{\mcC}=(\mcC,\mcC_0,i,\Gamma)$ and $\ol{\mcC'}=(\mcC',\mcC_0,i',\Gamma)$ be typed signatures, and let $\ol{F}=(F,F_0)\taking\ol{\mcC}\to\ol{\mcC'}$ be a typed signature morphism. We say that $\ol{F}$ is {\em $\Sigma$-ready} if
\begin{itemize}
\item $p$ is a discrete opfibration, and 
\item $C_0=p^{\m1}(C_0')$, i.e. the following is a fiber product of categories,
$$\xymatrix{\mcC_0\ar[r]^i\ar[d]_{p_0}\ullimit&\mcC\ar[d]^p\\\mcC_0'\ar[r]_{i'}&\mcC'}$$.
\end{itemize}

\end{definition}

\begin{proposition}\label{prop:typed sigma}

Let  $\ol{\mcC}=(\mcC,\mcC_0,i,\Gamma)$ and $\ol{\mcC'}=(\mcC',\mcC_0',i',\Gamma')$ be typed signatures, and let $\ol{p}=(p,p_0)\taking\ol{\mcC}\to\ol{\mcC'}$ be a typed signature morphism that is $\Sigma$-ready. Then the left pushforward functor $\Sigma_p\taking\mcC\inst\to\mcC'\inst$ of untyped instances extends to a functor of typed instances 
$$\Sigma_{\ol{p}}\taking\ol{\mcC}\inst\to\ol{\mcC'}\inst.$$

\end{proposition}

\begin{proof}

An instance $\delta\taking\Delta_iI\to\Gamma$ is drawn below
$$
\xymatrix@=25pt{
\mcC_0\ar[rr]^i\ar[dr]_\Gamma\ar[dd]_{p_0}&\ar@{}[d]|(.4){\stackrel{\delta}{\Leftarrow}}&\mcC\ar[dd]^p\ar[dl]^I\\
&\Set\\
\mcC_0'\ar[ur]^{\Gamma'}\ar[rr]_{i'}&&\mcC'
}
$$
Let $I'=\Sigma_pI$; we need a morphism $\Delta_{i'}I'\to\Gamma'$. Since $\Gamma\iso\Delta_{p_0}\Gamma'$, we indeed have by Proposition \ref{prop:pb beck sigma}
$$
\Delta_{i'}\Sigma_pI\To{\iso}\Sigma_{p_0}\Delta_iI\To{\Sigma_{p_0}\delta}\Sigma_{p_0}\Delta_{p_0}\Gamma'\To{\epsilon}\Gamma'.
$$
\end{proof}

\begin{remark}\label{rem:interpretation of typed sigma}

A $\ol{\mcC}$-instance is a functor $\int\Pi_i\Gamma\to\Set$, whereas a $\ol{\mcC'}$ instance is a functor $\int\Pi_{i'}\Gamma'\to\Set$. We can reinterpret the typed-$\Sigma$ functor using the following diagram.
\begin{align}\label{dia:typed sigma helper}
\xymatrix{
\int\Pi_i\Gamma\ullimit\ar[d]\ar[r]^{p'}&\int\Pi_{i'}\Gamma'\ar[d]\\
\mcC\ar[r]_p&\mcC'
}
\end{align}
The above diagram is a fiber product diagram because of the isomorphism (see Corollary \ref{cor:pb bc DP dop}):
$$
\Pi_i\Gamma=\Pi_i\Delta_{p_0}\Gamma'\iso\Delta_p\Pi_{i'}\Gamma'.
$$
At this point we can interpret our typed pushforward using the natural isomorphism
$$\Sigma_{\ol{p}}\iso\Sigma_{p'}.$$

\end{remark}

\begin{definition}\label{def:typed FQL query}

Let $\ol{C}=(S,S_0,i_S,\Gamma_S)$ and $\ol{T}=(T,T_0,i_T,\Gamma_T)$ be typed signatures. A {\em typed FQL query} $Q$ from $\ol{S}$ to $\ol{T}$, denoted $Q\taking\ol{S}\queryto\ol{T}$ is a triple of typed signature morphisms $(\ol{F},\ol{G},\ol{H})$:
$$
\ol{S}\From{\ol{F}}\ol{S'}\To{\ol{G}}\ol{S''}\To{\ol{H}}\ol{T}
$$
such that $\ol{G}$ is $\Pi$-ready and $\ol{H}$ is $\Sigma$-ready.

\end{definition}

\subsection{Typed query composition}
\begin{proposition}[Comparison morphism for typed $\Delta,\Pi$]\label{prop:typed comparison Delta Pi}

Suppose given the following diagram of typed signatures, 
\begin{align}\label{dia:compare delta pi}
\xymatrix{
D_0\ar@{..>}[ddrrr]^(.8){\Gamma_D}\ar[rrrr]^{F_0}\ar[rd]^{i_Y}\ar@{=}[ddd]&&&&C_0\ar@{..>}[ddl]^(.7){\Gamma_C}\ar@{=}'[d][ddd]\ar[dr]^{i_X}\\
&Y\ar[rrrr]^(.4){G}\ar[ddd]_{q}&&&&X\ar[ddd]^{p}\\
&&&\Set\\
D_0\ar@{..>}[urrr]^(.7){\Gamma_D}\ar'[r][rrrr]^{F_0}\ar[dr]_{i_D}&&&&C_0\ar@{..>}[ul]^(.6){\Gamma_C}\ar[dr]^{i_C}\\
&D\ar[rrrr]^{F}&&&&C
}
\end{align}
such that all diagrams commute, except that the front square for which we have a natural transformation $\alpha\taking Fq\to pG$:
$$
\xymatrix{Y\ar[r]^-G\ar[d]_q\ar@{}[dr]|{\stackrel{\alpha}{\Nearrow}}&X\ar[d]^p\\D\ar[r]_F&C
}$$
Then the comparison transformation for untyped instances $\Delta_{F}\Pi_{p}\to\Pi_{q}\Delta_{G}$, from Lemma \ref{lemma:comparison morphism}, extends to a {\em typed comparison morphism} of typed queries $\ol{X}\inst\to\ol{D}\inst$,
$$\Delta_{\ol{F}}\Pi_{\ol{p}}\to\Pi_{\ol{q}}\Delta_{\ol{G}}.$$

\end{proposition}

\begin{proof}

Suppose given an $\ol{X}$-instance $(I,\delta\taking\Delta_{i_X}I\to\Gamma_C$. The formulas in Propositions \ref{prop:typed Delta} and \ref{prop:typed pi} become
\begin{align}\label{dia:delta-pi 1}
\Delta_{i_D}\Delta_F\Pi_pI=\Delta_{F_0}\Delta_{i_C}\Pi_pI
=\Delta_{F_0}\Delta_{i_X}\Delta_p\Pi_pI
\To{\epsilon_p}\Delta_{F_0}\Delta_{i_X}I
\To{\delta}\Delta_{F_0}\Gamma_C=\Gamma_D\\\label{dia:delta-pi 2}
\Delta_{i_D}\Pi_q\Delta_GI=\Delta_{i_Y}\Delta_q\Pi_q\Delta_GI
\To{\epsilon_q}\Delta_{i_Y}\Delta_GI
=\Delta_{F_0}\Delta_{i_X}I
\To{\delta}\Delta_{F_0}\Gamma_C=\Gamma_D
\end{align}

We need to show that the comparison morphism above commutes with these maps to $\Gamma_D$. To see this, consider the following commutative diagram:
$$\xymatrix{
\color{blue}{\Delta_{i_D}\Delta_F\Pi_p}\ar@{=}[d]\\
\Delta_{i_Y}\Delta_q\Delta_F\Pi_p\ar[d]_{\alpha}\ar[r]^{\eta_q}&\Delta_{i_Y}\Delta_q\Pi_q\Delta_q\Delta_F\Pi_p\ar[d]^\alpha\\
\Delta_{i_Y}\Delta_G\Delta_p\Pi_p\ar@{=}[d]&\Delta_{i_Y}\Delta_q\Pi_q\Delta_G\Delta_p\Pi_p\ar[l]_{\epsilon_q}\ar[d]^{\epsilon_p}\\
\Delta_{F_0}\Delta_{i_X}\Delta_p\Pi_p\ar[d]_{\epsilon_p}&\Delta_{i_Y}\Delta_q\Pi_q\Delta_G\ar@{=}[r]\ar[d]^{\epsilon_q}&\color{ForestGreen}{\Delta_{i_D}\Pi_q\Delta_G}\\
\color{red}{\Delta_{F_0}\Delta_{i_X}}\ar@{=}[r]&\Delta_{i_Y}\Delta_G
}
$$
The left-hand composite $\color{blue}{\Delta_{i_D}\Delta_F\Pi_p}\to\color{red}{\Delta_{F_0}\Delta_{i_X}}$ is the first part of (\ref{dia:delta-pi 1}), the lower-right zig-zag $\color{ForestGreen}{\Delta_{i_D}\Pi_q\Delta_G}\to\color{red}{\Delta_{F_0}\Delta_{i_X}}$ is the first part of (\ref{dia:delta-pi 2}), and the map from top to lower-right $\color{blue}{\Delta_{i_D}\Delta_F\Pi_p}\to\color{ForestGreen}{\Delta_{i_D}\Pi_q\Delta_G}$ is $\Delta_{i_D}$ applied to the comparison morphism. The commutativity of the diagram (which follows by the triangle identities and the associativity law) is what we were trying to prove.

\end{proof}

\begin{corollary}[Pullback Beck-Chevalley for typed $\Delta,\Pi$]\label{cor:pullback BC for typed Delta Pi}

Suppose given Diagram (\ref{dia:compare delta pi}) such that the front square is a pullback:
$$
\xymatrix{Y\ar[r]^{G}\ar[d]_{q}\ullimit&X\ar[d]^p\\D\ar[r]_{F}&C
}$$
and $p$ is a discrete op-fibration. Then the typed comparison morphism is an isomorphism: 
$$\Delta_{\ol{F}}\Pi_{\ol{p}}\To{\iso}\Pi_{\ol{q}}\Delta_{\ol{G}}$$

\end{corollary}

\begin{proof}

By Corollary \ref{cor:pb bc DP dop}, the comparison morphism is an isomorphism, 
$$\Delta_F\Pi_p\To{\iso}\Pi_q\Delta_G.$$ 
The  result follows from Proposition \ref{prop:typed comparison Delta Pi}, where $i_Y\taking D_0\to Y$ is induced by the fact that $Y$ is a fiber product.

\end{proof}

\begin{corollary}[Comma Beck-Chevalley for typed $\Delta, \Pi$]\label{cor:comma BC for typed Delta Pi}

Suppose given Diagram (\ref{dia:compare delta pi}) such that the front square is a comma category:
$$
\xymatrix{Y=(F\down p)\ar[r]^-G\ar[d]_q\ar@{}[dr]|{\stackrel{\alpha}{\Nearrow}}&X\ar[d]^p\\D\ar[r]_F&C}
$$
Then the typed comparison morphism is an isomorphism: 
$$\Delta_{\ol{F}}\Pi_{\ol{p}}\To{\iso}\Pi_{\ol{q}}\Delta_{\ol{G}}$$

\end{corollary}

\begin{proof}

By Corollary \ref{cor:BC Pi}, the comparison morphism is an isomorphism, 
$$\Delta_F\Pi_p\To{\iso}\Pi_q\Delta_G.$$ 
The result follows from Proposition \ref{prop:typed comparison Delta Pi}, where $i_Y\taking D_0\to Y$ is the induced map factoring through the canonical morphism $F\times_Cp\to (F\down p)$.

\end{proof}

%\def\q{{\color{ForestGreen}q\color{black}}}
%\def\qq{{\color{ForestGreen}q_0\color{black}}}
%\def\G{{\color{red}G\color{black}}}
%\def\GG{{\color{red}G_0\color{black}}}
%\def\F{{\color{orange}F\color{black}}}
%\def\FF{{\color{orange}F_0\color{black}}}
%\def\p{{\color{blue}p\color{black}}}
%\def\pp{{\color{blue}p_0\color{black}}}

%$$\xymatrix{
%Y_0\ar@{..>}[ddrrr]^(.8){\Gamma_Y}\ar@[red][rrrr]^{\GG}\ar[rd]^{i_Y}\ar@[ForestGreen][ddd]_{\qq}&&&&X_0\ar@{..>}[ddl]^(.7){\Gamma_X}\ar@[blue]'[d][ddd]^{\pp}\ar[dr]^{i_X}\\
%&Y\ar@[red][rrrr]^(.4){\G}\ar@[ForestGreen][ddd]_{\q}&&&&X\ar@[blue][ddd]^{\p}\\
%&&&\Set\\
%D_0\ar@{..>}[urrr]^(.7){\Gamma_D}\ar@[orange]'[r][rrrr]^{\FF}\ar[dr]_{i_D}&&&&C_0\ar@{..>}[ul]^(.6){\Gamma_C}\ar[dr]^{i_C}\\
%&D\ar@[orange][rrrr]^{\F}&&&&C
%}
%$$

\begin{proposition}[Pullback Beck-Chevalley for typed $\Sigma,\Delta$]\label{prop:pullback BC for typed Delta Sigma}

Suppose given the following commutative diagram of typed signatures, 

$$
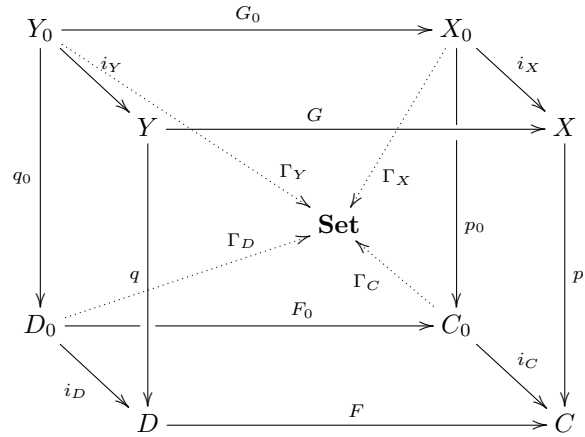
\xymatrix{
Y_0\ar@{..>}[ddrrr]^(.8){\Gamma_Y}\ar[rrrr]^{G_0}\ar[rd]^{i_Y}\ar[ddd]_{q_0}&&&&X_0\ar@{..>}[ddl]^(.7){\Gamma_X}\ar'[d][ddd]^{p_0}\ar[dr]^{i_X}\\
&Y\ar[rrrr]^(.4){G}\ar[ddd]_{q}&&&&X\ar[ddd]^{p}\\
&&&\Set\\
D_0\ar@{..>}[urrr]^(.7){\Gamma_D}\ar'[r][rrrr]^{F_0}\ar[dr]_{i_D}&&&&C_0\ar@{..>}[ul]^(.6){\Gamma_C}\ar[dr]^{i_C}\\
&D\ar[rrrr]^{F}&&&&C
}
$$
in which $p\taking X\to C$ is a discrete op-fibration, and the four side squares
$$
\xymatrix{Y_0\ar[r]^{i_Y}\ar[d]_{q_0}\ullimit&Y\ar[d]^q\\D_0\ar[r]_{i_D}&D}\hsp
\xymatrix{Y_0\ar[r]^{G_0}\ar[d]_{q_0}\ullimit&X_0\ar[d]^{p_0}\\D_0\ar[r]_{F_0}&C_0}\hsp
\xymatrix{Y\ar[r]^{G}\ar[d]_{q}\ullimit&X\ar[d]^p\\D\ar[r]_{F}&C}\hsp
\xymatrix{X_0\ar[r]^{i_X}\ar[d]_{p_0}\ullimit&X\ar[d]^p\\C_0\ar[r]_{i_C}&C}
$$
are pullbacks. Then there is an isomorphism 
$$\Sigma_{\ol{q}}\Delta_{\ol{G}}\To{\iso}\Delta_{\ol{F}}\Sigma_{\ol{p}}$$
of functors $\ol{X}\inst\to\ol{D}\inst$.

\end{proposition}

\begin{proof}

Note that the functors $q,p_0$, and $q_0$ are discrete opfibrations too. Suppose given an $\ol{X}$-instance $(I,\delta\taking\Delta_{i_X}I\to\Gamma_X$. By Proposition \ref{prop:pb beck sigma}, we have a Beck-Chevalley isomorphism $\Sigma_{q}\Delta_{G}\To{\iso}\Delta_{F}\Sigma_{p}$. The formulas in Propositions \ref{prop:typed Delta} and \ref{prop:typed sigma} become
\begin{align}\label{dia:sigma delta thing1}
\Delta_{i_D}\Sigma_q\Delta_GI\To{\iso}\Sigma_{q_0}\Delta_{i_Y}\Delta_GI
&=\Sigma_{q_0}\Delta_{G_0}\Delta_{i_X}I\\\nonumber
&\To{\delta}\Sigma_{q_0}\Delta_{G_0}\Gamma_X=\Sigma_{q_0}\Gamma_Y\\\nonumber
&=\Sigma_{q_0}\Delta_{q_0}\Gamma_D\To{\epsilon_{q_0}}\Gamma_D\\\label{dia:sigma delta thing2}
\Delta_{i_D}\Delta_F\Sigma_pI=\Delta_{F_0}\Delta_{i_C}\Sigma_pI
&\To{\iso}\Delta_{F_0}\Sigma_{p_0}\Delta_{i_X}I\\\nonumber
&\To{\delta}\Delta_{F_0}\Sigma_{p_0}\Gamma_X=\Delta_{F_0}\Sigma_{p_0}\Delta_{p_0}\Gamma_C\\\nonumber
&\To{\epsilon_{p_0}}\Delta_{F_0}\Gamma_C=\Gamma_D
\end{align}

We need to show that the Beck-Chevalley isomorphism commutes with these maps to $\Gamma_D$. It suffices to show that that the following diagram commutes:
$$\xymatrix@=15pt{
\color{red}{\Sigma_{q_0}\Delta_{G_0}\Delta_{i_X}}\ar[r]^-{\eta_{p_0}}\ar@{=}[d]&
\Sigma_{q_0}\Delta_{G_0}\Delta_{p_0}\Sigma_{p_0}\Delta_{i_X}\ar@{=}[r]&
\Sigma_{q_0}\Delta_{q_0}\Delta_{F_0}\Sigma_{p_0}\Delta_{i_X}\ar[r]^{\epsilon_{q_0}}&
\color{blue}{\Delta_{F_0}\Sigma_{p_0}\Delta_{i_X}}\ar[d]^{\eta_p}\\
\Sigma_{q_0}\Delta_{i_Y}\Delta_G\ar[d]_{\eta_q}&&&\Delta_{F_0}\Sigma_{p_0}\Delta_{i_X}\Delta_p\Sigma_p\ar@{=}[d]\\
\Sigma_{q_0}\Delta_{i_Y}\Delta_q\Sigma_q\Delta_G\ar@{=}[d]&&&\Delta_{F_0}\Sigma_{p_0}\Delta_{p_0}\Delta_{i_C}\Sigma_p\ar[d]^-{\epsilon_{p_0}}\\
\Sigma_{q_0}\Delta_{q_0}\Delta_{i_D}\Sigma_q\Delta_G\ar[d]_{\epsilon_{q_0}}&&&\Delta_{F_0}\Delta_{i_C}\Sigma_p\ar@{=}[d]\\
\color{ForestGreen}{\Delta_{i_D}\Sigma_q\Delta_G}\ar[r]_-{\eta_p}&
\Delta_{i_D}\Sigma_q\Delta_G\Delta_p\Sigma_p\ar@{=}[r]&
\Delta_{i_D}\Sigma_q\Delta_q\Delta_F\Sigma_p\ar[r]_-{\epsilon_q}&
\color{orange}{\Delta_{i_D}\Delta_F\Sigma_p}
}
$$
Indeed, the left-hand composite is the inverse to (\ref{dia:sigma delta thing1}), the right-hand composite is the inverse to (\ref{dia:sigma delta thing2}), the bottom map is the Beck-Chevalley isomorphism, and the top map relates $\Sigma_{q_0}\Delta_{G_0}$ to $\Delta_{F_0}\Sigma_{p_0}$, which completes the comparison between the two maps to $\Gamma_D$ above.

Proving that the above diagram commutes may be much easier than what follows, which is quite unenlightening. It is mainly an exercise in finding something analogous to a least common denominator in the square above. We include it for completeness. 
$$\footnotesize\xymatrix@=15pt{
\color{blue}{\Delta_{F_0}\Sigma_{p_0}\Delta_{i_X}}\ar[r]^{\eta_p}&\Delta_{F_0}\Sigma_{p_0}\Delta_{i_X}\Delta_p\Sigma_p\ar@{=}[r]&\Delta_{F_0}\Sigma_{p_0}\Delta_{p_0}\Delta_{i_C}\Sigma_p\ar@/^1pc/[ddr]^{\epsilon_{p_0}}\\
\Sigma_{q_0}\Delta_{q_0}\Delta_{F_0}\Sigma_{p_0}\Delta_{i_X}\ar[u]^{\epsilon_{q_0}}\ar[r]^{\eta_p}&\Sigma_{q_0}\Delta_{q_0}\Delta_{F_0}\Sigma_{p_0}\Delta_{i_X}\Delta_p\Sigma_p\ar[u]_{\epsilon_{q_0}}\ar@{=}[r]&\Sigma_{q_0}\Delta_{q_0}\Delta_{F_0}\Sigma_{p_0}\Delta_{p_0}\Delta_{i_C}\Sigma_p\ar[d]^{\epsilon_{p_0}}\ar[u]_{\epsilon_{q_0}}\\
\color{red}{\Sigma_{q_0}\Delta_{G_0}\Delta_{i_X}}\ar[u]^{\eta_{p_0}}\ar[r]^{\eta_p}\ar@{=}[d]&\Sigma_{q_0}\Delta_{G_0}\Delta_{i_X}\Delta_p\Sigma_p\ar[u]_{\eta_{p_0}}\ar@{=}[d]\ar@{=}[r]&\Sigma_{q_0}\Delta_{q_0}\Delta_{F_0}\Delta_{i_C}\Sigma_p\ar@{=}[d]\ar[r]^{\epsilon_{q_0}}&\Delta_{F_0}\Delta_{i_C}\Sigma_p\ar@{=}[d]\\
\Sigma_{q_0}\Delta_{i_Y}\Delta_G\ar[d]_{\eta_q}\ar[r]^{\eta_p}&\Sigma_{q_0}\Delta_{i_Y}\Delta_G\Delta_p\Sigma_p\ar[d]^{\eta_q}\ar@{=}[r]&\Sigma_{q_0}\Delta_{q_0}\Delta_{i_D}\Delta_F\Sigma_p\ar[r]^{\epsilon_{q_0}}&\color{orange}{\Delta_{i_D}\Delta_F\Sigma_p}\\
\Sigma_{q_0}\Delta_{q_0}\Delta_{i_D}\Sigma_q\Delta_G\ar[d]_{\epsilon_{q_0}}\ar[r]^{\eta_p}&\Sigma_{q_0}\Delta_{q_0}\Delta_{i_D}\Sigma_q\Delta_G\Delta_p\Sigma_p\ar[d]^{\epsilon_{q_0}}\ar@{=}[r]&\Sigma_{q_0}\Delta_{q_0}\Delta_{i_D}\Sigma_q\Delta_q\Delta_F\Sigma_p\ar[u]_{\epsilon_q}\ar[d]^{\epsilon_{q_0}}\\
\color{ForestGreen}{\Delta_{i_D}\Sigma_q\Delta_G}\ar[r]^{\eta_p}&\Delta_{i_D}\Sigma_q\Delta_G\Delta_p\Sigma_p\ar@{=}[r]&\Delta_{i_D}\Sigma_q\Delta_q\Delta_F\Sigma_p\ar@/_1pc/[uur]_{\epsilon_q}
}
$$
Every square in the above diagram clearly commutes, completing the proof.

\end{proof}

\begin{proposition}[Typed distributivity]\label{prop:typed distributivity}

Suppose given the diagram of typed signatures
$$\xymatrix{
\Set&C_0\ar[l]_{\Gamma_C}\ar[r]^{i_C}\ar[d]_{u_0}\ullimit&C\ar[d]^u\\
&B_0\ar[ul]^{\Gamma_B}\ar[r]^{i_B}\ar@/_1pc/[rr]_{i_A}&B\ar[r]^f&A
}
$$
such that $u$ is a discrete op-fibration and the square is a pullback. Note that $\ol{u}=(u,u_0)$ is $\Sigma$-ready and that $\ol{f}=(f,\id_{B_0})$ is $\Pi$-ready. Construct the {\em typed distributivity diagram}
$$\xymatrix{
N_0\ar[r]^{i_N}\ar[d]_{e_0}\ullimit&N\ar[r]^g\ar[d]^e\ullimit&M\ar[dd]^v\\
C_0\ar[r]^{i_C}\ar[d]_{u_0}\ullimit&C\ar[d]^u&\\
B_0\ar[r]^{i_B}&B\ar[r]^f&A
}
$$
as follows. First form the distributivity diagram as in Proposition \ref{prop:distributive}, which is the big rectangle to the right. Now take $N_0$ to be the pullback of either the top-left square, the big-left rectangle, or the big outer square --- all are equivalent. The functors $v$ and $e$ are discrete opfibrations. Note that $\ol{g}=(g,\id_{N_0})$ is $\Pi$-ready and $\ol{v}=(v,u_0\circ e_0)$ is $\Sigma$-ready.

Then there is an isomorphism 
$$\Sigma_{\ol{v}}\Pi_{\ol{g}}\Delta_{\ol{e}}\To{\iso}\Pi_{\ol{f}}\Sigma_{\ol{u}}$$
of functors $\ol{C}\inst\to\ol{A}\inst$.

\end{proposition}

\begin{proof}

By Proposition \ref{prop:distributive} we have a distributivity isomorphism 
\begin{align}\label{dia:dist for typed}
\Sigma_v\Pi_g\Delta_e\To{\iso}\Pi_f\Sigma_u
\end{align}
of functors $C\inst\to A\inst$, but we want an isomorphism of functors $\ol{C}\inst\to\ol{A}\inst$. 

Suppose given a $\ol{C}$-instance $(I,\delta\taking\Delta_{i_C}I\to\Gamma_C)$. This is equivalent to a natural transformation $\delta\taking I\to\Pi_{i_C}\Gamma_C$ or equivalently a functor $\delta\taking\int\Pi_{i_C}\Gamma_C\to\Set$. The two sides of isomorphism (\ref{dia:dist for typed}) give rise to the same instance $A\to\Set$, but not a priori the same typed instance. A typed $A$-instance is a functor $\int\Pi_f\Pi_{i_B}\Gamma_B\to\Set$. It may be useful to find $\int\Pi_{i_C}\Gamma_C$ (middle vertex of left-hand square) and $\int\Pi_f\Pi_{i_B}\Gamma_B$ (bottom right back vertex) in the diagram below, which we will presently describe:
\begin{align}\label{dia:typed distributivity helper}
\xymatrix@=15pt{
~\\
&&\int\Delta_f\Pi_f\Pi_{i_C}\Gamma_C\ar[rrr]^{g'}\ar[dddl]_{e'}\ar[dl]\ar[dddd]&&&\int\Pi_f\Pi_{i_C}\Gamma_C\ar[ddddddll]^{v'}\ar[ddll]\ar[dddd]\\
&\int\Delta_e\Pi_{i_C}\Gamma_C\ar[dd]\ar[dl]_r\\
\int\Delta_f\Pi_f\partial u\ar[rrr]^(.6)g\ar[dd]_e&&&\int\Pi_f\partial u\ar[dddd]^(.3)v\\
&\int\Pi_{i_C}\Gamma_C\ar[dl]_q\ar[dddl]^{u'}\ar[dd]\\
C\ar[dd]_u&&\int\Delta_f\Pi_f\Pi_{i_B}\Gamma_B\ar[rrr]^(.55){g''}\ar[dl]_{e''}&&&\int\Pi_f\Pi_{i_B}\Gamma_B\ar[ddll]^{v''}\\
&\int\Pi_{i_B}\Gamma_B\ar[dl]^{u''}\\
B\ar[rrr]_f&&&A
}
\end{align}
The front square is the (untyped) distributivity diagram of Proposition \ref{prop:distributive}; note that it is a pullback. The bottom square is our interpretation of $\Pi_{\ol{f}}$ from Remark \ref{rem:interpretation of typed pi}; note that it is a pullback. In the left-hand square we begin by form the little square in the front portion (the part that includes $q,u,u', u''$) as a pullback; see Remark \ref{rem:interpretation of typed sigma}. Form the other little square as a pullback along $e$. Form the diagonal square as another distributivity diagram for $\Pi_f$ of $\partial u'$. Note that we can complete the right-hand square by the functoriality of $\Pi_f$ and that it is a pullback because $\Pi_f$ preserves limits. Finally, form the back square as a pullback. It is now easy to see that all six sides are pullbacks.

Our only task is to show that the top square is a distributivity square. It is already a pullback, so our interest is in the top right corner: we need to show that there is a natural isomorphism 
$$\Pi_f\Pi_{i_C}\Gamma_C=\Pi_f\partial u'\iso^?\Pi_g\partial r=\Pi_g\Delta_e\Pi_{i_C}\Gamma_C.$$
The adjunction isomorphism for $\Pi_g\partial r$ says that for any discrete opfibration $X\to\int\Pi_{f}\partial u$ there is an isomorphism
$$\Hom_{/\int\Delta_f\Pi_{f}\partial u}(g^\m1X,\dispInt\Delta_e\Pi_{i_C}\Gamma_C)\iso\Hom_{/\int\Pi_f\partial u}(X,\dispInt\Pi_g\Delta_e\Pi_{i_C}\Gamma_C).$$
We thus need to show that there is a one-to-one correspondence between functors $X\to\dispInt\Pi_f\Pi_{i_C}\Gamma_C)$ over $\int\Pi_f\partial u$ and dotted arrows in the diagram
$$\xymatrix{
&g^\m1X\ullimit\ar[r]\ar[d]\ar@{-->}[ld]&X\ar[d]\\
\int\Delta_e\Pi_{i_C}\Gamma_C\ar[r]_r&\int\Delta_f\Pi_f\partial u\ar[r]_g&\int\Pi_f\partial u
}
$$

Given a map $X\to\dispInt\Pi_f\Pi_{i_C}\Gamma_C)$, we can pull it back along $g$ to get a map $g^\m1X\to\int\Delta_f\Pi_f\Pi_{i_C}\Gamma_C$ because the top square in (\ref{dia:typed distributivity helper}) is a pullback; this induces the required dotted arrow. Conversely, given a dotted arrow, we obtain the diagram
$$
\xymatrix{
&g^\m1X\ullimit\ar[r]\ar[d]\ar[ld]&X\ar[d]\\
\int\Delta_e\Pi_{i_C}\Gamma_C\ar[r]_r\ar[d]&\int\Delta_f\Pi_f\partial u\ar[r]_g\ar[d]&\int\Pi_f\partial u\ar[d]\\
\int\Pi_{i_C}\Gamma_C\ar[r]_{u'}&B\ar[r]_f&A
}
$$
and because the diagonal diagram in (\ref{dia:typed distributivity helper}) was constructed as a distributivity diagram, this induces a map $X\to\int\Pi_f\Pi_{i_C}\Gamma_C$ as desired.

Now that we have completed the aforementioned task, we are ready to prove the result. Our goal is as follows. We begin with an instance in the middle of the left square, a functor $\delta\taking\int\Pi_{i_C}\Gamma_C\to\Set$. We are interested in $\Pi_{\ol{f}}\Sigma_{\ol{u}}\delta$ and $\Sigma_{\ol{v}}\Pi_{\ol{g}}\Delta_{\ol{e}}\delta$; we consider them in this order.

By Remark \ref{rem:interpretation of typed sigma}, we have a natural isomorphism
$$\Pi_{\ol{f}}\Sigma_{\ol{u}}\iso\Pi_{f}\Sigma_{u'}.$$
Because $e$ is a discrete opfibration, Remarks \ref{rem:interpretation of typed delta}, \ref{rem:interpretation of typed pi}, and \ref{rem:interpretation of typed sigma}, we have a natural isomorphism 
$$\Sigma_{\ol{v}}\Pi_{\ol{g}}\Delta_{\ol{e}}\iso\Sigma_{v'}\Pi_{g'}\Delta_{e'}.$$
The result now follows from the usual distributive law, Proposition \ref{prop:distributive}.

\end{proof}

\begin{theorem}
\label{xyzabc}
Suppose that one has typed FQL queries $Q\taking\ol{S}\queryto\ol{T}$ and $Q'\taking\ol{T}\queryto \ol{U}$ as follows:
\begin{align}
\xymatrix{&\ol{B}\ar[r]^f\ar[dl]_s&\ol{A}\ar[dr]^t&&&\ol{D}\ar[r]^g\ar[dl]_u&\ol{C}\ar[dr]^v\\
\ol{S}\ar@{}[rrr]|Q&&&\ol{T}&\ol{T}\ar@{}[rrr]|{Q'}&&&\ol{U}}
\end{align}
Then there exists a typed FQL query $Q''\taking\ol{S}\queryto\ol{U}$ such that $\church{Q''}\iso\church{Q'}\circ\church{Q}$.

\end{theorem}
 
\begin{proof}

We will form $Q''$ by constructing the following diagram, which we will explain step by step:

\begin{align}
\xymatrix{
&&&\ol{N}\ar[rr]^{\ol{p}}\ar[dl]_n\ar@{}[dr]|{(iv)}&&\ol{D'}\ar[r]^{\ol{q}}\ar[dl]^{\ol{e}}\ar@{}[ddr]|{(ii)}&\ol{M}\ar[dd]^{\ol{w}}\\
&&\ol{B'}\ar[rr]^{\ol{r}}\ar[dl]_{\ol{m}}\ar@{}[dr]|{(iii)}&&\ol{A'}\ar[rd]^{\ol{k}}\ar[ld]_{\ol{h}}\ar@{}[dd]|{(i)}&&&\\
&\ol{B}\ar[dl]_{\ol{s}}\ar[rr]^{\ol{f}}&&\ol{A}\ar[dr]_{\ol{t}}&&\ol{D}\ar[dl]^{\ol{u}}\ar[r]^{\ol{g}}&\ol{C}\ar[dr]^{\ol{v}}&\\
\ol{S}&&&&\ol{T}&&&\ol{U}}
\end{align}
We form the pullback $(i)$; note that since $\ol{t}$ was $\Sigma$-ready, so is $\ol{k}$. Now form the typed distributivity diagram $(ii)$ as in Proposition \ref{prop:typed distributivity}; note that $\ol{w}$ is $\Sigma$-ready and $\ol{q}$ is $\Pi$-ready. Now form the typed comma categories $(iii)$ and $(iv)$ as in Proposition \ref{prop:typed comparison Delta Pi}. Note that $\ol{r}$ and $\ol{p}$ are $\Pi$-ready. The result now follows from Proposition \ref{prop:pullback BC for typed Delta Sigma}, \ref{prop:typed distributivity}, and Corollary \ref{cor:comma BC for typed Delta Pi}, 

\end{proof}


\begin{thebibliography}{10}

\bibitem{DBLP:books/aw/AbiteboulHV95}
Serge Abiteboul, Richard Hull, and Victor Vianu.
\newblock {\em Foundations of Databases}.
\newblock Addison-Wesley, 1995.

\bibitem{Abiteboul:1989:OIQ:67544.66941}
Serge Abiteboul and Paris~C. Kanellakis.
\newblock Object identity as a query language primitive.
\newblock In {\em Proceedings of the 1989 ACM SIGMOD International Conference
  on Management of Data}, SIGMOD '89, pages 159--173, New York, NY, USA, 1989.
  ACM.

\bibitem{Bachmair89completionwithout}
Leo Bachmair, Nachum Dershowitz, and David~A. Plaisted.
\newblock Completion without failure, 1989.

\bibitem{BW}
Michael Barr and Charles Wells, editors.
\newblock {\em Category theory for computing science, 2nd ed.}
\newblock 1995.

\bibitem{Buneman92theoreticalaspects}
P.~Buneman, S.~Davidson, and A.~Kosky.
\newblock Theoretical aspects of schema merging.
\newblock In {\em EDBT}, 1992.

\bibitem{Carmody1995459}
S.~Carmody, M.~Leeming, and R.F.C. Walters.
\newblock The todd-coxeter procedure and left kan extensions.
\newblock {\em Journal of Symbolic Computation}, 19(5):459 -- 488, 1995.

\bibitem{Fagin:2005:CSM:1114244.1114249}
Ronald Fagin, Phokion~G. Kolaitis, Lucian Popa, and Wang-Chiew Tan.
\newblock Composing schema mappings: Second-order dependencies to the rescue.
\newblock {\em ACM Trans. Database Syst.}, 30(4):994--1055, December 2005.

\bibitem{Fleming02adatabase}
Michael Fleming, Ryan Gunther, and Robert Rosebrugh.
\newblock A database of categories.
\newblock {\em Journal of Symbolic Computing}, 35:127--135, 2002.

\bibitem{GK}
N.~Gambino and J.~Kock.
\newblock Polynomial functors and polynomial monads.

\bibitem{haas:clio}
Laura~M. Haas, Mauricio~A. Hern\'{a}ndez, Howard Ho, Lucian Popa, and Mary
  Roth.
\newblock Clio grows up: from research prototype to industrial tool.
\newblock In {\em SIGMOD '05}.

\bibitem{xyz}
Jieh Hsiang and Michael Rusinowitch.
\newblock On word problems in equational theories.
\newblock In {\em Automata, Languages and Programming}, volume 267 of {\em
  LNCS}. 1987.

\bibitem{opac-b1094856}
Bart Jacobs.
\newblock {\em Categorical logic and type theory}.
\newblock PhD thesis, Mathematics, Amsterdam, Lausanne, New York, 1999.

\bibitem{Johnson200251}
Michael Johnson and Robert Rosebrugh.
\newblock Sketch data models, relational schema and data specifications.
\newblock {\em Electronic Notes in Theoretical Computer Science}, 61(0):51 --
  63, 2002.

\bibitem{Johnstone:MR1953060}
Peter~T. Johnstone.
\newblock {\em Sketches of an elephant: a topos theory compendium. {V}ol. 1}.
\newblock 2002.

\bibitem{Spivak:1202.2591}
David~I Spivak.
\newblock Database queries and constraints via lifting problems.
\newblock 2012.

\bibitem{Spivak:2012:FDM:2324905.2325108}
David~I. Spivak.
\newblock Functorial data migration.
\newblock {\em Inf. Comput.}, 217:31--51, August 2012.

\bibitem{Wong:1994:QNC:921235}
Limsoon Wong.
\newblock {\em Querying nested collections}.
\newblock PhD thesis, Philadelphia, PA, USA, 1994.
\newblock AAI9503855.

\end{thebibliography}
\end{document}